\documentclass[12pt]{article} 
\usepackage{xr-hyper} 
 
\usepackage{amsfonts,amsmath,amssymb,amsthm}
\usepackage{appendix}
\usepackage{natbib}
\usepackage[hypertexnames=false]{hyperref}
\usepackage{setspace}
\usepackage{pdfsync}
\usepackage{lscape}
\usepackage{verbatim}
\usepackage[hmargin=1.25in,vmargin=1.25in]{geometry}
\usepackage{caption}
\usepackage{subcaption}
\usepackage{enumitem}

\usepackage{booktabs, multirow, rotating, threeparttable} 
\usepackage{xr}

\hypersetup{colorlinks, linkcolor = {red}, citecolor = {blue}, urlcolor ={purple}}

\renewcommand\paragraph{\@startsection{paragraph}{4}{\z@}%
            {-2.5ex\@plus -1ex \@minus -.25ex}%
            {1.25ex \@plus .25ex}%
            {\normalfont\normalsize\bfseries}}
\makeatother
\usepackage{mathrsfs,dsfont}
\usepackage[dvipsnames,svgnames, x11names]{xcolor}
\usepackage{pgfplots}
\usetikzlibrary{patterns}
\usepackage{caption}

\DeclareGraphicsExtensions{.jpg}

\newtheorem{proposition}{Proposition}[section]
\newtheorem{theorem}{Theorem}[section]

\newtheorem{lemma}{Lemma}[section]

\newtheorem{assumption}{Assumption}[section]

\newtheorem{remark}{Remark}[section]

\let\oldmarginpar\marginpar
\renewcommand{\marginpar}[2][rectangle,draw,fill=black, text=white,text width= 2cm,rounded corners]{
    \oldmarginpar{
    \tiny \tikz \node at (0,0) [#1]{#2};}
    }

\makeatletter

\newcommand*{\addFileDependency}[1]{
\typeout{(#1)}
\@addtofilelist{#1}
\IfFileExists{#1}{}{\typeout{No file #1.}}
}\makeatother

\begin{document}
\author{Amilcar Velez\\
Department of Economics \\ 
Cornell University \\ 
\url{amilcare@cornell.edu}
}

\bigskip

\title{The Local Projection Residual Bootstrap for AR(1) Models%
\footnote{I am deeply grateful to Ivan Canay, Federico Bugni, and Joel Horowitz for their guidance and support and for the extensive discussions that have helped shape the paper. 
I thank the editor, co-editor, and two anonymous referees for their comments and suggestions that have significantly helped to improve this paper. 
I am also thankful to Federico Crippa, Bruno Fava, Danil Fedchenko, Diego Huerta, Eleftheria Kelekidou, Pepe Montiel-Olea, Filip Obradovic, Mikkel Plagborg-Møller, Sebastian Sardon, and Ke-Li Xu for valuable comments and suggestions. Financial support from the Robert Eisner Memorial Fellowship and the Dissertation Year Fellowship is gratefully acknowledged.
}
}
\date{August 27, 2025}
\maketitle

\begin{spacing}{1.2}
\begin{abstract}
This paper proposes a local projection residual bootstrap method to construct confidence intervals for impulse response coefficients of AR(1) models. Our bootstrap method is based on the local projection (LP) approach and involves a residual bootstrap procedure applied to AR(1) models. We present theoretical results for our bootstrap method and proposed confidence intervals. First, we prove the uniform consistency of the LP-residual bootstrap over a large class of AR(1) models that allow for a unit root, conditional heteroskedasticity of unknown form, and martingale difference shocks. Then, we prove the asymptotic validity of our confidence intervals over the same class of AR(1) models. Finally, we show that the LP-residual bootstrap provides asymptotic refinements for confidence intervals on a restricted class of AR(1) models relative to those required for the uniform consistency of our bootstrap. 
\end{abstract}
\end{spacing}

\noindent KEYWORDS: Bootstrap, Local Projection, Uniform Inference, Asymptotic Refinements.

\thispagestyle{empty} 

\newpage
\setcounter{page}{1}
\onehalfspacing 

\section{Introduction}
%
This paper contributes to a growing literature on confidence interval construction for impulse response coefficients based on the local projection (LP) approach (\cite{Jorda2005}). In this literature, the LP approach estimates an impulse response coefficient as one of the slope coefficients in a linear regression of a future outcome on current or lag-augmented covariates (\cite{ramey2016macroeconomic}; \cite{NakamuraSteinsson2018}; \cite{MO-PM2020}). Recent theoretical results exist for the asymptotic validity of the confidence intervals constructed around the LP estimator, which hold over a large class of vector autoregressive models (\cite{xu2022}). Since these confidence intervals have small-sample coverage distortions (e.g., coverage probability is lower than expected), their bootstrap versions are recommended for practical use. However, theoretical results for these bootstrap versions are unknown, even for the AR(1) model. This paper proposes a different bootstrap method to construct LP confidence intervals with theoretical guarantees for a class of AR(1) models that allow for a unit root, conditional heteroskedasticity of unknown form, and martingale difference shocks.

We propose an LP-residual bootstrap method to construct confidence intervals for impulse response coefficients of AR(1) models. Our bootstrap method is based on the LP approach and involves a residual bootstrap procedure applied specifically to AR(1) models.\footnote{Section \ref{sec:lp-residual-bootstrap_VAR} presents the LP-residual bootstrap for VAR(p) models, but its theoretical properties are unknown and left for future research; see Remarks \ref{rem:var-stationary} and \ref{rem:var-stationar-asymp-refinement} for further discussion. Appendix \ref{appendix:simulations-var} reports a Monte Carlo simulation for the LP-residual bootstrap for vector autoregressive models.} Our bootstrap confidence intervals are centered at the LP estimator and use heteroskedasticity-consistent (HC) standard errors and a bootstrap critical value. Section \ref{sec:lp-residual-bootstrap} presents the details.

We rely on the asymptotic distribution theory initially developed in \cite{MO-PM2020} and generalized in \cite{xu2022}. 
In their framework, a root $R_n(h)$ based on the LP approach can be defined for a given horizon $h$ and a sample size $n$. 
Here, by a root we refer to a real-valued function depending on the data and an impulse response coefficient. Their results guarantee the root  $R_n(h)$ is asymptotically distributed as a standard normal distribution for a class of VAR models that allow for multiple unit roots and conditional heteroskedasticity of unknown form\textcolor{black}{, and even at intermediate horizons, i.e., horizons $h$ that are allowed to grow with $n$, e.g., $h = h_n \propto   n^{\zeta}$, $\zeta \in [0,1)$}. As a result, the root $R_n(h)$  can be used to construct a confidence interval $C_n(h, 1-\alpha)$ for an impulse response coefficient using a normal critical value (quantile of the asymptotic distribution).
Furthermore,  $C_n(h, 1-\alpha)$ has asymptotic coverage equal to the nominal level $1-\alpha$ uniformly over the parameter space (VAR model coefficients) and \textcolor{black}{a wide range of} intermediate horizons (e.g., \textcolor{black}{uniform over $h \le h_n$, where $h_n$ is any fixed sequence such that $h_n = o(n)$}). 
Nevertheless, Monte Carlo simulations report \textcolor{black}{that} $C_n(h, 1-\alpha)$ has a lower coverage probability than expected.  

We propose the LP-residual bootstrap method to approximate the distribution of the root $R_n(h)$ as an alternative to the asymptotic distribution. We use our approximation to calculate bootstrap-based critical values; see Section \ref{sec:bootstrap_cv} for the step-by-step procedure. Specifically, we construct a confidence interval $C_n^*(h,1-\alpha)$ for an impulse response coefficient using the root $R_n(h)$ and a bootstrap critical value; see Section \ref{sec:lp-residual-bootstrap} for details.   

Our first result proves the uniform consistency of the LP-residual bootstrap. More concretely, we demonstrate in Section \ref{sec:uniform_cons} that the distribution of the root $R_n(h)$ can be approximated by its bootstrap version uniformly over the parameter space (e.g., $\rho \in [-1,1]$) and \textcolor{black}{a wide range of} intermediate horizons (e.g., \textcolor{black}{uniform over $h \le h_n$, where $h_n$ is any fixed sequence such that $h_n = o(n)$}). 
Our result applies to a large class of AR(1) models that allow for a unit root, conditional heteroskedasticity of unknown form as in \cite{gonccalves2004bootstrapping}, which includes ARCH and GARCH shocks, and a sequence of shocks that satisfy the martingale difference assumption. To obtain this result, we prove the root  $R_n(h)$ is asymptotically distributed as a standard normal distribution for sequences of AR(1) models with i.i.d. shocks (Theorem \ref{thm:sequence_models}). In particular, we prove that a high-level assumption (Assumption 3 in \cite{MO-PM2020corrigendum} and Assumption 4 in \cite{xu2022}) necessary for the theoretical properties of $C_n(h,1-\alpha)$ can be verified for sequences of AR(1) models with i.i.d. shocks (Proposition \ref{proposition:A3}).


Our first result implies that the LP-residual bootstrap method provides asymptotically valid confidence intervals over a large class of AR(1) models that allow for a unit root, conditional heteroskedasticity of unknown form (e.g., GARCH shocks), and martingale difference shocks. Moreover, our confidence interval $C_n^*(h,1-\alpha)$ has an asymptotic coverage equal to the nominal level $1-\alpha$ uniformly over $\rho \in [-1,1]$ and \textcolor{black}{a wide range of} intermediate horizons. 

Our second set of results shows that the LP-residual bootstrap provides asymptotic refinements to the confidence intervals on a more restricted class of AR(1) models (e.g., \textcolor{black}{$|\rho|\le 1-a$, where $a \in (0,1)$,} \textcolor{black}{and} i.i.d. shocks \textcolor{black}{with} positive continuous density), that is,  the size of the error in coverage probability (ECP) of $C_n^*(h,1-\alpha)$ is $o(n^{-1})$, whereas the size of the ECP of $C_n(h,1-\alpha)$ is $O(n^{-1})$. 
More concretely, Theorem \ref{thm:rates_bootstrap} shows the ECP of $C_n^*(h,1-\alpha)$ is $o(n^{-(1+\epsilon)})$ for some $\epsilon \in (0,1/2)$. 
To obtain these results, we derive Edgeworth expansions for the distribution of the root $R_n(h)$ and its bootstrap version for a fixed $h$ and \textcolor{black}{$|\rho| \le 1-a$, where $a \in (0,1)$; that is, the Edgeworth expansions are obtained for stationary AR(1) models and fixed horizons}. An informal discussion to calculate the size of the ECP using Edgeworth expansions appears in Section \ref{sec:why_boots}, while the formal results are established in Section \ref{sec:formal-results}.

Other bootstrap methods to construct confidence intervals for the impulse response coefficients have been considered and recommended based on simulation studies in the growing literature on LP inference. 
\cite{MO-PM2020} use a wild bootstrap procedure to generate new samples and compute critical values, but the theoretical results for their bootstrap method are unknown. \cite{KilianKim2011} present a simulation study including a block-bootstrap method to construct confidence intervals based on the LP approach, but the theory of their block-bootstrap method is unknown; see Remarks \ref{rem:block-bootstrap} and \ref{rem:block-bootstrap-2} for alternative block-bootstrap procedures with theoretical guarantees. Recently, \cite{lusompa2021local} proposes a block wild bootstrap method for confidence interval construction that is point-wise valid for a class of stationary data-generating processes; however, his bootstrap method is not applicable for an AR(1) model with a unit root. In contrast, we present a bootstrap method based on the LP approach with theoretical guarantees for a class of AR(1) models that allow for a unit root, conditional heteroskedasticity of unknown form, and martingale difference shocks.

More broadly, we contribute to the literature on confidence interval construction for impulse response coefficients. For short horizons (fixed $h$), the problem of confidence interval construction has been studied by \cite{andrews1993exactly}, \cite{hansen1999grid}, \cite{inoue2002bootstrapping}, \cite{Jorda2005}, \cite{Mikusheva2007,mikusheva2015second}, among others. For long horizons ($\textcolor{black}{h} \textcolor{black}{=} h_n \textcolor{black}{\propto} (1-b)n$,  $b \in (0,1)$), the problem of confidence interval construction was \textcolor{black}{discussed} and revised by \textcolor{black}{\cite{phillips1998impulse}}, \cite{gospodinov2004asymptotic}, \cite{pesavento2006small}, and \cite{Mikusheva2012} since the \textcolor{black}{standard} methods for short horizons \textcolor{black}{may} produce invalid confidence intervals when the data-generating process allows for unit roots. Recently, the problem of confidence interval construction for intermediate horizons ($h_n = o\left(n\right)$) was addressed in \cite{MO-PM2020} and \cite{xu2022}, which was a case not covered in the literature. In this paper, we propose bootstrap confidence intervals that are asymptotically valid at short and intermediate horizons. 



\textcolor{black}{
We also contribute to the literature on uniform inference in autoregressive models, where the confidence intervals for impulse response coefficients are uniformly valid, that is, they have an asymptotic coverage equal to the nominal level uniformly over the parameter space (e.g., uniformly over $ \rho \in [-1,1]$ for the AR(1) model). 
\cite{Mikusheva2007,Mikusheva2012} shows that the grid bootstrap proposed by \cite{hansen1999grid} provides confidence sets that are uniformly valid for the impulse responses when the sequence of shocks is a martingale difference sequence with constant conditional variance. However, it is unknown if the grid bootstrap is uniformly valid for AR(1) models with GARCH shocks; we report simulations for the grid bootstrap in Section \ref{sec:simulation_study} and Online Supplemental Appendix \ref{sec:appendix_table}. 
\cite{InoueKillian2020} show that confidence intervals based on a lag-augmented autoregressive method are uniformly valid for impulse response coefficients when the sequence of shocks is i.i.d. It is unknown if their results hold for martingale difference shocks. 
\cite{MO-PM2020} and \cite{xu2022} show that confidence intervals based on (lag-augmented) local projections are uniformly valid for impulse response coefficients; nevertheless, Monte Carlo simulations report lower coverage probability than expected. 
In contrast, our bootstrap method produces confidence intervals that are uniformly valid for a larger class of martingale difference shocks with conditional heteroskedasticity of unknown form (allowing for GARCH shocks).  
}

The remainder of the paper is organized as follows. In Section \ref{section:preliminaries}, we describe the setup and previous results. In Section \ref{sec:lp-residual-bootstrap}, we introduce our bootstrap confidence interval and the LP-residual bootstrap. In Sections \ref{sec:uniform_cons} and \ref{sec:asymp_refin}, we study the theoretical properties of the LP-residual bootstrap: uniform consistency and asymptotic refinements. In Section \ref{sec:simulation_study}, we investigate the numerical performance of the LP-residual bootstrap using a small simulation study. In Section \ref{sec:lp-residual-bootstrap_VAR}, we describe how to implement the LP-residual bootstrap for VAR models. Finally, in Section \ref{sec:conclusion}, we present concluding remarks. All the proofs are presented in Appendices \ref{sec:appendix_a} and \ref{sec:appendix_b}, and \textcolor{black}{Online} Supplemental Appendices \ref{sec:appendix_online1} and \ref{sec:appendix_online2}. Additional simulation results appear in \textcolor{black}{Online} Supplemental Appendix \ref{sec:appendix_table}.

\section{Setup and Previous Results on Local Projection}\label{section:preliminaries}

Consider an AR(1) model, 
\begin{equation}\label{eq:ar1_model}
    y_t = \rho y_{t-1} + u_t, \quad y_0 = 0, \quad \rho \in [-1,1]~.
\end{equation}
Denote the impulse response coefficient at horizon $h \in \mathbf{N}$ by 
\begin{equation}\label{eq:beta}
    \beta(\rho,h) \equiv \rho^h~.
\end{equation}
%
%
An estimator for $\beta(\rho,h)$ based on the LP approach is obtained as the slope coefficient of $y_t$ in the linear regression of $y_{t+h}$ on $y_t$ and $y_{t-1}$,
\begin{equation}\label{eq:la_lp}
    y_{t+h} = \hat{\beta}_n(h) y_t + \hat{\gamma}_n(h) y_{t-1} + \hat{\xi}_t(h), \quad t=1,\ldots,n-h~, 
\end{equation}
where $(\hat{\beta}_n(h),~\hat{\gamma}_n(h))$ and $\{\hat{\xi}_t(h): 1 \le t \le n-h\}$ are the coefficient vector and residuals of the linear regression \eqref{eq:la_lp}, respectively. This \emph{lag-augmented} LP approach was developed in  \cite{MO-PM2020}, where they give conditions under which the coefficient $\hat{\beta}_n(h)$ consistently estimates $\beta(\rho,h)$. 
Equation \eqref{eq:la_lp} is a lag-augmented  LP regression since the coefficient on $y_{t-1}$ is known to be zero under \eqref{eq:ar1_model}; see Remark \ref{rem:la-lp-purpose} for additional details on this LP approach.

Let $\hat{s}_n(h)$ be the heteroskedasticity-consistent (HC) standard error of $\hat{\beta}_n(h)$ in the lag-augmented LP regression \eqref{eq:la_lp}, which can be computed as follows

\begin{equation}\label{eq:HC_se}
    \hat{s}_n(h) \equiv \left( \sum_{t=1}^{n-h} \hat{u}_t(h)^2 \right)^{-1/2} \left(\sum_{t=1}^{n-h} \hat{\xi}_t(h)^2 \hat{u}_t(h)^2 \right)^{1/2} \left( \sum_{t=1}^{n-h} \hat{u}_t(h)^2 \right)^{-1/2}  ~,
\end{equation}
where $\hat{u}_t(h) \equiv y_t- \hat{\rho}_n(h) y_{t-1}$ and
\begin{equation}\label{eq:rho_hat_h}
    \hat{\rho}_n(h) \equiv \left(\sum_{t=1}^{n-h} y_{t-1}^2\right)^{-1} \left( \sum_{t=1}^{n-h} y_t y_{t-1} \right) ~.
\end{equation}

For a given $h \in \textbf{N}$, we consider the following real-valued root for the parameter $\beta(\rho,h)$:
\begin{equation}\label{eq:root}
    R_n(h) \equiv \frac{\hat{\beta}_n(h) - \beta(\rho,h)}{\hat{s}_n(h)}~,
\end{equation}
where $\beta(\rho,h)$ is as in \eqref{eq:beta}, $\hat{\beta}_n(h)$ is computed as in \eqref{eq:la_lp}, and $\hat{s}_n(h)$ is as in \eqref{eq:HC_se}. We denote the distribution of the root $R_n(h)$ by
\begin{equation}\label{eq:cdf_root}
    J_n(x,h,P,\rho) \equiv P_{\rho} \left( R_n(h) \le x \right)~,
\end{equation}  
where  $x \in \mathbf{R}$, $h \in \textbf{N}$, $P$ is the distribution of the shocks $\{u_t: t \ge 1 \}$, $\rho \in \textbf{R}$, and $P_{\rho}$ denote the probability distribution of the sequence $\{y_t: t \ge 1\}$, which is defined jointly by the distribution  $P$ and the parameter $\rho$ in \eqref{eq:ar1_model}. 

Let $c_n(h,1-\alpha)$ be the $1-\alpha$ quantile of $|R_n(h)|$ under the distribution $P_{\rho}$,
\begin{equation}\label{eq:cv_finite}
    c_n(h,1-\alpha)  \equiv \inf \left \{ u \in \textbf{R} : P_{\rho} \left( |R_n(h)| \le u \right) \ge 1-\alpha  \right\}~.
\end{equation}
Ideally, we would use the root $R_n(h)$ and the critical value $c_n(h,1-\alpha)$ to construct confidence sets for $\beta(\rho,h)$ with a coverage probability of 1-$\alpha$. That is collecting all the parameters $\beta(\rho,h)$ such that $|R_n(h)| \le c_n(h,1-\alpha)$, which is equivalent to defining the next confidence interval
\begin{align*}
    \tilde{C}_n(h,1-\alpha) \equiv  \left[ \hat{\beta}_n(h) - c_n(h,1-\alpha)~\hat{s}_n(h),~\hat{\beta}_n(h) + c_n(h,1-\alpha)~ \hat{s}_n(h)\right] ~.
\end{align*}
However, the critical value $c_n(h,1-\alpha)$ is unknown since the distribution of the root is unknown in general. As a result, the confidence interval $\tilde{C}_n(h,1-\alpha) $ is infeasible. For this reason, it is common to approximate the distribution of the root $R_n(h)$ relying on asymptotic distribution theory or bootstrap methods to approximate the infeasible $c_n(h,1-\alpha)$.

\subsection{Previous Results}\label{sec:prev_result}

The asymptotic distribution theory developed in \cite{MO-PM2020} and \cite{xu2022} implies that the distribution $J_n(x,h,P,\rho)$ converges to the standard normal distribution $\Phi(x)$ whenever certain assumptions on the distribution of the shocks $P$ hold. Moreover, this convergence is uniform over the values of $\rho \in [-1,1]$ and \textcolor{black}{a wide range of} intermediate horizons, that is 
\begin{equation}\label{eq:cdf_AA}
        \sup_{|\rho|\le 1} ~ \sup_{h \le h_n} ~  \sup_{x\in \mathbf{R}} |J_n(x,h,P,\rho) - \Phi(x)|    \to 0 \quad \text{as} \quad n \to \infty~,
\end{equation}
where \textcolor{black}{$h_n$ is any fixed sequence such that} $h_n \le n$ and $h_n = o\left(n\right)$. Assumptions \ref{xu_assumptions} and \ref{A3:MO-PM} in Section \ref{sec:uniform_cons} are sufficient conditions on the distribution $P$ to obtain \eqref{eq:cdf_AA} due to Theorem 2 in \cite{xu2022}. 

The confidence interval for $\beta(\rho,h)$ based on asymptotic distribution theory is defined as
\begin{equation}\label{eq:CI_MOPM}
        C_n(h,1-\alpha) \equiv \left[ \hat{\beta}_n(h) - z_{1-\alpha/2}~\hat{s}_n(h),~\hat{\beta}_n(h) + z_{1-\alpha/2}~ \hat{s}_n(h)\right] ~,
\end{equation}
where $z_{1-\alpha/2} \equiv \Phi^{-1}(1-\alpha/2)$ is the $1-\alpha/2$ quantile of the standard normal distribution. The result in \eqref{eq:cdf_AA} implies that \textcolor{black}{the} confidence interval $C_n(h,1-\alpha)$ is uniformly asymptotically valid in the sense that its asymptotic coverage probability is equal to the nominal level $1-\alpha$ uniformly over $\rho$ and \textcolor{black}{a wide range of intermediate horizons} $h$,
\begin{equation*}
        \sup_{|\rho|\le 1} ~ \sup_{h \le h_n} \left| P_{\rho} \left( \beta(\rho,h) \in C_n(h,1-\alpha) \right) - (1-\alpha) \right|   \to 0 \quad \text{as} \quad n \to \infty~,  
\end{equation*}
where \textcolor{black}{$h_n$ is any fixed sequence such that} $h_n \le n$ and $h_n = o\left(n\right)$. Three features of $C_n(h,1-\alpha)$ deserve further discussion. First, it is simpler to compute than \textcolor{black}{the} available alternatives in the sense that it does not require any tuning parameter. It is common to use heteroskedasticity- and autocorrelation-robust (HAR) standard errors for inference whenever we have dependent data. The major complication of HAR standard errors is the choice of the (truncation) tuning parameter; see  \cite{lazarus2018har}. In contrast, the HC standard errors $\hat{s}_n(h)$ defined in \eqref{eq:HC_se} are simple to compute and sufficient for inference under certain conditions on the distribution $P$; see Remark \ref{rem:hc_se} for further explanation. Second, the uniform asymptotic validity of the confidence interval $C_n(h,1-\alpha)$ avoids pre-testing procedures about the nature of the data-generating process ($|\rho|<1$ vs $\rho=1$) that can distort inference; see \cite{Mikusheva2007}. In particular, inference using $C_n(h,1-\alpha)$ holds regardless of the value of $\rho \in [-1,1]$. Third, the confidence interval $C_n(h,1-\alpha)$ has theoretical guarantees at intermediate horizons (e.g., \textcolor{black}{$h=h_n \propto n^{\zeta}$, $\zeta \in (0,1)$}). This is an important feature for inference on impulse response coefficients at intermediate horizons. Other methods to construct confidence intervals that work at short horizons ($h$ fixed) may have problems at long and intermediate horizons; see \textcolor{black}{\cite{phillips1998impulse}}, \cite{gospodinov2004asymptotic}, \cite{pesavento2006small}, \cite{Mikusheva2012}, and \cite{MO-PM2020} for additional discussion.

\begin{remark}\label{rem:hc_se}
The HC standard errors $\hat{s}_n(h)$ defined in \eqref{eq:HC_se} are sufficient for the construction of valid confidence intervals under certain conditions on the distribution $P$. In particular, as it was pointed out by \cite{xu2022}, it is sufficient and necessary that the \emph{scores} $\{ \xi_t(\rho,h) u_t : 1 \le t \le n-h \}$ be serially uncorrelated, where $\xi_t(\rho,h) \equiv \sum_{\ell=1}^{h} \rho^{h-\ell} u_{t+\ell}$. To explain the sufficiency of this condition, we use the derivations presented on page 1811 in \cite{MO-PM2020} that imply that the root $R_n(h)$ defined in \eqref{eq:root} can be written as follows
$$ \frac{\left((n-h)^{-1/2} \sum_{t=1}^{n-h} \xi_t(\rho,h) u_t \right)}{  E\left[ \xi_t(\rho,h)^2 u_t^2\right]^{1/2}} \times \frac{ \left[(n-h)^{-1} \sum_{t=1}^{n-h} \hat{\xi}_t(h)^2 \hat{u}_t(h)^2\right]^{-1/2} }{E\left[ \xi_t(\rho,h)^2 u_t^2\right]^{-1/2}} + \varepsilon_n(\rho,h)~,$$  
where $\varepsilon_n(\rho,h)$ is a remainder error term. We derive three implications under Assumptions \ref{xu_assumptions} and \ref{A3:MO-PM}, presented in Section \ref{sec:uniform_cons}. First, the term in parentheses converges to a normal distribution with variance correctly scaled by the denominator when the scores are serially uncorrelated. This condition is guaranteed by part (ii) of Assumption \ref{xu_assumptions}. Second, the term between brackets converges in probability to its denominator due to serially uncorrelated scores. Third, the remainder error term $\varepsilon_n(\rho,h)$ converges in probability to zero. Importantly, \cite{xu2022} proposed alternative standard errors for the construction of confidence intervals under serially correlated scores.
\end{remark}

\begin{remark}\label{rem:la-lp-purpose}
    The lag-augmented LP regression has the purpose of making the effective regressor of interest stationary. To see this, let us use the AR(1) model in \eqref{eq:ar1_model} to obtain $ y_{t+h} = \beta(\rho,h) y_t + \xi_t(\rho,h)$, where $\xi_t(\rho,h) = \sum_{\ell=1}^h \rho^{h-\ell} u_{t+h}$, which can be rewritten as 
    $$y_{t+h} = \beta(\rho,h) u_t + \rho \beta(\rho,h) y_{t-1} + \xi_t(\rho,h)~. $$
    Based on the previous equality, an estimator for $\beta(\rho,h)$ is defined as the slope coefficient of $u_t$ in the linear regression of $y_{t+h}$ on $u_t$ and $y_{t-1}$. This estimator is ideal since the effective regressor is stationary (by assumption). However, this regression is unfeasible since $u_t$ is not observed. Nevertheless, the estimator can also be obtained in the lag-augmented LP regression of $y_{t+h}$ on $y_t$ and $y_{t-1}$ since $y_t$ is a linear combination of $u_t$ and $y_{t-1}$ due to \eqref{eq:ar1_model}.
\end{remark}

\section{The LP-Residual Bootstrap}\label{sec:lp-residual-bootstrap}

This paper proposes an LP-residual bootstrap for confidence interval construction. Our confidence interval for the impulse response coefficient $\beta(\rho,h)$ is defined as
\begin{equation}\label{eq:bootstrap_CI}
    C_n^{*}(h,1-\alpha) \equiv \left[ \hat{\beta}_n(h) - c_n^*(h,1-\alpha)~\hat{s}_n(h),~\hat{\beta}_n(h) + c_n^*(h,1-\alpha)~ \hat{s}_n(h)\right] ~,
\end{equation}
where $\hat{\beta}_n(h)$ is an estimator for $\beta(\rho,h)$ defined in \eqref{eq:la_lp}, $\hat{s}_n(h)$ is its heteroskedasticity-consistent (HC) standard error defined in \eqref{eq:HC_se}, and $c_n^*(h, 1-\alpha)$ is a bootstrap critical value defined in \eqref{eq:cv_bootstrap}. %

\subsection{Bootstrap Critical Value}\label{sec:bootstrap_cv}

Let $Y^{(n)} \equiv \{y_t : 1 \le t \le n\}$ be data generated by \eqref{eq:ar1_model}. Let $c_n^*(h, 1-\alpha)$ be the bootstrap critical value involving the following steps:

\begin{itemize}
   \item[\textbf{Step 1}:]   Estimate $\rho$ in the AR(1) model defined in \eqref{eq:ar1_model} with the data $Y^{(n)}$ using linear regression, denoted by 
   \begin{equation}\label{eq:rho_hat}
       \hat{\rho}_n \equiv \left(\sum_{t=1}^n y_{t-1}^2\right)^{-1}\left(\sum_{t=1}^n y_{t-1}y_t\right)~,
   \end{equation}
   and compute the centered residuals
   \begin{equation}\label{eq:centered_residuals}
         \{ \Tilde{u}_t \equiv \hat{u}_t - n^{-1}\sum_{t=1}^n \hat{u}_t : 1 \le t \le n   \}~,
   \end{equation}
   where $\hat{u}_t \equiv y_{t} - \hat{\rho}_n y_{t-1}$. 

   \item[\textbf{Step 2}:] Generate a new sample of size $n$ using \eqref{eq:ar1_model}, \eqref{eq:rho_hat}, and \eqref{eq:centered_residuals}. Define the sample as 
   \begin{equation*}
       y_{b,t}^* = \hat{\rho}_n y_{b,t-1}^* + u_{b,t}^*~, \quad y_{b,0}^* = 0~,~ t=1,\ldots,n~,
   \end{equation*}
   where $\{u_{b,t}^* : 1 \le t \le n\}$ is a random sample from the empirical distribution of the centered residuals defined in \eqref{eq:centered_residuals}. 
   The new sample $\{ y_{b,t}^* : 1 \le t \le n \} $ is called the bootstrap sample.

   \item[\textbf{Step 3}:] Compute $\hat{\beta}_{b,n}^*(h)$ and $\hat{s}_{b,n}^*(h)$ as in \eqref{eq:la_lp} and \eqref{eq:HC_se} using the lag-augmented LP regression and the bootstrap sample $\{ y_{b,t}^* : 1 \le t \le n \} $. Define the bootstrap version of the root
   \begin{equation}\label{eq:root_bootstrap}
       R_{b,n}^*(h) =   \frac{\hat{\beta}_{b,n}^*(h) - \beta(\hat{\rho}_n,h)}{\hat{s}_{b,n}^*(h)} ~,
   \end{equation}
   where $\beta(\rho,h)$ and $\hat{\rho}_n$ are as in \eqref{eq:beta} and \eqref{eq:rho_hat}, respectively.

   \item[\textbf{Step 4}:] Define the bootstrap critical value as the $1-\alpha$ quantile of $|R_{b,n}^*(h)|$ conditional on the data $Y^{(n)}$, denoted by
   \begin{equation}\label{eq:cv_bootstrap}
       c_n^{*}(h, 1-\alpha) \equiv \inf \left \{ u \in \textbf{R} :  P_{\rho}\left ( |R_{b,n}^*(h)| \le u \mid Y^{(n)} \right) \ge 1-\alpha  \right\}~.
   \end{equation}
\end{itemize} 

We named this procedure the LP-residual bootstrap due to steps 2 and 3. Step 2 generates bootstrap samples based on the estimated model and a residual bootstrap procedure. Step 3 computes the bootstrap version of the root based on the lag-augmented LP regression. To our knowledge, this bootstrap procedure is new; see Remark \ref{rem:wild-bootstrap} and \ref{rem:block-bootstrap} for other bootstrap procedures involving roots based on LP estimators. 

We use the bootstrap critical value $c_n^{*}(h, 1-\alpha)$ in the construction of the confidence interval defined in \eqref{eq:bootstrap_CI}. The explicit formula in \eqref{eq:cv_bootstrap} has two implications. First, the bootstrap critical value $c_n^{*}(h, 1-\alpha)$ depends on the data, the sample size $n$, and the horizon $h$. Second, we can compute $c_n^{*}(h, 1-\alpha)$ with perfect accuracy whenever we use the exact empirical distribution of the centered residuals defined in \eqref{eq:centered_residuals}. However, the computation of an exact distribution can be computationally demanding; therefore, it is common to approximate it using Monte Carlo procedures as we describe in Remark \ref{rem:cv_bootstrap}, which has a theoretical justification due to the Glivenko–Cantelli theorem.

\begin{remark}\label{rem:cv_bootstrap}
It is a common practice to approximate the bootstrap critical value $c_n^{*}(h, 1-\alpha)$ using a Monte Carlo procedure (\cite{Horowitz2001, Horowitz2019}). We generate $B$ bootstrap samples of size \textcolor{black}{$n$}, where each $b$-th bootstrap sample $\{y_{b,t}^*: 1 \le t \le n\}$ is generated as in step 2. We then obtain $\{ |R_{b,n}^*(h)| : 1 \le b \le B \}$, where each $R_{b,n}^*(h)$ is computed as in step 3. Finally, we approximate the bootstrap critical value $c_n^{*}(h, 1-\alpha)$ by the $1-\alpha$ quantile of $\{ |R_{b,n}^*(h)| : 1 \le b \le B\}$, denoted by
   \begin{equation*}
       c_{b,n}^*(h, 1-\alpha) \equiv \inf \left \{ u \in \textbf{R} : \frac{1}{B} \sum_{b=1}^B I\left\{ |R_{b,n}^*(h)| \le u \right\} \ge 1-\alpha  \right\}~.
   \end{equation*} 
The accuracy of the approximation improves as the number of bootstrap samples $B$ increases. We use $B=1000$ in our simulation study presented in Section \ref{sec:simulation_study}. 
\end{remark}

\begin{remark}\label{rem:wild-bootstrap}
Another bootstrap procedure to approximate the infeasible critical value $c_n(h,1-\alpha)$ is presented in Section 5 of \cite{MO-PM2020}. 
They use the wild bootstrap procedure described in \cite{gonccalves2004bootstrapping}. For this reason, we name their procedure the LP-wild bootstrap. The only difference with respect to the LP-residual bootstrap is in Step 2. The LP-wild bootstrap defines the shocks as follows: $ u_{b,t}^* =   \tilde{u}_t z_{b,t}$ for all $t = 1, \ldots, n$, where $\{\tilde{u}_t: 1 \le t \le n\}$ are the centered residuals defined in \eqref{eq:centered_residuals} and $\{z_{b,t} : 1 \le t \le n \}$ is an i.i.d. sequence of standard normal random variables independent of the data $Y^{(n)}$. To our knowledge, the theoretical properties of the LP-wild bootstrap are unknown. We include the LP-wild bootstrap in our simulation study presented in Section \ref{sec:simulation_study}. 
\end{remark}

\begin{remark}\label{rem:bootstrap-equal-tailed}
An alternative to the symmetric percentile-t confidence interval defined in \eqref{eq:bootstrap_CI} is the equal-tailed percentile-t confidence interval. The latter is proposed and recommended in Section 5 of \cite{MO-PM2020}, while the former has been found to perform better in simulations reported by \cite{gonccalves2004bootstrapping}. Furthermore, symmetric confidence intervals are known to perform better asymptotically in terms of coverage error in the case of i.i.d. data; see Section 3.6 in \cite{Hall1992TheBA}. For these reasons, we focus on and study the properties of the symmetric percentile-t confidence interval in the next sections.  
Remark \ref{rem:percentile-t_bootstrap} presents additional discussion of the equal-tailed percentile-t confidence interval based on the LP-residual bootstrap. We include equal-tailed percentile-t confidence intervals based on both LP-residual and LP-wild bootstrap in our simulation study in Section \ref{sec:simulation_study}.
\end{remark}
 
\begin{remark}\label{rem:miss-specification}
We propose the LP-residual bootstrap method for constructing confidence intervals, aiming to provide a more accurate asymptotic approximation than the first-order asymptotic distribution for conducting inference. In Sections \ref{sec:uniform_cons} and \ref{sec:asymp_refin}, we study the validity of this bootstrap method and its theoretical properties under assumptions on the distribution of the shocks and under the assumption of correct specification, i.e., the data are generated from the AR(1) model in \eqref{eq:ar1_model}. To our knowledge, the theoretical properties of the root $R_n(h)$ for general forms of misspecification are unknown. Recent work by \cite{olea2024double} imply that $R_n(h)$ is still asymptotically pivotal under a specific form of local misspecification. 
The analysis of the theoretical properties of the LP-residual bootstrap under misspecification is \textcolor{black}{outside} the scope of this paper. 
%
\end{remark}

\section{Uniform Consistency}\label{sec:uniform_cons}

We show the uniform consistency of the LP-residual bootstrap (Theorem \ref{thm:consistency}) and that our proposed bootstrap confidence interval $C_n^*(h,1-\alpha)$ defined in  \eqref{eq:bootstrap_CI} is uniformly asymptotically valid (Theorem \ref{thm:CI_bootstrap}). In what follows, we first present and discuss the assumptions, and we then establish the results. 
 
The following assumption imposes restrictions on the distributions of the shocks $P$. These assumptions are based on the general framework developed by \cite{xu2022} that generalized the work of \cite{MO-PM2020}. 

\begin{assumption}\label{xu_assumptions} 
\hspace{1cm}\vspace{-0.25cm}

\begin{enumerate}
    \item[i)] $\{u_t: 1 \le t \le n\}$ is covariance-stationary and satisfies $E[u_t \mid \{u_s \}_{s < t}] = 0$ almost surely. 

    \item[ii)] $E[u_t^2u_{t-s}u_{t-r}] = 0$ for all $s \neq r$, for all $t,r,s \ge 1$.    
    
    \item[iii)] $\{u_t: 1 \le t \le n\}$  is strong mixing with mixing numbers $\{ \alpha(j): j \ge 1\}$. There exists $\zeta > 2$, $\epsilon>1$, and $C_{\alpha}<\infty$, such that $\alpha(j) \le C_{\alpha}j^{-2\zeta \epsilon/(\zeta-2)}$, for all $j$.
    
    \item[iv)] For $\zeta$ defined in (iii), $E[u_t^{8\zeta}] \le C_8 < \infty$, and $E[u_t^2 \mid \{ u_s \}_{s<t}] \ge C_\sigma $ almost surely.
    
\end{enumerate}
\end{assumption}

Part (i) of Assumption \ref{xu_assumptions} assumes the shocks are a martingale difference sequence. This assumption allows for uncorrelated dependent shocks and implies that the shock $u_t$ is uncorrelated with $y_{t-1}$. Part (ii) in Assumption \ref{xu_assumptions} includes a large class of conditional heteroskedastic autoregressive models (e.g., ARCH and GARCH shocks), and it has been common in the literature; for instance, \cite{gonccalves2004bootstrapping} use a similar assumption (Assumption A') to prove the asymptotic consistency of the wild bootstrap for autoregressive processes. Moreover, this assumption implies that the process $\{ \xi_t(\rho,h) u_t : 1 \le t \le n-h \}$ is serially uncorrelated, where $\xi_t(\rho,h) \equiv \sum_{\ell=1}^{h} \rho^{h-\ell} u_{t+\ell}$, which is important for the use of HC standard errors as we discussed in Remark \ref{rem:hc_se}. Part (iii) and (iv) of Assumption \ref{xu_assumptions} are mild regularity conditions on the distribution of the shocks $P$ to establish uniform bounds of approximation errors, which can be relaxed if stronger assumptions are imposed over the serial dependence of the shocks; see Assumption \ref{appendix:AV_assumptions} in Appendix \ref{sec:appendix_b}.

The next assumption is a high-level assumption and imposes additional restrictions on the distributions of the shocks $P$.
 
\begin{assumption}\label{A3:MO-PM}  \hspace{1cm}\vspace{-0.25cm}

\begin{equation*}
     \lim_{M \to \infty} ~ \lim_{n \to \infty} ~ \inf_{ |\rho| \le 1 } 
 ~ P_{\rho} \left(~  g(\rho,n)^{-2} ~n^{-1}\sum_{t=1}^{n}y_{t-1}^2  \ge 1/M ~\right) = 1~,  
\end{equation*}
where $ g(\rho,k) = \left(\sum_{\ell=0}^{k-1} ~ \rho^{2 \ell}\right)^{1/2}~.$ 
\end{assumption} 

This assumption implies that the estimator $\hat{\rho}_n(h)$ defined in \eqref{eq:rho_hat_h} is well-behaved in the sense that its denominator after scaled by the factor $g(\rho,n-h)$ converges to a strictly positive limit. As a result, we can replace the residual $\hat{u}_t(h) \equiv y_t - \hat{\rho}_n(h)y_{t-1}$ by the shock $u_t$, which implies the second and third implication discussed in Remark \ref{rem:hc_se}. We show in Proposition \ref{proposition:A3} that Assumption \ref{A3:MO-PM} can be verified if the shocks are i.i.d. and satisfied mild regularity conditions (Assumption \ref{appendix:AV_assumptions}). In Appendix C of \cite{MO-PM2020}, this assumption is verified for AR(1) models whenever a \emph{contiguity condition} holds.

Assumptions \ref{xu_assumptions} and  \ref{A3:MO-PM} guarantee that the distribution $J_n(\cdot,h,P,\rho)$ defined in \eqref{eq:cdf_root} can be approximated by the standard normal distribution $\Phi(\cdot)$ uniformly on $\rho \in [-1,1]$ and \textcolor{black}{a wide range of horizons} $h$  as in \eqref{eq:cdf_AA}. Let $\hat{P}_n$ be the empirical distribution of the centered residuals defined in \eqref{eq:centered_residuals} and let $\hat{\rho}_n$ be the estimator of $\rho$ defined in \eqref{eq:rho_hat}. Using this notation $J_n(\cdot,h,\hat{P}_n,\hat{\rho}_n)$ is the distribution of the bootstrap root $R_{b,n}^*(h)$ defined in \eqref{eq:root_bootstrap} conditional on the data $Y^{(n)}$. The next theorem shows that the distribution $J_n(\cdot,h,P,\rho)$ can be approximated by the bootstrap distribution $J_n(\cdot,h,\hat{P}_n,\hat{\rho}_n)$ uniformly on  $\rho \in [-1,1]$ and \textcolor{black}{a wide range of}  intermediate horizons (e.g., \textcolor{black}{uniform over $h \le h_n$, where $h_n$ is any fixed sequence such that $h_n = o(n)$}), i.e., the LP-residual bootstrap is uniformly consistent. 

\begin{theorem}\label{thm:consistency}
    Suppose Assumptions \ref{xu_assumptions} and \ref{A3:MO-PM} hold. Then, for any $\epsilon>0$ and for any sequence $h_n$ \textcolor{black}{such that $h_n \le n$ and} $h_n = o\left(n\right)$, we have
    \begin{equation}\label{eq:thm_consistency}
        \sup_{|\rho|\le 1} P_{\rho}\left( \sup_{h \le h_n} ~ \sup_{x\in \mathbf{R}} | J_n(x,h,P,\rho) - J_n(x,h,\hat{P}_n,\hat{\rho}_n) | > \epsilon \right) \to 0 \quad \text{as} \quad n \to \infty~,
    \end{equation}    
    where $J_n(x,h,\cdot,\cdot)$ is as in \eqref{eq:cdf_root}, $\hat{P}_n$ is the empirical distribution of the centered residuals defined in \eqref{eq:centered_residuals}, and $\hat{\rho}_n$ is as in \eqref{eq:rho_hat}.
\end{theorem} 

Theorem \ref{thm:consistency} shows that the LP-residual bootstrap is uniformly consistent, i.e., the bootstrap distribution $J_n(\cdot,h,\hat{P}_n,\hat{\rho}_n)$ approximates the distribution $J_n(\cdot,h,P,\rho)$ uniformly over the parameter space ($\rho \in [-1,1]$) and \textcolor{black}{a wide range of} intermediate horizons  ($h \le h_n)$. Two features of this uniform approximation result deserve further discussion. First, uniform consistency of bootstrap methods over the parameter spaces of autoregressive models is not just a technical detail but a crucial property to guarantee reliable inference methods; see \cite{Mikusheva2007}. Otherwise, it is possible to obtain for any sample size $n$ a parameter $\rho_n$ such that the distance between the distributions $J_n(\cdot,h,\hat{P}_n,\hat{\rho}_n)$ and $J_n(\cdot,h,P,\rho)$ is far from zero. Second, the uniform approximation over the horizons is necessary for inference purposes at intermediate horizons. Other valid methods for a fixed $h$ do not necessarily work for $h$ growing with the sample size.

The proof of Theorem \ref{thm:consistency} is presented in Appendix \ref{appendix:thm:consistency}. It has two main ideas. First, we show that the approximation result presented in \eqref{eq:cdf_AA} also holds for sequences of AR(1) models with i.i.d. shocks (Theorem \ref{thm:sequence_models}), 
\begin{equation*}
  \sup_{P \in \mathbf{P}_{n,0}}   ~ \sup_{h \le h_n} ~  \sup_{|\rho| \le 1} ~ \sup_{x \in \mathbf{R}} ~ \left|  J_n(x,h,P,\rho) - \Phi(x)   \right| \to 0 \quad \text{as} \quad n \to \infty~,
\end{equation*}
where $\mathbf{P}_{n,0}$ denotes the set of all distributions that satisfy Assumption \ref{appendix:AV_assumptions} in Appendix \ref{sec:appendix_uniform}, $h_n$ is as in Theorem \ref{thm:consistency},  $J_n(\cdot,h,P,\rho)$ is as in \eqref{eq:cdf_root} and $\Phi(\cdot)$ is the standard normal distribution. Assumption \ref{appendix:AV_assumptions} imposes stronger restrictions on the dependence of the shocks (i.i.d.) and some mild regularity conditions. The formal result is presented in Appendix \ref{sec:appendix_uniform} as Theorem \ref{thm:sequence_models}. Second, we show that Assumptions \ref{xu_assumptions} and \ref{A3:MO-PM} imply the existence of a sequence of events $E_n$ with probability approaching 1 such that the empirical distributions $\hat{P}_n $ conditional on the event $E_n$ verify Assumption \ref{appendix:AV_assumptions}. In other words, we show that $\hat{P}_n \in \mathbf{P}_{n,0}$ holds with a probability approaching 1. The construction of the events $E_n$ relies on Lemma \ref{AppendixA:lemma_event_En} in Appendix \ref{appendixB:lemmas}. We use the previous two ideas to approximate the distribution $J_n(\cdot,h,\hat{P}_n,\hat{\rho}_n)$ by the standard normal distribution $\Phi(\cdot)$ conditional on the event $E_n$. Finally, we conclude that the distributions $J_n(\cdot,h,\hat{P}_n,\hat{\rho}_n)$ and $J_n(\cdot,h,P,\rho)$ are asymptotically close since both have the same asymptotic limit.

The next result shows that the confidence interval $C_n^*(h,1-\alpha)$ defined in \eqref{eq:bootstrap_CI} is uniformly asymptotically valid in the sense that its asymptotic coverage probability is equal to $1-\alpha$ uniformly over $\rho$ and \textcolor{black}{a wide range of horizons} $h$.

\begin{theorem}\label{thm:CI_bootstrap}
Suppose Assumptions \ref{xu_assumptions} and \ref{A3:MO-PM} hold. Then, for any sequence $h_n$ \textcolor{black}{such that $h_n \le n$ and} $h_n = o\left(n\right)$, we have
\begin{equation}\label{eq:thm_CI_bootstrap}
    \sup_{|\rho|\le 1} ~  \sup_{h \le h_n} \left| P_{\rho} \left( \beta(\rho,h) \in C_n^*(h,1-\alpha)  \right) - (1-\alpha) \right|   \to 0 \quad \text{as} \quad n \to \infty~,
\end{equation}
where $\beta(\rho,h)$ and $C_n^*(h,1-\alpha)$ are as in \eqref{eq:beta} and \eqref{eq:bootstrap_CI}, respectively.
\end{theorem}

Theorem \ref{thm:CI_bootstrap} provides the theoretical justification to conduct inference on the impulse response coefficient $\beta(\rho,h)$ using our bootstrap confidence interval $C_n^*(h,1-\alpha)$. Note that the only difference with respect to the confidence interval $C_n(h,1-\alpha)$ defined in \eqref{eq:CI_MOPM} is the critical value, which was equal to $z_{1-\alpha/2}$. The critical value $z_{1-\alpha/2}$ was the same for different sample sizes $n$ and horizons $h$. Instead, we now use a critical value $c_n^*(h,1-\alpha)$ that depends on the data, the sample size, and the horizon. We evaluate the difference in coverage probability between the confidence intervals $C_n(h,1-\alpha)$  and $C_n^*(h,1-\alpha)$ using simulations in Section \ref{sec:simulation_study}. The simulation results provide evidence that the coverage probability of our proposed confidence interval $C_n^*(h,1-\alpha)$  is closer to $1-\alpha$ than that of $C_n(h,1-\alpha)$. 

The proof of Theorem \ref{thm:CI_bootstrap} is presented in Appendix \ref{appendix:thm:CI_bootstrap}. It only relies on the uniform consistency of the bootstrap procedure. We next sketch the main arguments of the proof. We first note that \eqref{eq:thm_CI_bootstrap} is equivalent to
$$ \sup_{|\rho|\le 1} ~ \sup_{h \le h_n} \left| P_{\rho} \left( |R_n(h)| \le c_n^*(h,1-\alpha)   \right) - (1-\alpha) \right|   \to 0 \quad \text{as} \quad n \to \infty~. $$
We then use that the bootstrap critical value $c_n^*(h,1-\alpha)$ is included in $[z_{1-\alpha/2 -\epsilon},~z_{1-\alpha/2 +\epsilon}]$ with a probability approaching 1 for arbitrary $\epsilon>0$; see Lemma \ref{lemma:quantiles_high-prob} in Appendix \ref{appendixB:lemmas}. This result is possible because the root $R_n(h)$ is asymptotically normal and the LP-residual bootstrap is uniformly consistent. Third, we can conclude using algebra manipulation and the asymptotic normality of the root $R_n(h)$ that 
$$ \limsup_{n \to \infty} \sup_{|\rho|\le 1}  \sup_{h \le h_n} \left| P_{\rho} \left( |R_n(h)| \le c_n^*(h,1-\alpha)   \right) - (1-\alpha) \right| \le 2\epsilon ~,$$
which implies \eqref{eq:thm_CI_bootstrap} since $\epsilon>0$ was arbitrary.   
 
\begin{remark}\label{rem:percentile-t_bootstrap}
We can use the LP-residual bootstrap to construct equal-tailed percentile-$t$ confidence intervals denoted by $C_{per-t,n}^{*}(h,1-\alpha)$. That is
\begin{equation}\label{eq:CI_RB-per-t}
    C_{per-t,n}^{*}(h,1-\alpha) \equiv \left[ \hat{\beta}_n(h) - q_n^*(h,1-\alpha/2)~\hat{s}_n(h),~\hat{\beta}_n(h) - q_n^*(h,\alpha/2)~ \hat{s}_n(h)\right] ~,
\end{equation}
where $\hat{\beta}_n(h)$ is as in \eqref{eq:la_lp}, $\hat{s}_n(h)$ is as in \eqref{eq:HC_se}, and $q_n^*(h,\alpha_0)$ is the $\alpha_0$-quantile of the bootstrap root $R_{b,n}^*(h)$ defined in \eqref{eq:root_bootstrap}. Three features of $C_{per-t,n}^{*}(h,1-\alpha)$ deserve further discussion. First, the bootstrap quantiles $q_n^*(h,\alpha_0)$ can be approximated using Monte Carlo procedures in a similar way as we discussed in Remark \ref{rem:cv_bootstrap}. Second, the confidence interval $C_{per-t,n}^{*}(h,1-\alpha)$ can be asymmetric around $\hat{\beta}_n(h)$ by construction, which is not the case of $C_n^{*}(h,1-\alpha)$ that is a symmetric one. Third, $C_{per-t,n}^{*}(h,1-\alpha)$ is uniformly asymptotically valid,
$$ \sup_{|\rho|\le 1}  \sup_{h \le h_n} \left| P_{\rho} \left( \beta(\rho,h) \in C_{per-t,n}^{*}(h,1-\alpha)  \right) - (1-\alpha) \right|   \to 0 \quad \text{as} \quad n \to \infty~,$$
where \textcolor{black}{$h_n$ is any fixed sequence such that} $h_n \le n$ and $h_n = o\left(n\right)$. The proof of this claim follows directly by Theorem \ref{thm:consistency}, Lemma \ref{lemma:quantiles_high-prob}, and the proof of Theorem \ref{thm:CI_bootstrap}. We include $C_{per-t,n}^{*}(h,1-\alpha)$ in our simulation study in Section \ref{sec:simulation_study}. 
\end{remark}


\begin{remark}\label{rem:grid_bootstrap}
For short horizons (fixed $h$), the available grid bootstrap (\cite{hansen1999grid,Mikusheva2012}) is a valid alternative to our bootstrap confidence interval $C_n^*(h,1-\alpha)$ when the conditional variance of the shocks is constant. The grid bootstrap is a method to construct confidence intervals for the parameter $\beta(\rho,h)$ defined in \eqref{eq:beta} based on test inversion.   \cite{Mikusheva2007,Mikusheva2012} shows that the grid bootstrap provides confidence intervals that are uniformly asymptotically valid in the sense that its asymptotic coverage probability is equal to $1-\alpha$ uniformly on $\rho \in [-1,1]$. Nevertheless, when the conditional variance of the shocks is not constant (e.g., GARCH shocks), it is unknown if the confidence intervals based on the grid bootstrap are valid. In contrast, $C_n^*(h,1-\alpha)$ remains valid for a larger class of AR(1) models. 
%
We include the grid bootstrap presented in \citet[Section 3.3]{Mikusheva2012} in our simulation study presented in Section \ref{sec:simulation_study}.
\end{remark}

\begin{remark}\label{rem:long-horizon}
If we restrict our analysis to data-generating processes with weak dependence (e.g., $|\rho| \le 1-a$ for some $a \in (0,1)$) and consider stronger assumptions in the distribution of the shocks $\{u_t: 1 \le t \le n\}$, then both claims in \eqref{eq:thm_consistency} and \eqref{eq:thm_CI_bootstrap} can hold for long horizons (e.g., $h_n \le (1-b)n$ for some $b \in (0,1)$).
In other words, the confidence interval $C_n^*(h,1-\alpha)$ has theoretical guarantees for long horizons under certain conditions. 
Assumptions 1-2 in \cite{MO-PM2020} are sufficient to guarantee this claim; a formal proof can be derived following the same strategy presented in Appendix \ref{sec:appendix_a} to prove Theorem \ref{thm:consistency} and \ref{thm:CI_bootstrap}.  
In particular, the proof of Theorem \ref{thm:sequence_models} can be adapted for long horizons since $|\rho| \le 1-a$ implies that $g(\rho,h_n)^2/(n-h_n) \to 0$ as $n \to \infty$ for any $h_n \le (1-b) n$, where $g(\rho,h) = \{ \sum_{\ell=1}^{h} \rho^{2(\ell-1)} \}^{1/2}$. This technical condition was satisfied when $|\rho| \le 1$ and $h_n = o(n)$.  
\end{remark}

\begin{remark}\label{rem:var-stationary}
For strictly stationary data, the results in Theorems  \ref{thm:consistency} and \ref{thm:CI_bootstrap} can be extended to vector autoregressive (VAR) models considered in \cite{MO-PM2020} that satisfy their Assumptions 1 and 2. A proof of these extensions may be done using the finite sample inequalities presented in their online appendix and following the approach we presented in Appendixes \ref{sec:appendix_a} and \ref{sec:appendix_b}. We leave the details of a formal proof to future research. For non-stationary data, it is an open question whether the LP-residual bootstrap is consistent for VAR models. Our approach relies on verifying  Assumption \ref{A3:MO-PM} for an appropriate sequence of AR(1) models; therefore, an analogous approach may require a similar step for VAR models, which is \textcolor{black}{outside} the scope of this paper.
\end{remark}

\begin{remark}\label{rem:other-valid-bootstraps}
    We can use Theorem \ref{thm:CI_bootstrap} to show the uniform validity of alternative methods to construct confidence intervals for $\beta(\rho,h)$; however, some alternative confidence intervals can be impractical at the intermediate horizon. For instance, a confidence interval $C_{la-ar}^*(h,1-\alpha)$ for $\beta(\rho,h)$ can be obtained by first constructing a confidence interval for $\rho$ using Theorem \ref{thm:CI_bootstrap} (taking $h=1$) and then by using $\beta(\rho,h) = \rho^h$ (monotone transformation). Unfortunately, the confidence interval $C_{la-ar}^*(h,1-\alpha)$ can be very wide asymptotically for certain data-generation processes and intermediate horizons.  More concretely, for any $L>1$ it can be shown $P_{\rho}\left( [1/L,L] \subseteq  C_{la-ar}^*(h,1-\alpha) \right) \to 1 $ as $n \to \infty$ when $\rho = 1-c_1/n$ (local-to-unit models) and $h \sim \sqrt{n}$. We formally establish this result in Proposition \ref{prop:CI_LA_AR}  in Appendix \ref{sec:appendix_b}.     
    This result is similar to the ones presented in Appendix B.2.2 in \cite{MO-PM2020}  for the lag-augmented AR bootstrap confidence interval of \cite{InoueKillian2020}, which is a bootstrap confidence interval related but different to  $C_{la-ar}^*(h,1-\alpha)$.
\end{remark}

\section{Asymptotic Refinements}\label{sec:asymp_refin}

\textcolor{black}{This section will impose conditions on the data-generating process that further restrict the class of AR(1) models relative to that considered in Section \ref{sec:uniform_cons}, ruling out local-to-unity and unit-root models. These conditions are explicit in Theorems \ref{thm:rates_AA} and \ref{thm:rates_bootstrap}, where we calculate the sizes of the error in coverage probability (ECP) for the confidence intervals $C_n(h,1-\alpha)$ and $C_n^*(h,1-\alpha)$ defined in \eqref{eq:CI_MOPM} and \eqref{eq:bootstrap_CI}, respectively. The results of these theorems show that the LP-residual bootstrap can provide asymptotic refinements for confidence intervals, that is, the ECP of $C_n^{*}(h,1-\alpha)$ is $o(n^{-1})$, whereas the ECP of $C_n(h,1-\alpha)$ is $O(n^{-1})$.}

\textcolor{black}{Section \ref{sec:why_boots} first provides an informal discussion of the elements and challenges involved in obtaining
asymptotic refinements for confidence intervals with the LP-residual bootstrap. Section \ref{sec:formal-results} then formalizes the discussion by giving conditions on the data-generating process (Assumption \ref{AV_assumptions} and $\rho \in [-1+a, 1-a] $ for a given $a \in (0,1)$) that are sufficient to establish these asymptotic refinements (Theorems \ref{thm:rates_AA} and \ref{thm:rates_bootstrap}).} 


\subsection{Informal Discussion on Asymptotic Refinements}\label{sec:why_boots}

\textcolor{black}{This section gives an informal exposition on how a bootstrap method can provide asymptotic refinements for confidence intervals when the root is \emph{asymptotically pivotal}, i.e., the asymptotic distribution of the root does not depend on any unknown parameters.} 
The explanation below is not new; see \cite{hall1996bootstrap}, \cite{Horowitz2001,Horowitz2019}, and \cite{lahiri2003resampling}. It has the purpose of introducing the main elements and challenges that arise to obtain asymptotic refinements \textcolor{black}{in the context of dependent data generated from an AR(1) model. It also describes the approach considered in this paper; see Remark \ref{rem:mikusheva2015} for alternative methods. } 


%
 \textcolor{black}{\textbf{Main elements:} For the sake of exposition, suppose the root $R_n(h)$ has an \emph{Edgeworth expansion} up to an error of size $o(n^{-1})$, that is, the distribution of the root $R_n(h)$ has an asymptotic expansion,
\begin{equation}\label{eq:edgeworth_expansion_sec2}
    J_n(x,h,P,\rho) = \Phi(x) + \sum_{j=1}^2 n^{-j/2} q_j(x,h,P,\rho)\phi(x)  + o\left(n^{-1}\right)~,
\end{equation}
where $q_j(x,h,P,\rho)$ are polynomials in $x \in \textbf{R}$ such that (i) their coefficients are continuous function of moments of $P$ and $\rho$ and (ii) $q_j(x,h,P,\rho) = (-1)^{j+1} q_j(-x,h,P,\rho)$ for $j=1,2$. Similarly, suppose the bootstrap root $R_n^*(h)$ has an Edgeworth expansion, 
\begin{equation}\label{eq:edgeworth_expansion_boots_sec2}
    J_n(x,h,\hat{P}_n,\hat{\rho}_n) = \Phi(x) + \sum_{j=1}^2 n^{-j/2} q_j(x,h,\hat{P}_n,\hat{\rho}_n)\phi(x) + o_p\left(n^{-1}\right)~,
\end{equation}
where $J_n(x,h,\cdot,\cdot)$ is as in  \eqref{eq:cdf_root}, $\hat{P}_n$ is the empirical distribution of the centered residuals defined in \eqref{eq:centered_residuals}, and $\hat{\rho}_n$ is the estimator of $\rho$ defined in \eqref{eq:rho_hat}.} 

\textcolor{black}{
The approximations in \eqref{eq:edgeworth_expansion_sec2} and \eqref{eq:edgeworth_expansion_boots_sec2} are commonly used to show that the bootstrap methods provide more accurate approximations than the asymptotic distribution theory; see \cite{Hall1992TheBA} for a textbook reference for the case of i.i.d. data. We next sketch an informal calculation of the sizes of the ECP of the confidence intervals $C_n(h,1-\alpha)$ and $C_n^{*}(h,1-\alpha)$.}

\textcolor{black}{The coverage probability of $C_n(h,1-\alpha)$ is equal to $P_{\rho} \left(  |R_n(h)| \le z_{1-\alpha/2} \right)$ by the definitions of $C_n(h,1-\alpha)$ and $R_n(h)$ in \eqref{eq:CI_MOPM} and \eqref{eq:root}, respectively. Note that \eqref{eq:edgeworth_expansion_sec2} and the properties of $q_j(\cdot,h,P,\rho)$ imply that for any $x>0$, we have
\begin{equation}\label{eq:aux1_sec2}
    P_{\rho} \left(  |R_n(h)| \le x \right) = 2\Phi(x) - 1 + n^{-1} 2q_2(x,h,P,\rho) \phi(x) + o\left(n^{-1}\right)~.
\end{equation}
Taking $x = z_{1-\alpha/2}$, we conclude the size of the ECP of $C_n(h,1-\alpha)$ is  $O(n^{-1})$.} 

\textcolor{black}{Similarly, the coverage probability of $C_n^*(h,1-\alpha)$ is equal to $P_{\rho} \left(  |R_n(h)| \le c_n^*(h,1-\alpha)  \right)$ by the definitions in \eqref{eq:bootstrap_CI} and \eqref{eq:root}. Now, we will argue that 
\begin{align}
     P_{\rho} \left(  |R_n(h)| \le c_n^*(h,1-\alpha) \right)
     = P_{\rho} \left(  |R_n(h)| \le c_n(h,1-\alpha)\right)  + o\left(n^{-1}\right)~,\label{eq:aux2_sec51}
\end{align}
where $c_n(h,1-\alpha)$ is as in \eqref{eq:cv_finite}. This is sufficient to conclude that the size of the ECP of $C_n^*(h,1-\alpha)$ is $o(n^{-1})$ since $P_{\rho} \left(  |R_n(h)| \le c_n(h,1-\alpha)\right) = 1-\alpha$ by definition. Using the properties of $q_j(\cdot,h,\hat{P}_n,\hat{\rho}_n)$ and  \eqref{eq:edgeworth_expansion_boots_sec2}, we obtain
\begin{align}
    P_{\rho} \left(  |R_{b,n}^*(h)| \le x \mid Y^{(n)}\right) &= 2\Phi(x) - 1 + n^{-1} 2q_2(x,h,\hat{P}_n,\hat{\rho}_n) \phi(x) + o_p\left(n^{-1}\right)~\notag\\
    &= 2\Phi(x) - 1 + n^{-1} 2q_2(x,h,P,\rho) \phi(x) + o_p\left(n^{-1}\right)~,\label{eq:aux2_sec2}
\end{align}
where the last equality uses $q_2(x,h,\hat{P}_n,\hat{\rho}_n)  = q_2(x,h,P,\rho) + o_p\left(1\right) $. Note that \eqref{eq:aux1_sec2} and \eqref{eq:edgeworth_expansion_boots_sec2} looks similar.
Taking $x = c_n(h,1-\alpha)$ in \eqref{eq:aux1_sec2} and $x = c_n^*(h,1-\alpha)$ in \eqref{eq:aux2_sec2}, it can be conclude that $c_n^*(h,1-\alpha) = c_n(h,1-\alpha) + o_p\left(n^{-1}\right)$, which will imply \eqref{eq:aux2_sec51}. }

\textcolor{black}{
The informal explanation presented above suggests that the LP-residual bootstrap can provide asymptotic refinements when there exist valid Edgeworth expansions as in  \eqref{eq:edgeworth_expansion_sec2}-\eqref{eq:edgeworth_expansion_boots_sec2}. We present in Section \ref{sec:formal-results} conditions (Assumption \ref{AV_assumptions} and $\rho \in [-1+a, 1-a] $, where $a \in (0,1)$) under which the previous informal discussion can be formalized.  
}

\textcolor{black}{
\textbf{The challenges:} Edgeworth expansions as in \eqref{eq:edgeworth_expansion_sec2}-\eqref{eq:edgeworth_expansion_boots_sec2} are not always available or valid in the context of AR(1) models. For instance, in the case of the local-to-unity and unit-root models, the Edgeworth expansion for the least-squares estimate of the AR(1) model defined in \eqref{eq:rho_hat} is no longer valid; see \cite{phillips2023estimation}. In this case, alternative asymptotic approximations were developed to prove asymptotic refinements of the bootstrap, e.g., \cite{park2003bootstrap,park2006bootstrap} and \cite{mikusheva2015second}. 
To our knowledge, there are no available theoretical results about valid Edgeworth expansions for the root $R_n(h)$ defined in \eqref{eq:root} that can be applied directly.} 

\textcolor{black}{
Nevertheless, for stationary AR(1) models (when $\rho \in [-1+a,1-a]$,  $a \in (0,1)$), asymptotically valid Edgeworth expansions were obtained; see \cite{phillips1977approximations, phillips1977general}, \cite{bose1988}, among others. Therefore, we will restrict our analysis to stationary AR(1) models to obtain valid Edgeworth expansions for the root $R_n(h)$ and its bootstrap version $R_n^*(h)$ when $\rho \in [-1+a,1-a]$,  $a \in (0,1)$, and  $h$ is fixed.}

\subsection{Formal Conditions and Results}\label{sec:formal-results}
 
This section presents conditions under which the LP-residual bootstrap provides asymptotic refinements to the confidence interval. Under these conditions, we calculate the sizes of the ECP for $C_n(h,1-\alpha)$ and $C_n^*(h,1-\alpha)$ in Theorems \ref{thm:rates_AA} and  \ref{thm:rates_bootstrap}, respectively.  

The following assumption imposes stronger conditions on the distribution of the shocks $P$ than the ones presented in Assumption \ref{xu_assumptions}. We use this assumption \textcolor{black}{and $\rho \in [-1+a, 1-a]$ for some $a \in (0,1)$} to formalize the informal explanation about asymptotic refinements presented in Section \ref{sec:why_boots}. 

\begin{assumption}\label{AV_assumptions}
\hspace{1cm}\vspace{-0.25cm}

\begin{enumerate}
    \item[i)] $\{u_t: 1 \le t \le n\}$ is a sequence of i.i.d. random variables with $E[u_t] =0$.

    \item[ii)] $u_t$ has a positive continuous density.

    \item[iii)] $E[e^{x u_t}] \le e^{x^2 c_u^2}$ for all $|x| \le 1/c_u$ and $E[u_t^2] \ge C_{\sigma} $ for some constants $c_u, C_\sigma>0$.    
\end{enumerate}
\end{assumption}

Part (i) of Assumption \ref{AV_assumptions} imposes stronger conditions over the serial dependence of the shocks. This assumption is common for theoretical analysis of the asymptotic refinement of the bootstrap method in autoregressive models. An incomplete list of previous research that uses this assumption includes \cite{bose1988}, \cite{park2003bootstrap,park2006bootstrap}, and \cite{mikusheva2015second}. Parts (ii) and (iii) of Assumption \ref{AV_assumptions} are sufficient technical conditions on the distribution of the shocks $P$ to establish the existence of the Edgeworth expansions presented in \eqref{eq:edgeworth_expansion_sec2}-\eqref{eq:edgeworth_expansion_boots_sec2}. Part (ii) implies that the distribution $J_n(\cdot,h,P,\rho)$ defined in \eqref{eq:cdf_root} is continuous and guarantees that a data-dependent version of the Cram\'er condition holds, which is a common condition to guarantee the existence of Edgeworth expansions; see Remark \ref{rem:GH1994} for further discussion. 
Part (iii) implies that any sufficiently large number of moments exist and are uniformly bounded by a function of the constant $c_u$, which is important to guarantee the Edgeworth expansion for the bootstrap distribution $J_n(\cdot,h,\hat{P}_n,\hat{\rho}_n)$. Although this condition is strong, it is not atypical in the literature of the asymptotic refinement of the bootstrap method with dependent data; for instance, \cite{hall1996bootstrap} and \cite{inoue2006bootstrapping} assume the existence of 33rd and 36th moments, respectively, while \cite{andrews2002higher} assumes that all the \textcolor{black}{moments exist}.

We rely on Assumption \ref{AV_assumptions}, the approach and results presented in \cite{bhattacharya1978validity} and \cite{bhattacharya1987some}, and the general framework developed by \cite{GotzeHipp1983} to prove the existence of Edgeworth expansions with dependent data. The framework of \cite{GotzeHipp1983} requires weakly dependent data and verifying stronger regularity conditions than the ones needed in the case of i.i.d. data; see \cite{Hall1992TheBA} and \cite{lahiri2003resampling} for textbook references. Therefore, we restrict our analysis to data-generating processes with weak dependence (e.g., $|\rho| \le 1-a$ for some $a \in (0,1)$) in a similar way to previous research on asymptotic refinements involving dependent data that includes \cite{bose1988}, \cite{hall1996bootstrap},  \cite{lahiri1996edgeworth}, \cite{andrews2002higher, andrews2004block}, and \cite{inoue2006bootstrapping}. It is an open question whether there exist Edgeworth expansions as in \eqref{eq:edgeworth_expansion_sec2}-\eqref{eq:edgeworth_expansion_boots_sec2} for the case \textcolor{black}{of local-to-unity or unit-root models}. See Remark \ref{rem:mikusheva2015} for further discussion on alternative methods and available results.

\begin{theorem}\label{thm:rates_AA}
 Suppose Assumption \ref{AV_assumptions} \textcolor{black}{holds}. Fix a given $h \in \textbf{N}$ and $a \in (0,1)$. Then, for any $\rho \in [-1+a,1-a]$, 
 we have
 %
\begin{equation}\label{eq:thm_rates_AA_2}
   | P_{\rho} \left( \beta(\rho,h) \in C_n(h,1-\alpha) \right) - (1-\alpha) | = O(n^{-1}) 
\end{equation}
where $\beta(\rho,h)$ is as in \eqref{eq:beta} and $C_n(h,1-\alpha)$ is as in \eqref{eq:CI_MOPM}.

\end{theorem}


The ECP of  $C_n(h,1-\alpha)$ has a similar size as the one derived in our informal explanation in Section \ref{sec:why_boots}. Similar sizes of the ECP were obtained for symmetrical confidence intervals in the i.i.d. data case; see \cite{Hall1992TheBA} and \cite{Horowitz2001, Horowitz2019}. 
 
The proof of Theorem \ref{thm:rates_AA} is presented in Appendix \ref{appendix:thm:rates_AA}. It uses two main ideas developed previously in the literature. First, we approximate the distribution $J_n(\cdot,h,P,\rho)$ by another distribution $\Tilde{J}_n(\cdot,h,P,\rho)$ up to an error of size $O\left(n^{-1-\epsilon}\right)$ for a fixed $\epsilon \in (0,1/2)$; similar approach has been used in \cite{hall1996bootstrap} and \cite{andrews2002higher, andrews2004block}. Second, we use that the distribution $\Tilde{J}_n(\cdot,h,P,\rho)$ admits an Edgeworth expansion up to an error of size $O\left(n^{-3/2}\right)$ based on the results of \cite{bhattacharya1978validity} and \cite{GotzeHipp1983,gotze1994asymptotic}; see Theorem \ref{thm:edgeworth_approx_Jtilde} in Appendix \ref{sec:appendix_asympt_ref}. These two ideas guarantee the existence of the Edgeworth expansion presented in \eqref{eq:edgeworth_expansion_sec2}. We then conclude the proof by standard derivations similar to the one derived in our informal explanation presented in Section \ref{sec:why_boots}.

The next theorem shows that the LP-residual bootstrap provides asymptotic refinements to the confidence intervals. In other words, the size of the ECP of our bootstrap confidence interval defined in \eqref{eq:bootstrap_CI} for $\beta(\rho,h)$ is $o(n^{-1})$.
 
\begin{theorem}\label{thm:rates_bootstrap}
 Suppose Assumption \ref{AV_assumptions} \textcolor{black}{holds}. Fix a given $h \in \textbf{N}$ and $a \in (0,1)$. Then, for any $\rho \in [-1+a,1-a]$ and $\epsilon \in (0,1/2)$, we have 
\begin{equation}\label{eq:thm_rates_boots_2}
   | P_{\rho} \left( \beta(\rho,h) \in C_n^*(h,1-\alpha) \right) - (1-\alpha) | = o \left( n^{-(1+\epsilon)} \right)  ~,
\end{equation}
where $\beta(\rho,h)$ is as in \eqref{eq:beta} and $C_n^*(h,1-\alpha)$ is as in \eqref{eq:bootstrap_CI}. 
\end{theorem}

Theorem \ref{thm:rates_bootstrap} presents the size of the ECP of the confidence interval $C_n^*(h,1-\alpha)$ in \eqref{eq:thm_rates_boots_2}. This is similar to the one derived in our informal explanation in Section \ref{sec:why_boots}, but it is typically larger than those obtained for the ECP of symmetrical confidence intervals using bootstrap methods in the i.i.d. data case; see \cite{Hall1992TheBA} and \cite{Horowitz2001, Horowitz2019}.
 
The proof of Theorem \ref{thm:rates_bootstrap} is presented in Appendix \ref{appendix:thm:rates_bootstrap}. It relies on two claims: the existence of the Edgeworth expansion for the distribution $J_n(\cdot,h,P,\rho)$ and the existence of constants $C_1$ and $C_2$ such that $P_{\rho}(|\Delta_n| > C_1 n^{-(1+\epsilon)} ) \le C_2 n^{-(1+\epsilon)} $, where $\Delta_n = c_n^*(h,1-\alpha)-c_n(h,1-\alpha)$, and $c_n(h,1-\alpha)$ and $c_n^*(h,1-\alpha)$ are defined in \eqref{eq:cv_finite} and \eqref{eq:cv_bootstrap}, respectively. We next sketch the proof based on those two claims. We can derive
\begin{align*}
    P_{\rho} \left( \beta(\rho,h) \in C_n^*(h,1-\alpha) \right)  &= P_{\rho} \left( |R_n(h)| \le c_n^*(h,1-\alpha) \right) \\
    &= P_{\rho} \left( |R_n(h)| \le c_n(h,1-\alpha) + \Delta_n, \left| \Delta_n \right| \le C_1 n^{-(1+\epsilon)} \right) + O\left(  n^{-(1+\epsilon)} \right)~\\
    &= 1-\alpha +  O\left(  n^{-(1+\epsilon)} \right)~,    
\end{align*}
where the last equality follows from the existence of the Edgeworth expansion for the distribution $J_n(\cdot,h,P,\rho)$ (our first claim), which implies
$$ P_{\rho} \left( |R_n(h)| \le c_n(h,1-\alpha) + O\left(  n^{-(1+\epsilon)} \right) \right) = 1-\alpha +  O\left(  n^{-(1+\epsilon)} \right)~.$$
Note that the first claim follows from Theorem \ref{thm:rates_AA}. To prove our second claim, we first show that there is an event $E_n$ such that (i) $J_n(\cdot,h,\hat{P}_n,\hat{\rho}_n)$ has an Edgeworth expansion as in \eqref{eq:edgeworth_expansion_boots_sec2} conditional on $E_n$ and (ii) the probability of the complement of $E_n$ is equal to $O\left(n^{-(1+\epsilon)}\right)$ for any $\epsilon \in (0,1/2)$; see Lemma \ref{lemma:high-prob-event} in Appendix \ref{appendixB:lemmas}. We then follow standard arguments in the literature to prove this claim. Finally, note that $O(n^{-(1+\epsilon)})$ for any $\epsilon \in (0,1/2)$ is equivalent to $o(n^{-(1+\epsilon)})$ for any $\epsilon \in (0,1/2)$, which is the error stated in Theorem \ref{thm:rates_bootstrap}.

\begin{remark}\label{rem:block-bootstrap}
The bootstrap methods proposed in \cite{hall1996bootstrap} and \cite{andrews2002higher} can be adapted for the construction of confidence intervals for the impulse response $\beta(h,\rho)$ defined in \eqref{eq:beta}. Four points based on their framework and results deserve further discussion. First, their bootstrap method consists of the nonoverlapping block bootstrap scheme (\cite{carlstein1986use}) and overlapping block bootstrap (\cite{kunsch1989jackknife}). Second, they show that their bootstrap methods provide asymptotic refinements to the critical values of $t$-tests based on generalized method of moments (GMM) estimators $\hat{\theta}_T$ and weakly dependent data $\{Z_t:1 \le t \le n\}$. One of their main conditions is that the series of moment functions $\{g(Z_t,\theta): t \ge 1 \}$ are uncorrelated beyond some finite lags, i.e. for some $\kappa >0$ we have $E[g(Z_t,\theta)g(Z_s,\theta)'] = 0$ for any $t,s\ge 1$ such that $|t-s|>\kappa$. %
Third, the LP estimator $\hat{\beta}_n(h)$ defined in \eqref{eq:la_lp} can be presented as a GMM estimator using the following dependent data $\{Z_t = (y_{t-1}, y_t, y_{t+h}) : 1 \le t \le n\}$ and moment function: $g(y_{t+h},x_t,\theta)=(y_{t+h}-\theta x_t)x_t$, where $x_t = (y_t, y_{t-1})'$. Then, we can invoke their results and use their bootstrap methods but only for the case of $|\rho|<1$ and under additional assumptions. Note that their main condition can be verified with $\kappa = h$. Fourth, we can construct confidence intervals for $\beta(\rho,h)$ based on their asymptotic distribution theory.
\end{remark} 
 
\begin{remark}\label{rem:block-bootstrap-2}
    As we mentioned in Remark \ref{rem:block-bootstrap}, we can use the bootstrap methods presented in \cite{hall1996bootstrap} and \cite{andrews2002higher, andrews2004block} to construct confidence intervals for $\beta(\rho,h)$ since the LP estimator $\hat{\beta}_n(h)$ defined in \eqref{eq:la_lp} can be presented as a GMM estimator. Their results provide sizes of the ECP of these confidence intervals that are qualitatively similar to the one found in Theorem \ref{thm:rates_bootstrap}. 
\end{remark}

\begin{remark}
The size of the ECP of $C_{per-t,n}^{*}(h,1-\alpha)$ is $O(n^{-1})$. We presented and discussed  the equal-tailed percentile-t confidence interval $C_{per-t,n}^{*}(h,1-\alpha)$ in Remark \ref{rem:percentile-t_bootstrap}. To compute the size of its ECP, we can use the existence of the Edgeworth expansions presented in \eqref{eq:edgeworth_expansion_sec2}-\eqref{eq:edgeworth_expansion_boots_sec2} and Theorem 5.2 in \cite{Hall1992TheBA}. The size of the  ECP of $C_{per-t,n}^{*}(h,1-\alpha)$ is similar to the one obtained in \eqref{eq:thm_rates_AA_2} for the ECP of $C_n(h,1-\alpha)$; therefore, the LP-residual bootstrap does not provide asymptotic refinement for equal-tailed percentile-t confidence intervals. Similar conclusions were obtained for the case of i.i.d. data; see \cite{Hall1992TheBA} and \cite{Horowitz2001, Horowitz2019}.
\end{remark}

\begin{remark}\label{rem:GH1994}
We use part (ii) of Assumption \ref{AV_assumptions} to verify that a dependent-data version of the Cramer condition required in \cite{GotzeHipp1983} holds, which is an important condition for the existence of the Edgeworth expansion in the dependent-data case. However, verifying that condition is quite difficult in general, as pointed out by \cite{hall1996bootstrap} and \cite{gotze1994asymptotic}, among others. Therefore, we proceed in two steps based on the results by \cite{gotze1994asymptotic} that propose simple and verifiable conditions to guarantee the conditions required by \cite{GotzeHipp1983}, including the dependent-data version of the Cramer condition. We first approximate the distribution $J_n(\cdot,h,P,\rho)$ by a distribution $\Tilde{J}_n(\cdot,h,P,\rho)$. We then use part (ii) of Assumption 5.1 to verify the conditions required in Theorem 1.2 of \cite{gotze1994asymptotic}, which guarantee the existence of Edgeworth expansion for the distribution $\Tilde{J}_n(\cdot,h,P,\rho)$.
\end{remark}  

\begin{remark}\label{rem:var-stationar-asymp-refinement}
    For strictly stationary data-generating processes, the results in Theorems \ref{thm:rates_AA} and \ref{thm:rates_bootstrap} can be extended to the family of vector autoregressive (VAR) models that satisfy similar assumptions to the ones presented in Assumption \ref{AV_assumptions}, which are stronger than Assumptions 1 and 2 in \cite{MO-PM2020}. These extensions can be shown by verifying the conditions required in \cite{gotze1994asymptotic}. We leave the details of a formal proof for the VAR models for future research.
\end{remark}

\begin{remark}\label{rem:mikusheva2015}
    An alternative method for asymptotically approximating a finite sample distribution is the stochastic embedding and strong approximation principle used in \cite{park2003bootstrap,park2006bootstrap} and \cite{mikusheva2015second}. Using this method in the local-to-unit asymptotic framework for the AR(1) model, \cite{mikusheva2015second} showed that the grid bootstrap version of the t-statistic approximates its finite sample distribution up to an error of size $o(n^{-1/2})$. It is an open question whether these techniques can be adapted to show that LP-residual bootstrap provides asymptotic refinements to the confidence intervals when $\rho=1$. 
\end{remark}

\section{Simulation Study}\label{sec:simulation_study}
 
We examine the finite sample performance of $C_n^*(h,1-\alpha)$ defined in \eqref{eq:bootstrap_CI} using different data-generating processes. We consider a sample size $n = 95$, which is the median sample size based on 71 papers that have utilized the LP approach; see \cite{herbst2024bias}. Additionally, we examine other confidence intervals presented in the paper.

\subsection{Monte Carlo Design}\label{sec:sim_design}

We use four designs for the distribution of the shocks $\{u_t:1 \le t \le n\}$ and two values for the parameter $\rho \in \{0.95, 1\}$ in our Monte Carlo simulation. The shocks are defined according to the GARCH(1,1) model:
\begin{equation*}
    u_t = \tau_t v_t, \quad \tau_t^2 = \omega_0 + \omega_1 u_{t-1}^2 + \omega_2  \tau_{t-1}^2, \quad v_t \text{ are } i.i.d.~,
\end{equation*}
where the distribution of $v_t$ and the parameter vector $(\omega_0,\omega_1,\omega_2)$ are specified as follows:
\begin{itemize}
\item [\textbf{Design 1:}]  $v_t \sim N(0,1)$, $\omega_0=1$, and $\omega_1=\omega_2=0$. 

\item [\textbf{Design 2:}]  $v_t \sim N(0,1)$, $\omega_0=0.05$, $\omega_1= 0.3$, and $\omega_2=0.65$.

\item [\textbf{Design 3:}]  $v_t \sim t_4/\sqrt{2}$, $\omega_0=1$, and $\omega_1=\omega_2=0$. 

\item [\textbf{Design 4:}]  $v_t|B_t=j \sim N(m_j,\sigma_j^2)$, where $B_t \in \{0,1\}$, $B_t = 1$ with probability $p = 0.25$, $m_0 = 2/\sigma_2$, $m_1 =-6/\sigma_2$, $\sigma_0 = 0.5/\sigma_2$, $\sigma_1 = 2/\sigma_2$, and $\sigma_2^2 = p(m_1^2+\sigma_1)+(1-p)(m_0^2+\sigma_0)$,  $\omega_0=0.05$, $\omega_1= 0.3$, and $\omega_2=0.65$.
\end{itemize}

We consider nine different confidence intervals for each design and each value of $\rho$. All our confidence intervals use the HC standard errors $\hat{s}_n(h)$ defined in \eqref{eq:HC_se}. Additionally, we consider alternative HC standard errors $\hat{s}_{j,n}(h)$ defined as
\begin{equation*}
     \hat{s}_{j,n}(h) \equiv \left( \sum_{t=1}^{n-h} \hat{u}_t(h)^2 \right)^{-1/2}  \left(\sum_{t=1}^{n-h} \hat{\xi}_{j,t}(h)^2 \hat{u}_t(h)^2 \right)^{1/2} \left( \sum_{t=1}^{n-h} \hat{u}_t(h)^2 \right)^{-1/2}  ~,
\end{equation*}
for $j=2,3$, where $\hat{\xi}_{2,t}(h)^2 = \hat{\xi}_{t}(h)^2/(1-\mathbb{P}_{h,tt})$ and $\hat{\xi}_{3,t}(h)^2 = \hat{\xi}_{t}(h)^2/(1-\mathbb{P}_{h,tt})^2$. We use the projection matrix $\mathbb{P}_h = \mathbb{X}_h (\mathbb{X}_h'\mathbb{X}_h)^{-1}\mathbb{X}_h'$, where $\mathbb{X}_h$ is a matrix with row elements equal to $(\hat{u}_t(h),~ y_{t-1})$ for $t=1,\ldots,n-h$. The confidence intervals that we use are listed below\textcolor{black}{.}

\begin{enumerate}
\item \textbf{RB:} confidence interval as in \eqref{eq:bootstrap_CI} based on the LP-residual bootstrap.

\item \textbf{$\text{RB}_{per-t}$:}  equal-tailed percentile-t confidence interval as in \eqref{eq:CI_RB-per-t}. It is based on the LP-residual bootstrap and discussed in Remark \ref{rem:percentile-t_bootstrap}.

\item \textbf{$\text{RB}_{hc3}$:}  confidence interval as in \eqref{eq:bootstrap_CI} but using $\hat{s}_{3,n}(h)$ and $c_{3,n}^*(h,1-\alpha)$ instead of $\hat{s}_{n}(h)$ and $c_n^*(h,1-\alpha)$, where $c_{3,n}^*(h,1-\alpha)$ is computed as in Section \ref{sec:bootstrap_cv} but using $\hat{s}_{3,n}^*(h)$ instead of $\hat{s}_{n}^*(h)$.

\item \textbf{WB:}  confidence interval as in \eqref{eq:bootstrap_CI} but using $c_n^{wb,*}(h,1-\alpha)$ instead of $c_n^*(h,1-\alpha)$, where $c_n^{wb,*}(h,1-\alpha)$ is based on the LP-wild bootstrap; see Remark \ref{rem:wild-bootstrap}.

\item \textbf{$\text{WB}_{per-t}$:}  equal-tailed percentile-t confidence interval as in \eqref{eq:CI_RB-per-t} but using $q_n^{wb,*}(h,\alpha_0)$ instead of $q_n^*(h,\alpha_0)$, where $q_n^{wb,*}(h,\alpha_0)$ is based on the LP-wild bootstrap discussed in Remark \ref{rem:wild-bootstrap}.

\item \textbf{$\text{GB}_{LR}$:} confidence interval based on the grid bootstrap presented in Section 3.3 in \cite{Mikusheva2012}. It uses the LR statistic.

\item \textbf{AA:} standard confidence interval as in \eqref{eq:CI_MOPM}. 

\item \textbf{$\text{AA}_{hc2}$:} standard confidence interval as in \eqref{eq:CI_MOPM} but using $\hat{s}_{2,n}(h)$ instead of $\hat{s}_{n}(h)$.

\item \textbf{$\text{AA}_{hc3}$:} standard confidence interval as in \eqref{eq:CI_MOPM} but using $\hat{s}_{3,n}(h)$ instead of $\hat{s}_{n}(h)$.

\end{enumerate}

\subsection{Discussion and Results}

In all the designs, the shocks have zero mean and variance one. Designs 1-2 verify Assumption \ref{xu_assumptions} presented in Section \ref{sec:uniform_cons}. Design 1 also verifies Assumption \ref{A3:MO-PM} due to Proposition \ref{proposition:A3} in Appendix \ref{sec:appendix_uniform}. Assumption \ref{A3:MO-PM} can be tedious to verify in general since it involves computing a probability for all the parameters $\rho$ in the parameter space and taking their infimum. In contrast, designs 3-4 do not verify all the parts of Assumption \ref{xu_assumptions}. Design 3 considers shocks without a fourth moment, i.e., it does not verify part (iv) of Assumption \ref{xu_assumptions}, which was a regularity condition. Design 4 considers a distribution of the shocks (GARCH errors with asymmetric $v$ and nonzero skewness) that lie outside the class of conditional heteroskedastic processes that we consider in this paper, i.e., it does not verify part (ii) of Assumption \ref{xu_assumptions}. As we discussed in Remark \ref{rem:hc_se}, part (ii) of Assumption \ref{xu_assumptions} was a sufficient condition for the validity of the HC standard errors $\hat{s}_n(h)$ in the construction of confidence intervals.  
  
\begin{table}
\begin{center}
\setlength{\tabcolsep}{6.0pt}  
\renewcommand{\arraystretch}{0.9}  
\begin{tabular}{cc ccccccccc}
\hline 
\multicolumn{1}{c}{$\rho$} &  $h$ & $\text{RB}$ & $\text{RB}_{per-t}$ & $\text{RB}_{hc3}$ & $\text{WB}$ & $\text{WB}_{per-t}$ & $\text{GB}_{LR}$ & $\text{AA}$ & $\text{AA}_{hc2}$ & $\text{AA}_{hc3}$  \\
\hline
 \multicolumn{11}{c}{Design 1: Gaussian i.i.d. shocks}  \\
\hline
 0.95 &     1 & 90.04 & 89.60 & 90.08 & 90.38 & 90.32 & 90.38 & 88.26 & 89.12 & 89.60\\
  &     6 & 89.36 & 88.98 & 89.38 & 90.46 & 90.22 & 90.38 & 85.00 & 85.58 & 86.44\\
  &    12 & 88.12 & 86.96 & 88.08 & 89.60 & 88.28 & 90.38 & 83.78 & 84.44 & 85.34\\
  &    18 & 87.96 & 86.08 & 87.88 & 89.46 & 88.08 & 90.38 & 84.44 & 85.16 & 85.86\\
\hline
 1.00 &     1 & 90.20 & 89.80 & 90.30 & 90.48 & 90.34 & 89.46 & 88.30 & 88.90 & 89.66\\
  &     6 & 89.80 & 89.44 & 89.80 & 90.68 & 90.22 & 89.46 & 83.54 & 84.42 & 85.28\\
  &    12 & 87.92 & 87.60 & 87.90 & 88.78 & 89.02 & 89.46 & 80.32 & 81.30 & 81.94\\
  &    18 & 86.22 & 84.76 & 86.22 & 87.02 & 86.36 & 89.46 & 78.34 & 79.16 & 79.98\\
\hline
 \multicolumn{11}{c}{Design 2: Gaussian GARCH shocks}  \\
\hline
 0.95 &     1 & 88.86 & 89.00 & 89.40 & 90.18 & 90.02 & 85.94 & 86.84 & 88.10 & 89.16\\
  &     6 & 87.94 & 88.00 & 88.26 & 90.12 & 90.74 & 85.94 & 83.64 & 84.52 & 85.60\\
  &    12 & 87.08 & 85.72 & 87.28 & 88.72 & 88.18 & 85.94 & 82.96 & 83.90 & 84.88\\
  &    18 & 86.36 & 84.36 & 86.40 & 87.98 & 86.94 & 85.94 & 82.76 & 83.44 & 84.38\\
\hline
 1.00 &     1 & 88.64 & 88.82 & 89.14 & 89.96 & 89.94 & 88.38 & 86.72 & 87.84 & 88.90\\
  &     6 & 88.96 & 88.52 & 89.08 & 90.76 & 90.96 & 88.38 & 82.34 & 83.76 & 84.52\\
  &    12 & 86.64 & 86.08 & 86.60 & 88.56 & 88.68 & 88.38 & 79.14 & 80.46 & 81.32\\
  &    18 & 84.90 & 83.74 & 84.78 & 86.56 & 86.52 & 88.38 & 76.64 & 77.74 & 78.70\\
\hline
\end{tabular}
 \caption{ \small Coverage probability (in \%) of confidence intervals for $\beta(\rho,h)$ with a nominal level of 90\% and $n = 95$. 5,000 simulations and 1,000 bootstrap iterations.}\vspace{-0.5cm}\label{table1}
\end{center} 
\end{table}

Tables \ref{table1} and \ref{table2} report the coverage probabilities (in \%) of our simulations. Columns are labeled as the confidence intervals we specified in Section \ref{sec:sim_design}. For all the designs on the distribution of the shock and values of $\rho$, we use 5000 simulations to generate data with a sample size $n = 95$ based on the AR(1) model \eqref{eq:ar1_model}. In each simulation, we compute the nine confidence intervals described above for horizons $h \in \{1,6,12,18\}$. The confidence intervals have a nominal level equal to $1-\alpha = 90\%$. The bootstrap critical values are computed using $B=1000$ as described in Remark \ref{rem:cv_bootstrap}. We summarize our findings from the simulations below.

\begin{table}[!t]
\begin{center}
\setlength{\tabcolsep}{6.0pt} 
\renewcommand{\arraystretch}{0.9} 
\begin{tabular}{cc ccccccccc}
\hline 
\multicolumn{1}{c}{$\rho$} &  $h$ & $\text{RB}$ & $\text{RB}_{per-t}$ & $\text{RB}_{hc3}$ & $\text{WB}$ & $\text{WB}_{per-t}$ & $\text{GB}_{LR}$ & $\text{AA}$ & $\text{AA}_{hc2}$ & $\text{AA}_{hc3}$  \\
\hline
 \multicolumn{10}{c}{Design 3: t-student i.i.d. shocks}  \\
\hline
 0.95 &     1 & 90.00 & 90.08 & 90.36 & 90.52 & 90.32 & 90.06 & 88.04 & 89.24 & 90.26\\
  &     6 & 89.08 & 88.48 & 89.28 & 89.76 & 89.64 & 90.06 & 84.04 & 85.40 & 86.66\\
  &    12 & 87.74 & 86.18 & 87.90 & 88.46 & 87.42 & 90.06 & 82.78 & 84.24 & 85.46\\
  &    18 & 88.08 & 85.38 & 88.26 & 89.12 & 87.52 & 90.06 & 83.36 & 84.80 & 86.20\\
\hline
 1.00 &     1 & 89.96 & 89.88 & 90.16 & 90.36 & 89.98 & 90.44 & 87.74 & 88.82 & 90.16\\
  &     6 & 89.78 & 88.60 & 89.84 & 90.52 & 89.84 & 90.44 & 82.88 & 84.54 & 85.78\\
  &    12 & 87.56 & 86.82 & 87.64 & 88.40 & 88.22 & 90.44 & 79.04 & 80.30 & 81.56\\
  &    18 & 85.64 & 84.40 & 86.00 & 86.80 & 86.24 & 90.44 & 77.50 & 78.84 & 80.22\\
\hline
 \multicolumn{10}{c}{Design 4: mixture-of-gaussian GARCH shocks}  \\
\hline
 0.95 &     1 & 89.00 & 89.86 & 89.32 & 88.80 & 89.60 & 87.82 & 86.38 & 87.20 & 87.88\\
  &     6 & 87.90 & 90.62 & 88.14 & 89.12 & 92.04 & 87.82 & 84.30 & 85.30 & 86.18\\
  &    12 & 84.14 & 86.64 & 84.00 & 85.58 & 87.98 & 87.82 & 80.70 & 81.52 & 82.32\\
  &    18 & 83.48 & 84.70 & 83.66 & 85.32 & 86.88 & 87.82 & 80.46 & 81.40 & 82.56\\
\hline
 1.00 &     1 & 88.84 & 90.24 & 89.04 & 88.98 & 89.70 & 88.02 & 86.60 & 87.24 & 88.00\\
  &     6 & 88.24 & 91.26 & 88.50 & 89.62 & 92.66 & 88.02 & 82.78 & 83.82 & 84.64\\
  &    12 & 84.96 & 88.54 & 85.08 & 86.74 & 89.86 & 88.02 & 77.40 & 78.32 & 79.50\\
  &    18 & 82.30 & 84.62 & 82.34 & 83.90 & 86.30 & 88.02 & 74.18 & 75.28 & 76.14\\
\hline
\end{tabular}
 \caption{\small Coverage probability (in \%) of confidence intervals for $\beta(\rho,h)$ with a nominal level of 90\% and $n = 95$. 5,000 simulations and 1,000 bootstraps iterations.}\vspace{-0.5cm}\label{table2}
\end{center} 
\end{table}

Five features of Table \ref{table1} deserve discussion. First, it shows that our recommended confidence interval \textbf{RB} has a coverage probability closer to 90\% than the confidence intervals \textbf{AA}, \textbf{$\text{AA}_{hc2}$}, and \textbf{$\text{AA}_{hc3}$} for all the designs 1-2,  values of $\rho$, and horizons $h$, with some few exceptions. The lowest coverage probability of \textbf{RB}, \textbf{AA}, \textbf{$\text{AA}_{hc2}$}, and \textbf{$\text{AA}_{hc3}$} are $85\%$, $77\%$, $78\%$, and $79\%$, respectively, and occur when $\rho=1$ and horizon $h=18$. Second, \textbf{RB} and \textbf{$\text{RB}_{hc3}$} have better performance than \textbf{$\text{RB}_{per-t}$}, especially when $\rho = 1$ and the horizon is a significant fraction of the sample size ($h \in \{12,18\}$). Third, \textbf{WB} and \textbf{$\text{WB}_{per-t}$} have larger coverage probability than \textbf{RB} for all the designs 1-2,  values of $\rho$, and horizons $h$, with some few exceptions. The larger coverage of \textbf{WB} and \textbf{$\text{WB}_{per-t}$} is associated with a larger median length of their confidence intervals, as we reported in Table \ref{table1a} in \textcolor{black}{Online Supplemental} Appendix \ref{sec:appendix_table}. Fourth, \textbf{$\text{AA}_{hc3}$} presents a coverage probability closer to $90\%$ and larger than \textbf{AA} and \textbf{$\text{AA}_{hc2}$} for all the designs 1-2,  values of $\rho$, and horizons $h$. This finding suggests that using $\hat{s}_{3,n}(h)$ instead of $\hat{s}_n(h)$ can improve the coverage probability of the confidence interval; however, confidence intervals based on bootstrap methods (e.g., \textbf{RB} and \textbf{$\text{WB}_{per-t}$}) report coverage probability closer to $90\%$. Fifth, \textbf{$\text{GB}_{LR}$} has a coverage probability close to 90\% on design 1  (i.i.d. shocks), while it has some distortions on design 2 that are larger on $\rho = 0.95$. As we mentioned in Remark \ref{rem:grid_bootstrap}, it is unknown if the grid bootstrap is valid for design 2. The coverage probability of \textbf{$\text{GB}_{LR}$} is constant across horizons because the LR statistic is invariant to monotonic transformations; see Section 4.3 and footnote 6 on \cite{Mikusheva2012} for more details.
 
Table \ref{table2} presents results for designs 3-4. Our findings for design 3 are qualitatively similar to Table \ref{table1}, which was discussed above. This suggests that failing part (iv) of Assumption \ref{xu_assumptions} (a regularity condition) does not have a major effect on the coverage probability of the confidence intervals that we considered. In contrast, design 4 shows that some of our qualitative findings can change if we fail to verify part (ii) of Assumption \ref{xu_assumptions}. This result is consistent with existing theory since this assumption was a sufficient condition for the validity of confidence intervals that use HC standard errors $\hat{s}_n(h)$; see Remark \ref{rem:hc_se}. In particular, \textbf{$\text{RB}_{per-t}$} has a coverage probability closer to $90\%$ and larger than  \textbf{RB} and \textbf{$\text{RB}_{hc3}$}. The small sample size ($n=95$) does not explain the findings for design 4. We obtain similar results for a sample size $n=240$ in Table \ref{table3a} in \textcolor{black}{Online Supplemental} Appendix \ref{sec:appendix_table}.

Finally, Table \ref{table4a} in \textcolor{black}{Online Supplemental} Appendix \ref{sec:appendix_table} reports the statistical power of the confidence intervals specified in Section \ref{sec:sim_design}. Here, we refer by statistical power to the coverage probabilities (in \%) of (size-adjusted) confidence intervals for parameters different than the true one. In this sense, a low coverage probability of a confidence interval is desirable. We find all the confidence intervals have coverage probability around $80\%$ on horizon $h=1$ and designs 1, 2, and 3, which suggests they have statistical power at $h=1$. We also notice that \textbf{$\text{RB}_{per-t}$}, \textbf{$\text{WB}_{per-t}$} and \textbf{$\text{GB}_{LR}$} have a coverage probability strictly lower than $90\%$ for horizon $h=6$ and designs 1, 2, and 3. Moreover, they have a lower coverage \textcolor{black}{probability} than all the other confidence intervals. Finally, all the confidence intervals have coverage \textcolor{black}{probability} above $90\%$ on design 4, with the exception of \textbf{$\text{GB}_{LR}$} for horizon $h=1$.

\section{LP-Residual Bootstrap for VAR Models}\label{sec:lp-residual-bootstrap_VAR}


This section describes the LP-residual bootstrap method to construct confidence intervals for a scalar function of impulse responses of VAR(p) models, where $p$ denotes the number of lags. More concretely, we propose the confidence interval in \eqref{eq:bootstrap_CI_VAR} for  $\nu' \beta_{h,i}$, where $\beta_{h,i} \in \mathbf{R}^k$ is the vector containing all the impulse response coefficients of the reduced-form shocks in the variable $i$  at $h$ periods in the future. Here, $\nu \in \mathbf{R}^k\setminus \{ 0 \}$ is a user-specified vector, e.g., $\nu = e_j$ (the j-th unit vector) implies $\nu' \beta_{h,i}$ is the impact of the $j$-th reduced-form shock in the variable $i$ at $h$ periods in the future. 

%
%
%



The confidence interval for $\nu' \beta_{h,i}$ is defined as
\begin{equation}\label{eq:bootstrap_CI_VAR}
    C_n^{*}(h,1-\alpha) \equiv \left[ \nu'\hat{\beta}_{i,n}(h) - c_n^*(h,1-\alpha)~\hat{s}_{i,n}(h,\nu),~\nu' \hat{\beta}_{i,n}(h) + c_n^*(h,1-\alpha)~ \hat{s}_{i,n}(h,\nu)\right] ~,
\end{equation}
where $\hat{\beta}_{i,n}(h)$, $\hat{s}_{i,n}(h,\nu)$, and $c_n^*(h,1-\alpha)$ are defined in \eqref{eq:la_lp_VAR}, \eqref{eq:HC_se_VAR}, and \eqref{eq:cv_bootstrap_VAR}, respectively.  

Let $\{y_t \in \mathbf{R}^k: 1 \le t \le n\}$ be the available time-series data. Suppose the data have been demeaned. Denote $X_t = (y_{t-1}',\ldots,y_{t-p}')'$ for all $t=p+1,\ldots, n$. 
Let $\hat{\beta}_{i,n}(h)$ be obtained from an OLS regression between $y_{i,t+h}$ and $(y_t',X_t')$, 
\begin{equation}\label{eq:la_lp_VAR}
   y_{i,t+h} = \hat{\beta}_{i,n}(h)' y_t + \hat{\gamma}_{i,n}(h) X_t + \hat{\xi}_{i,t}(h) ~.
\end{equation} 
Let $ \hat{s}_{i,n}(h,\nu)$ be the standard error for $\nu'\hat{\beta}_{i,n}(h)$ defined by
\begin{equation}\label{eq:HC_se_VAR}
    \hat{s}_{i,n}(h,\nu) = \frac{1}{n-h-p} \left \{  \nu' \hat{\Sigma}(h)^{-1} \left( \sum_{t=p+1}^{n-h} \hat{\xi}_{i,t}(h)^2 \hat{u}_t(h) \hat{u}_t(h)'   \right) \hat{\Sigma}(h)^{-1} \nu \right\}^{1/2}~\textcolor{black}{,}
\end{equation}
where
$$ \hat{u}_t(h) = y_t - \hat{A}(h) X_t~, \quad \hat{A}(h) = \left( \sum_{t=p+1}^{n-h} y_t X_t' \right) \left( \sum_{t=p+1}^{n-h} X_t X_t' \right)^{-1} $$
and
$$  \hat{\Sigma}(h) = \frac{1}{n-h-p} \sum_{t=p+1}^{n-h}  \hat{u}_t(h) \hat{u}_t(h)'~ \textcolor{black}{.}$$
Finally, let $c_n^*(h, 1-\alpha)$ be the bootstrap critical value involving the following steps:

\begin{itemize}
   \item[\textbf{Step 1}:]   Estimate a VAR(p) model with the data $Y^{(n)}$ using linear regression, 
   $$ y_t = \hat{A}_n X_t + \hat{u}_t~, ~ t = p+1,\ldots, n$$
   where
   \begin{equation}\label{eq:VAR_hat}
       \hat{A}_n = \left( \sum_{t=p+1}^{n} y_t X_t' \right) \left( \sum_{t=p+1}^{n} X_t X_t' \right)^{-1}~,
   \end{equation}
   and compute the centered residuals
   \begin{equation}\label{eq:centered_residuals_VAR}
         \left\{ \Tilde{u}_t \equiv \hat{u}_t - \frac{1}{n-p}\sum_{t=p+1}^n \hat{u}_t : p+1 \le t \le n   \right \}~.
   \end{equation}

   \item[\textbf{Step 2}:] Generate $B$ new samples of size $n$ using \eqref{eq:VAR_hat} and \eqref{eq:centered_residuals_VAR}. Define the sample as 
   \begin{equation*}
       y_{b,t}^* = \sum_{\ell=1}^p  \hat{A}_{n,\ell} ~ y_{b,t-\ell}^* + u_{b,t}^*~, \quad  t= p + 1,\ldots,n~,
   \end{equation*}
   where the initial p observations $(y_{b,1}^*,\ldots,y_{b,p}^*)$ are \textcolor{black}{drawn} at random from the $n-p+1$ blocks of $p$ consecutive observations in the original data. Here, $\hat{A}_n = (\hat{A}_{n,1},\ldots, \hat{A}_{n,p})$ are matrices estimated in \eqref{eq:VAR_hat} and $\{u_{b,t}^* : 1 \le t \le n\}$ is a random sample from the empirical distribution of the centered residuals defined in \eqref{eq:centered_residuals_VAR}. 
   The new sample $\{ y_{b,t}^* : 1 \le t \le n \} $ is called the bootstrap sample.

   \item[\textbf{Step 3}:] Compute $\hat{\beta}_{b,i,n}^*(h)$ and $\hat{s}_{b,i,n}^*(h)$  as in \eqref{eq:la_lp_VAR} and \eqref{eq:HC_se_VAR} using the lag-augmented LP regression and the bootstrap sample $\{ y_{b,t}^* : 1 \le t \le n \} $ for each $b=1,\ldots, B$. Define  
   \begin{equation*}\label{eq:root_bootstrap_VAR}
       R_{b,n}^*(h,\nu) =   \frac{\nu'\hat{\beta}_{b,i,n}^*(h) - \nu'\beta_i(\hat{A}_n,h)}{\hat{s}_{b,i,n}^*(h,\nu)} ~,~ b=1,\ldots, B
   \end{equation*}
   where $\beta_i(A,h) \in \mathbf{R}^k$ is the impulse response of all \textcolor{black}{reduced-form} shocks in the variable $i$ at horizon $h$ implied by the VAR(p) model with coefficients $A= (A_1,\ldots,A_p)$. Here, $\hat{A}_n$ is as in \eqref{eq:VAR_hat}.

   \item[\textbf{Step 4}:] Compute the $1-\alpha$ quantile of the $B$ draws of $R_{b,n}^*(h,\nu) $. Denote this by 
   %
   %
   \begin{equation}\label{eq:cv_bootstrap_VAR}
       c_n^{*}(h, 1-\alpha) \equiv \inf \left \{ u \in \textbf{R} : \frac{1}{B} \sum_{b=1}^B I\{ |R_{b,n}^*(h,\nu)| \le u  \} \ge 1-\alpha   \right\}~.
   \end{equation}
\end{itemize} 




The theoretical properties of the bootstrap confidence interval defined in \eqref{eq:bootstrap_CI_VAR} are unknown for general VAR models. However, Monte Carlo simulations presented in \textcolor{black}{Online Supplemental} Appendix \ref{appendix:simulations-var} suggest that confidence intervals based on the LP-residual bootstrap perform better in terms of coverage probability than those based on first-order asymptotic theory. Remarks \ref{rem:var-stationary} and \ref{rem:var-stationar-asymp-refinement} provide further discussion on how to extend some of the results presented in this paper to general VAR models.
 
\begin{remark}\label{rem:wild-bootstrap-var}
\cite{MO-PM2020} proposed a different bootstrap confidence interval for the impulse response coefficients of VAR(p) models. As we discussed in Remark \ref{rem:wild-bootstrap}, they use a wild bootstrap procedure ---which we refer to as the LP wild bootstrap--- to define the bootstrap shocks used to generate the bootstrap sample with an estimated VAR model (similar to Step 2 above). They use the LP wild bootstrap to construct equal-tailed percentile-t confidence intervals that differ from the symmetric percentile-t confidence intervals defined in \eqref{eq:bootstrap_CI_VAR}, which we recommend for the same reasons presented in Remark \ref{rem:bootstrap-equal-tailed} and based on our theoretical results for the AR(1) model. To our knowledge, the theoretical properties of the LP wild bootstrap procedure and the confidence intervals proposed by \cite{MO-PM2020} remain unknown. We include their recommended confidence intervals in the simulations presented in \textcolor{black}{Online Supplemental} Appendix \ref{appendix:simulations-var}.
\end{remark}
 
\section{Concluding Remarks}\label{sec:conclusion} 

This paper contributes to a growing literature on confidence interval construction for impulse response coefficients based on the local projection approach. Specifically, we propose the LP-residual bootstrap method to construct confidence intervals for the impulse response coefficients of AR(1) models at intermediate horizons. We prove two theoretical properties of this method: uniform consistency and asymptotic refinements. For a large class of AR(1) models that allow for a unit root, conditional heteroskedasticity of unknown form, and martingale difference shocks, we show that the proposed confidence interval $C_n^*(h,1-\alpha)$ defined in \eqref{eq:bootstrap_CI} has an asymptotic coverage probability equal to its nominal level $1-\alpha$ uniformly over the parameter space (e.g., $\rho \in [-1,1]$) and \textcolor{black}{a wide range of} intermediate horizons. For a restricted class of AR(1) models (e.g., $|\rho|\le \textcolor{black}{1-a}$ \textcolor{black}{where $a \in (0,1)$} and i.i.d. shocks with positive continuous density), we demonstrate that the error in coverage probability of $C_n^*(h,1-\alpha)$ has size $o(n^{-1})$, that is, the LP-residual bootstrap provides asymptotic refinements to the confidence intervals.

This paper considered the AR(1) model as the first step in understanding the theoretical properties of the LP-residual bootstrap. Three possible directions exist for future research. First, the uniform consistency of the LP-residual bootstrap method is an open question for the general vector auto-regressive (VAR) model. This bootstrap method is described in Section \ref{sec:lp-residual-bootstrap_VAR}. Second, the asymptotic refinement property of this method is unknown for the unit-root model ($\rho=1$) or general VAR models. Third, future work is needed to prove the uniform consistency of the LP-wild bootstrap discussed in Remark \ref{rem:wild-bootstrap}.

\appendix

\counterwithin{figure}{section}
\counterwithin{table}{section}
\renewcommand{\thetable}{\thesection.\arabic{table}}
\renewcommand{\thefigure}{\thesection.\arabic{figure}}
\renewcommand{\theequation}{A.\arabic{equation}}
\setcounter{equation}{0}

\section{Proofs of Result in Main Text}\label{sec:appendix_a}

\subsection{Proof of Theorem \ref{thm:consistency}}\label{appendix:thm:consistency}

We prove a stronger result:
$$  \sup_{|\rho|\le 1} P_{\rho}\left( \sup_{h \le h_n} ~ \sup_{|\tilde{\rho}|\le 1} ~ \sup_{x\in \mathbf{R}} | J_n(x,h,P,\tilde{\rho}) - J_n(x,h,\hat{P}_n,\hat{\rho}_n) | > \epsilon \right) \to 0 \quad \text{as} \quad n \to \infty~,$$
which is sufficient to conclude \eqref{eq:thm_consistency}. The proof has three steps.

\noindent \textbf{Step 1:} Let $E_{n,1} = \{ g(\rho,n) ~n^{1/2}~ | \hat{\rho}_n - \rho |  > M\}$, $E_{n,2} = \{ \left | n^{-1} \sum_{t=1}^n \Tilde{u}_t^2 - \sigma^2  \right| >  \sigma^2/2 \}$, and $E_{n,3} = \{  n^{-1} \sum_{t=1}^n \Tilde{u}_t^4  > \tilde{K}_4   \} $ be events, where $M$ and $\tilde{K}_4$ are \textcolor{black}{constants} defined next. Fix $\eta>0$. We use Lemma \ref{AppendixA:lemma_event_En} to guarantee the existence of $M$, $\tilde{K}_4$, and $N_0 = N_0(\eta)$ such that $ P_{\rho}(E_{n,j}) < \eta/3$ for $j=1,2,3$, $n \ge N_0$ and $\rho \in [-1,1]$. Define $E_n = E_{n,1}^c \cap E_{n,2}^c \cap E_{n,3}^c$. By construction $P_{\rho}(E_n ) > 1 - \eta$ for $n \ge N_0$ and for any $\rho \in [-1,1]$.
 
\noindent \textbf{Step 2:} Conditional on the event $E_n$, we have $ |\hat{\rho}_n - \rho| \le M n^{-1/2}/g(\rho,n)$ for $n \ge N_0$ and for any $\rho \in [-1,1]$. Therefore, conditional on the event $E_n$, we can use Lemma \ref{AppendixA:lemma_Mtilde} to conclude the existence of $\tilde{M}$ and $N_1 \ge N_0$ such that $ |\hat{\rho}_n| \le 1 + \tilde{M}/n$ for all $n \ge N_1$. Note also that conditional on the event $E_n$, we have that distribution $\hat{P}_n$ of the centered residuals defined in \eqref{eq:centered_residuals} verifies Assumption \ref{appendix:AV_assumptions} taking  $K_4 = M$, $\underline{\sigma} = \sigma^2/2$, and $\overline{\sigma} = 3\sigma^2/2$, i.e., $\hat{P}_n \in \mathbf{P}_{n,0}$, where $\mathbf{P}_{n,0}$ is defined in Appendix \ref{sec:appendix_uniform}.

\noindent \textbf{Step 3:} 
We use Theorem \ref{thm:sequence_models} taking $M = \Tilde{M}$. This implies that for any $\epsilon>0$, there exists $N_2 =N_2(\epsilon,\eta) \ge N_1$ such that
$      \sup_{x \in \mathbf{R}} ~ \left|  J_n(x,h,P_n,\rho) - \Phi(x)   \right| < \epsilon/2~,  $ 
 for any $n \ge N_2$, $|\rho| \le 1 + \tilde{M}/n$, $h \le h_n \le n$ and $h_n = o \left(n\right)$, and $P_n \in \mathbf{P}_{n,0}$. Conditional on $E_n$, we have $\hat{P}_n \in \mathbf{P}_{n,0}$ due to Step 2, then 
 \begin{equation}\label{eq:aux0_proof_thm3.1}
     \sup_{h \le h_n} \: \sup_{x \in \mathbf{R}} ~ \left|  J_n(x,h,\hat{P}_n,\hat{\rho}_n) - \Phi(x)   \right| < \epsilon/2~,
 \end{equation}
 for any $n \ge N_2$, $  h_n \le n$ and $h_n = o\left(n\right) $. By \eqref{eq:cdf_AA} there exists $N_3 \ge N_2$ such that
 $$ \sup_{h \le h_n} \: \sup_{ \tilde{\rho} \in [-1,1]} ~   \sup_{x \in \mathbf{R}} ~ \left|  J_n(x,h,P,\tilde{\rho}) - \Phi(x)   \right| < \epsilon/2~,  $$
 for any $n \ge N_3$, $ h_n \le n$, and $h_n = o\left(n\right)$. Therefore, conditional on the event $E_n$ and using triangular inequality, we conclude that 
 $$ \sup_{h \le h_n} \: \sup_{\tilde{\rho} \in [-1,1]} ~  \sup_{x \in \mathbf{R}} ~ \left|  J_n(x,h,P,\tilde{\rho}) -  J_n(x,h,\hat{P}_n,\hat{\rho}_n)   \right| < \epsilon~,$$
 for any $n \ge N_3$, $h_n \le n$, and $h_n = o\left(n\right)$. Since $P_{\rho}(E_n) \ge 1-\eta$ for any $\rho \in [-1,1]$, the previous conclusion is equivalent to
 $$  \sup_{\rho \in [-1,1]} ~ P \left( \sup_{h \le h_n} \: \sup_{\tilde{\rho} \in [-1,1]} ~ \sup_{x\in \mathbf{R}} \: \left|  J_n(x,h,P,\tilde{\rho}) -  J_n(x,h,\hat{P}_n,\hat{\rho}_n)   \right| < \epsilon ~ \right ) \ge 1- \eta~, $$
for any $n \ge N_3$, $h_n \le n$ and $h_n = o\left(n\right)$, which concludes the proof of the theorem.


\subsection{Proof of Theorem \ref{thm:CI_bootstrap}}\label{appendix:thm:CI_bootstrap}
By Lemma \ref{lemma:quantiles_high-prob}, for any $\epsilon>0$, there exists $N_0 = N_0(\epsilon)$  such that
\begin{equation}\label{eq:aux0_thm3_2}
    P_{\rho} \left( z_{1-\alpha/2-\epsilon/2} \le c_n^{*}(h,1-\alpha)  \le z_{1-\alpha/2+\epsilon/2} \right) \ge 1 - \epsilon~,
\end{equation}
for any $n \ge N_0$,  $\rho \in [-1,1]$ and any $h \le h_n \le n $ and $h_n = o\left(n\right)$. Assumptions \ref{xu_assumptions} and \ref{A3:MO-PM} guarantee \eqref{eq:cdf_AA}; therefore, there exist $N_1 \ge N_0$ such that 
\begin{equation}\label{eq:aux1_thm3_2}
    P_{\rho}\left( |R_n(h)| \le z_{1-\alpha/2+\epsilon/2} \right) \le 1-\alpha+2\epsilon \quad \text{and} \quad P_{\rho}\left( |R_n(h)| \le z_{1-\alpha/2-\epsilon/2} \right) \ge 1-\alpha-2\epsilon~,
\end{equation} 
for any $n \ge N_1$,  $\rho \in [-1,1]$ and any $h \le h_n \le n $ and $h_n = o\left(n\right)$. Consider the derivation 
\begin{align*}
    P_{\rho}\left( \beta(\rho,h) \in C_n^*(h,1-\alpha) \right) &= P_{\rho}\left( |R_n(h)| \le c_n^*(h,1-\alpha) \right) \\
      &= P_{\rho}\left( |R_n(h)| \le c_n^*(h,1-\alpha) , c_n^*(h,1-\alpha) > z_{1-\alpha/2 +\epsilon/2} \right) \\
    &~+ P_{\rho}\left( |R_n(h)| \le c_n^*(h,1-\alpha) ,  c_n^*(h,1-\alpha) \le z_{1-\alpha/2+\epsilon/2}  \right) \\
     &\le P_{\rho}\left( c_n^*(h,1-\alpha) > z_{1-\alpha/2 + \epsilon/2}  \right) + P_{\rho}\left( |R_n(h)| \le z_{1-\alpha/2+\epsilon/2} \right) \\
    & \le \epsilon + 1-\alpha + 2\epsilon~,
\end{align*}
where the last inequality follows by \eqref{eq:aux0_thm3_2} and \eqref{eq:aux1_thm3_2}. Similarly, we obtain the inequality
$$ P_{\rho}\left( |R_n(h)| \le z_{1-\alpha/2-\epsilon/2} \right) \le P_{\rho}\left( \beta(\rho,h) \in C_n^*(h,1-\alpha) \right) + P_{\rho}\left( c_n^*(h,1-\alpha) < z_{1-\alpha/2-\epsilon/2}  \right)~,$$
which implies that $P_{\rho}\left( \beta(\rho,h) \in C_n^*(h,1-\alpha) \right) \ge  1-\alpha-2\epsilon - \epsilon$. We conclude that for any $n \ge N_1$, $\rho \in [-1,1]$ and any $h \le h_n \le n $ and $h_n = o\left(n\right)$, we have
$$ | P_{\rho}\left( \beta(\rho,h) \in C_n^*(h,1-\alpha) \right) - (1-\alpha)| \le 3\epsilon~, $$
which completes the proof of Theorem \ref{thm:CI_bootstrap}.


\subsection{Proof of Theorem \ref{thm:rates_AA}}\label{appendix:thm:rates_AA}

We first show that $J_n(x,h,P,\rho)$ admits a valid Edgeworth expansion, that is
\begin{equation}\label{eq:edgeworth_Jn}
  \sup_{x\in \textbf{R}}     \left|J_n(x,h,P,\rho) - \left( \Phi(x) + \sum_{j=1}^2 n^{-j/2}q_j(x,h,P,\rho) \phi(x) \right) \right| = O\left(n^{-1-\epsilon} \right)~ 
\end{equation}
for some $\epsilon \in (0,1/2)$, where $q_j(x,h,P,\rho)$ are polynomials on $x$ with coefficients that are continuous functions of the moments of $P$ (up to order 12) and $\rho$. Furthermore, we have $q_1(x,h,P,\rho) = q_1(-x,h,P,\rho)$ and $q_2(x,h,P,\rho) = -q_2(-x,h,P,\rho)$. 

To show \eqref{eq:edgeworth_Jn}, we first use Lemma  \ref{lemma:approximation_Jn} to approximate $J_n(x,h,P,\rho)$ by $\tilde{J}_n(x,h,P,\rho)$,
$$ \sup_{x\in \textbf{R}} |J_n(x,h,P,\rho) - \tilde{J}_n(x,h,P,\rho)| = D_n + O\left(n^{-1-\epsilon} \right)~,$$
for some $\epsilon \in (0,1/2)$, where
    $$ D_n = \sup_{x \in \mathbf{R}} \left| \tilde{J}_n(x+n^{-1-\epsilon},h,P,\rho) - \tilde{J}_n(x-n^{-1-\epsilon},h,P,\rho)   \right|~.$$ 
Due to Theorem \ref{thm:edgeworth_approx_Jtilde}, we can conclude $D_n = O\left(n^{-1-\epsilon} \right)$. We then use Theorem \ref{thm:edgeworth_approx_Jtilde} to approximate $\tilde{J}_n(x,h,P,\rho)$ by a valid Edgeworth expansion,
\begin{equation*} 
     \sup_{x \in \textbf{R}} \left| \tilde{J}_n(x,h,P,\rho) - \left(  \Phi(x) + \sum_{j=1}^2 n^{-j/2}q_j(x,h,P,\rho) \phi(x) \right) \right| = O\left(n^{-3/2} \right)~.
\end{equation*}
Note that we can use Theorem \ref{thm:edgeworth_approx_Jtilde} since Assumption \ref{AV_assumptions} implies Assumption \ref{appendix:AV_assumptions2} and the distribution $\tilde{J}_n(x,h,P,\rho)$ that we obtain from Lemma  \ref{lemma:approximation_Jn} satisfy the required conditions. We conclude \eqref{eq:edgeworth_Jn} by triangular inequality. The polynomials $q_j$ that appear in \eqref{eq:edgeworth_Jn} are the polynomials in the Edgeworth expansion of $\tilde{J}_n(x,h,P,\rho)$.

Now, we show that $P_{\rho} \left( |R_n(h)| \le x  \right)$ also admits an asymptotic approximation, that is
\begin{equation}\label{eq:edgeworth_Jn_2}
    \sup_{x\in \textbf{R}}  \left| P_{\rho} \left( |R_n(h)| \le x  \right) - \left( 2\Phi(x) - 1 + 2n^{-1}   q_2(x,h,P,\rho)\phi(x) \right)\right| = O\left( n^{-1-\epsilon}\right)~,
\end{equation}
where $q_2(x,h,P,\rho)$ and $\epsilon \in (0,1/2)$ are defined in \eqref{eq:edgeworth_Jn}. Note that \eqref{eq:thm_rates_AA_2} follows from  \eqref{eq:edgeworth_Jn_2} since we can write \eqref{eq:thm_rates_AA_2} as follows
$$ \left| P_{\rho} \left( |R_n(h)| \le z_{1-\alpha/2}  \right) - (1-\alpha) \right| = O \left( n^{-1} \right)~, $$
and the previous expression is what we obtain taking $x=z_{1-\alpha/2}$ in \eqref{eq:edgeworth_Jn_2}, where we used that $1-\alpha = 2\Phi(z_{1-\alpha/2}) - 1$ holds by definition of $z_{1-\alpha/2}$.

To show \eqref{eq:edgeworth_Jn_2}, we first write 
$$ P_{\rho} \left( |R_n(h)| \le x  \right) = J_n(x,h,P,\rho) -J_n(-x,h,P,\rho) + r_n(x)~,$$
where $r_n(x) = P_{\rho} \left( R_n(h) = -x  \right)$. We then use \eqref{eq:edgeworth_Jn} to approximate $J_n(\cdot,h,P,\rho) $ and the properties of the polynomials $q_j(\cdot,h,P,\rho)$ to obtain the following approximation
\begin{equation*} 
    \sup_{x\in \textbf{R}}  \left|  P_{\rho} \left( |R_n(h)| \le x  \right)  - \left( 2\Phi(x) - 1 + 2n^{-1}   q_2(x,h,P,\rho)\phi(x)   + r_n(x) \right) \right|  = O\left( n^{-1-\epsilon}\right)~.
\end{equation*}
Finally, $\sup_{x\in \textbf{R}}  r_n(x) = O\left(n^{-1-\epsilon}\right)$ since $ r_n(x) \le  P_{\rho} \left( R_n(h) \in (-x - n^{-1-\epsilon},-x]  \right) $ and \eqref{eq:edgeworth_Jn} holds. We use this in the previous expression to complete the proof of  \eqref{eq:edgeworth_Jn_2}. 


\subsection{Proof of Theorem \ref{thm:rates_bootstrap}}\label{appendix:thm:rates_bootstrap}

The proof has two parts. In the first part we assume that $P(|\Delta_n| > C_1 n^{-1-\epsilon}) \le C_2 n^{-1-\epsilon}$ for some constants $C_1$ and $C_2$, where $\Delta_n = c_n^*(h,1-\alpha) - c_n(h,1-\alpha)$. We use this assumption to prove the theorem with an error of size $O(n^{-(1+\epsilon)})$ for any $\epsilon \in (0,1/2)$, which is sufficient to conclude. In the second part, we prove the assumption of the first part.

\noindent \textbf{Part 1:} By \eqref{eq:bootstrap_CI}, we have $P_{\rho} \left( \beta(\rho,h) \in C_n^*(h,1-\alpha) \right)  = P_{\rho} \left( |R_n(h)| \le c_n^*(h,1-\alpha) \right) $. We can write this term as the sum of
$ P_{\rho} \left( |R_n(h)| \le c_n(h,1-\alpha) + \Delta_n, \left| \Delta_n \right| \le C_1 n^{-1-\epsilon} \right)$
and
$P_{\rho} \left( |R_n(h)| \le c_n(h,1-\alpha) + \Delta_n, \left| \Delta_n \right| > C_1 n^{-1-\epsilon} \right) $.
We conclude
$ P_{\rho} \left( \beta(\rho,h) \in C_n^*(h,1-\alpha) \right)$
is equal to 
$$ P_{\rho} \left( |R_n(h)| \le c_n(h,1-\alpha) + \Delta_n, \left| \Delta_n \right| \le C_1 n^{-1-\epsilon} \right) +O\left(n^{-1-\epsilon} \right)~.$$
    %
    %
    %
%
By \eqref{eq:edgeworth_Jn_2}  in the proof of Theorem \ref{thm:rates_AA}, we have 
$$  P_{\rho} \left( |R_n(h)| \le x + z n^{-1-\epsilon} \right) =  P_{\rho} \left( |R_n(h)| \le x \right) + O\left(n^{-1-\epsilon} \right)$$
for $z = -C_1,C_1$ and any $x \in \mathbf{R}$. Since 
$$P_{\rho} \left( |R_n(h)| \le x + \Delta_n, \left| \Delta_n \right| \le C_1 n^{-1-\epsilon} \right) \le P_{\rho} \left( |R_n(h)| \le x + C_1 n^{-1-\epsilon} \right)$$
and
$$P_{\rho} \left( |R_n(h)| \le x + \Delta_n, \left| \Delta_n \right| \le C_1 n^{-1-\epsilon} \right) \ge P_{\rho} \left( |R_n(h)| \le x - C_1 n^{-1-\epsilon} \right) + O\left(n^{-1-\epsilon} \right)~,$$
we conclude $P_{\rho} \left( |R_n(h)| \le x + \Delta_n, \left| \Delta_n \right| \le n^{-1-\epsilon} \right) = P_{\rho} \left( |R_n(h)| \le x \right) + O\left(n^{-1-\epsilon} \right)$. Taking $x = c_n(h,1-\alpha) $ and using that $P_{\rho} \left( |R_n(h)| \le c_n(h,1-\alpha) \right) = 1-\alpha$ (due to part 2 in Assumption \ref{AV_assumptions}), we conclude
$  P_{\rho} \left( \beta(\rho,h) \in C_n^*(h,1-\alpha) \right)  = 1-\alpha +  O\left(n^{-1-\epsilon} \right).$

\noindent \textbf{Part 2:} Fix $\epsilon \in (0,1/2)$. Define $E_{n,1} = \{ |\hat{\rho}_n| \le 1- a/2\}$, $E_{n,2} = \{ n^{-1} \sum_{t=1}^n \Tilde{u}_t^2  \ge \Tilde{C}_\sigma  \}$, $E_{n,3} = \{ n^{-1} \sum_{t=1}^n \Tilde{u}_t^{4k} \le M \}$, and $E_{n,4} = \{ \max_{ 1 \le r \le 12 } |n^{-1} \sum_{t=1}^n \Tilde{u}_t^r - E[u_t^r]| \le  n^{-\epsilon}\}$, where $\Tilde{C}_\sigma$ and $M$ are as in Lemma \ref{lemma:high-prob-event}. Define $E_{n} = E_{n,1} \cap E_{n,2} \cap E_{n,3}  \cap E_{n,4}$. By Lemma \ref{lemma:high-prob-event} and Assumption \ref{AV_assumptions}, it follows that $P(E_n^c) \le C_2 n^{-1-\epsilon}$ for some constant $C_2=C_2(a,h,k,C_\sigma, \epsilon,c_u)$. Note that conditional on the event $E_n$, we can use Lemma \ref{lemma:approximation_Jn} for the distribution of the bootstrap root $R_n^*(h)$. That is 
$$ \sup_{x \in \mathbf{R}} |J_n(x,h,\hat{P}_n,\hat{\rho}_n) - \tilde{J}_n(x,h,\hat{P}_n,\hat{\rho}_n)| \le D_n + n^{-1-\epsilon} C\left( n^{-1} \sum_{t=1}^{n} |\tilde{u}_t|^{k} + \Tilde{u}_t^{2k} + \Tilde{u}_t^{4k} \right) ~,$$
for some constant $C $,  where 
$$ D_n = \sup_{x \in \mathbf{R}} \left| \tilde{J}_n(x+n^{-1-\epsilon},h,\hat{P}_n,\hat{\rho}_n) - \tilde{J}_n(x-n^{-1-\epsilon},h,\hat{P}_n,\hat{\rho}_n)   \right|~.$$
By Theorem \ref{thm:edgeworth_approx_Jtilde_bootstrap}, there is an Edgeworth expansion for $\tilde{J}_n(x,h,\hat{P}_n,\hat{\rho}_n)$ conditional on $E_n$. This implies $D_n \le C  n^{-1-\epsilon}$ conditional on $E_n$, for some constant $C $. Similarly, conditional on $E_n$, $ n^{-1} \sum_{t=1}^{n} \left( |\tilde{u}_i|^{k} + \Tilde{u}_t^{2k} + \Tilde{u}_t^{4k} \right) \le C$, for some constant $C$ that depends on $M$. We conclude that, conditional on $E_n$, $J_n(x,h,\hat{P}_n,\hat{\rho}_n)$ has the following Edgeworth expansion,
$$ \sup_{x \in \mathbf{R}} \left|J_n(x,h,\hat{P}_n,\hat{\rho}_n) - \left( \Phi(x) + \sum_{j=1}^2 n^{-j/2}q_j(x,h,\hat{P}_n,\hat{\rho}_n) \phi(x)  \right) \right| \le C n^{-1-\epsilon}~.$$
The properties of $q_j(x,h,\hat{P}_n,\hat{\rho}_n)$ from Theorem \ref{thm:edgeworth_approx_Jtilde_bootstrap} and arguments from the proof of Theorem \ref{thm:rates_AA} imply
$$ \sup_{x \in \mathbf{R}} \left| P_{\rho} \left( |R_n^*(h)| \le x \mid Y^{(n)} \right) - \left( 2\Phi(x) - 1 + 2n^{-1}   q_2(x,h,\hat{P}_n,\hat{\rho}_n)\phi(x) \right)\right|  \le C n^{-1-\epsilon}~.$$
Recall that the coefficients of  $q_2(x,h,\hat{P}_n,\hat{\rho}_n)$ are polynomial of the moments of $\hat{P}_n$ (up-to order 12) and $\hat{\rho}_n$. Conditional on $E_n$, we \textcolor{black}{know} the moments of $\hat{P}_n$ are close to the moments of $P$: $ |n^{-1} \sum_{t=1}^n \Tilde{u}_t^r - E[u_t^r]| \le  n^{-\epsilon}$ for $r=1,\dots,12$. Therefore, conditional on $E_n$, we have
$$ \sup_{x \in \mathbf{R}} \left| P_{\rho} \left( |R_n^*(h)| \le x \mid Y^{(n)} \right) - \left( 2\Phi(x) - 1 + 2n^{-1}   q_2(x,h,P,\rho)\phi(x) \right)\right|  \le C n^{-1-\epsilon}~,$$
for some constant $C$. By \eqref{eq:edgeworth_Jn_2}  in the proof of Theorem \ref{thm:rates_AA}, the previous inequality, and the definition of $c_n^*(h,1-\alpha)$ and  $c_n(h,1-\alpha)$ as quantiles, we conclude that
$$ |c_n^*(h,1-\alpha) - c_n(h,1-\alpha)| \le C_1 n^{-1-\epsilon} $$
for some constant $C_1$. This completes the proof of our assumption in part 1. 
 

 \counterwithin{figure}{section}
\counterwithin{table}{section}
\renewcommand{\thetable}{\thesection.\arabic{table}}
\renewcommand{\thefigure}{\thesection.\arabic{figure}}
\renewcommand{\theequation}{B.\arabic{equation}}
\setcounter{equation}{0}

\section{Auxiliary Results}\label{sec:appendix_b} 

\subsection{Lemmas}\label{appendixB:lemmas}

\begin{lemma}\label{AppendixA:lemma_event_En}
    Suppose Assumptions \ref{xu_assumptions} and \ref{A3:MO-PM} hold. Then, for any fixed $\eta>0$, there exist constants $M>0$, $\tilde{K}_4>0$, and $N_0 = N_0(\eta)$ such that
    \begin{enumerate}
        \item $ P_{\rho} \left( g(\rho,n) ~n^{1/2}~ | \hat{\rho}_n - \rho |  > M \right) < \eta~,$

        \item $ P_{\rho} \left( \left | n^{-1} \sum_{t=1}^n \Tilde{u}_t^2 - \sigma^2  \right| > \sigma^2/2    \right) < \eta ~,$ 

        \item $ P_{\rho} \left(   n^{-1} \sum_{t=1}^n \Tilde{u}_t^4  > \tilde{K}_4    \right) < \eta ~,$ 
    \end{enumerate}
    for $n \ge N_0$ and $\rho \in [-1,1]$, where $g(\rho,k) = \left(\sum_{\ell=0}^{k-1}  \rho^{2 \ell}\right)^{1/2}$, $\hat{\rho}_n$ is as in \eqref{eq:rho_hat}, and $\{ \Tilde{u}_t: 1 \le t \le n\}$ are centered residuals as in \eqref{eq:centered_residuals}.
\end{lemma}

\begin{proof}
    See Section \ref{sec:proof_lemma_b1} in \textcolor{black}{Online Supplemental} Appendix \ref{sec:appendix_online1}.
\end{proof}

\begin{lemma}\label{AppendixA:lemma_Mtilde}
For any fixed $M>0$. Suppose that for any $\rho \in [-1,1]$ we have
\begin{equation*} 
    |\hat{\rho}_n - \rho| \le  \frac{M}{n^{1/2}g(\rho,n)}~,
\end{equation*}
where $g(\rho,k) = \left(\sum_{\ell=0}^{k-1} ~ \rho^{2 \ell}\right)^{1/2}$. Then, there exist constants $\tilde{M} = \tilde{M}(M)>0$ and $N_0=N_0(M)>0$ such that
$
    |\hat{\rho}_n| \le 1 + \tilde{M}/n~
$
for all $n \ge N_0$.
\end{lemma}

\begin{proof}
    See Section \ref{sec:proof_lemma_b2} in \textcolor{black}{Online Supplemental} Appendix \ref{sec:appendix_online1}.
\end{proof}

\begin{lemma}\label{lemma:quantiles_high-prob}
    Suppose Assumptions \ref{xu_assumptions} and \ref{A3:MO-PM} hold. Fix $\epsilon>0$. Then, for any $\alpha \in (0,1)$ and for any sequence $h_n \le n$ such that $h_n = o\left(n\right)$, we have
    \begin{enumerate}
        \item $\lim_{n \to \infty} ~ \sup_{h \le h_n} ~ \sup_{\rho \in [-1,1]} P_{\rho} \left( z_{1-\alpha/2-3\epsilon/2} \le c_n^*(h,1-\alpha) \le z_{1-\alpha/2+3\epsilon/2}  \right) = 1~,$

        \item $\lim_{n \to \infty} ~ \sup_{h \le h_n} ~ \sup_{\rho \in [-1,1]} P_{\rho} \left( z_{\alpha_0-\epsilon/2} \le q_n^*(h,\alpha_0) \le z_{\alpha_0+\epsilon/2}  \right) = 1~,$
    \end{enumerate}
    where $z_{\alpha_0}$ is the $\alpha_0$-quantiles of the standard normal distribution, $c_n^*(h,1-\alpha) $ is as in \eqref{eq:cv_bootstrap}, and $q_n^*(h,\alpha_0)$ is the $\alpha_0$-quantile of $R_{b,n}^*(h)$ defined in \eqref{eq:root_bootstrap}.
\end{lemma}

\begin{proof}
    See Section \ref{sec:proof_lemma_b3} in \textcolor{black}{Online Supplemental} Appendix \ref{sec:appendix_online1}.
\end{proof}

\begin{lemma}\label{lemma:approximation_Jn}
    Suppose Assumption \ref{AV_assumptions} holds. For any fixed $h \in \mathbf{N}$ and $a \in (0,1)$. Then, for any $\rho \in [-1+a,1-a]$ and $\epsilon \in (0,1/2)$, there exist constant $C = C(a,h,k,C_\sigma)>0$ and a real-valued function 
    $$\mathcal{T}(\cdot;\sigma^2,\psi_4^4,\rho) : \mathbf{R}^8 \to \mathbf{R}~,$$
    such that %
    \begin{enumerate}
        \item $\mathcal{T}(\mathbf{0};\sigma^2,\psi_4^4,\rho) = 0$, 

        \item $\mathcal{T}(x;\sigma^2,\psi_4^4,\rho) $ is a polynomial of degree 3 in $x \in \mathbf{R}^8$   with coefficients that are continuously differentiable functions of $\sigma^2$, $\psi_4^4$, and  $\rho$, 

        \item $ \sup_{x \in \mathbf{R}} |J_T(x,h,P,\rho) - \Tilde{J}_n(x,h,P,\rho)| \le D_n + n^{-1-\epsilon} C\left(E[|u_t|^{k}] + E[u_t^{2k}] + E[u_t^{4k}] \right),$
    \end{enumerate}
    where  $\sigma^2 = E_P[u_1^2]$, $\psi_4^4 = E_P[u_1^4]$, $k \ge 8(1+\epsilon)/(1-2\epsilon)$,
    $$\Tilde{J}_n(x,h,P,\rho) \equiv P_{\rho} \left( (n-h)^{1/2} \mathcal{T} \left( \frac{1}{n-h} \sum_{t=1}^{n-h} X_t ;\sigma^2,\psi_4^4,\rho \right) \le x \right)~,$$ 
    and
    $$ D_n = \sup_{x \in \mathbf{R}} \left| \Tilde{J}_n(x+n^{-1-\epsilon},h,P,\rho) - \Tilde{J}_n(x-n^{-1-\epsilon},h,P,\rho)   \right|~.$$
    The sequence $\{ X_t : 1 \le t \le n-h \}$ is defined in \eqref{eq:Xi-vec}. Furthermore, the asymptotic variance of $(n-h)^{1/2} \mathcal{T}( (n-h)^{-1} \sum_{t=1}^{n-h} X_t ;\sigma^2,\psi_4^4,\rho  )$ is equal to one.
\end{lemma}

\begin{proof}
    See Section \ref{sec:proof_lemma_b4} in \textcolor{black}{Online Supplemental} Appendix \ref{sec:appendix_online2}.
\end{proof}

\begin{lemma}\label{lemma:high-prob-event}
Suppose Assumption \ref{AV_assumptions} holds.  For any fixed $h \in \mathbf{N}$ and $a \in (0,1)$. Then, for any $|\rho|\le 1-a$ and $\epsilon \in (0,1/2)$, there exist $ {C} = C(a,k,h,C_\sigma,\epsilon,c_u)$, $ \tilde{C}_{\sigma}$, and $M$ such that
\begin{enumerate}
    \item $P \left( \left|\hat{\rho}_n \right| > 1-a/2 \right)  \le {C} n^{-1-\epsilon} $

    \item $P \left( \left| n^{-1} \sum_{t=1}^{n} \Tilde{u}_t^{r} - E[u_t^r] \right| > n^{-\epsilon} \right)  \le {C} n^{-1-\epsilon} $  

    \item $P \left(  n^{-1} \sum_{t=1}^{n} \Tilde{u}_t^2  < \tilde{C}_{\sigma} \right)  \le {C} n^{-1-\epsilon}  $

    \item $P \left(   n^{-1} \sum_{t=1}^{n} \Tilde{u}_t^{4k}  > M \right)  \le {C} n^{-1-\epsilon}$  
\end{enumerate}
for fixed $r \ge 1$, $k \ge 8(1+\epsilon)/(1-2\epsilon)$, where $\hat{\rho}_n$ and the centered residuals $\{ \Tilde{u}_t: 1 \le t \le n\}$ are defined in \eqref{eq:rho_hat} and \eqref{eq:centered_residuals}, respectively.
\end{lemma} 

\begin{proof}
    See Section \ref{sec:proof_lemma_b5} in \textcolor{black}{Online Supplemental} Appendix \ref{sec:appendix_online2}.
\end{proof}

\subsection{Uniform Consistency }\label{sec:appendix_uniform}

For any fixed $M>0$, consider the sequence of models:
$$ y_{n,t} = \rho_n y_{n,t-1} + u_{n,t}~,~ y_{n,0} = 0~,\quad \text{and} \quad \rho_n \in [-1-M/n,1+M/n]~,$$
where $\{u_{n,t} : 1 \le t \le n\}$ is a sequence of shocks with probability distribution denoted by $P_n$. We use $P_n$ and $E_n$ to compute respectively probabilities and expected values of the sequence $\{(y_{n,t},u_{n,t}) : 1 \le t \le n \}$. This appendix presents results for a sequence of AR(1) models.
 
We extend the notation introduced in Section \ref{section:preliminaries} for the sequence of models. For fixed any $h<n$, the coefficients in the linear regression of $y_{n,t+h}$ on $(y_{n,t},y_{n,t-1})$ are defined by
\begin{equation}\label{eq:appendix_A_LP}
  \begin{pmatrix}
    \hat{\beta}_n(h) \\
    \hat{\gamma}_n(h)
\end{pmatrix} = \left( \sum_{t=1}^{n-h} x_{n,t} x_{n,t}' \right)^{-1} \left( \sum_{t=1}^{n-h} x_{n,t} y_{n,t+h} \right)~,
\end{equation}
where $x_{n,t} \equiv (y_{n,t}, y_{n,t-1})'$. And the HC standard error $\hat{s}_n(h)$ is defined by
\begin{equation*}
     \hat{s}_n(h) \equiv \left( \sum_{t=1}^{n-h} \hat{u}_{n,t}(h)^2 \right)^{-1/2} \left(\sum_{t=1}^{n-h} \hat{\xi}_{n,t}(h)^2 \hat{u}_{n,t}(h)^2 \right)^{1/2} \left( \sum_{t=1}^{n-h} \hat{u}_{n,t}(h)^2 \right)^{-1/2}~,
\end{equation*}
where $\hat{\xi}_{n,t}(h) = y_{n,t+h} - \hat{\beta}_n(h) y_{n,t} - \hat{\gamma}_n(h) y_{n,t-1}$,  $\hat{u}_{n,t}(h) = y_{n,t} - \hat{\rho}_n(h) y_{n,t-1}$, and $\hat{\rho}_n(h)$ is defined as 
\begin{equation}\label{eq:appendix_A_OLS}
    \hat{\rho}_n(h) \equiv \left( \sum_{t=1}^{n-h} y_{n,t-1}^2 \right)^{-1} \left(  \sum_{t=1}^{n-h} y_{n,t} y_{n,t-1} \right)~.
\end{equation}

For any fixed positive constants $K_4 >0$ and $ \overline{\sigma} \ge \underline{\sigma} >0$, we consider the next assumption that imposes restrictions on the distribution of the shocks $P_n$.

\begin{assumption}\label{appendix:AV_assumptions} \hspace{1cm} \vspace{-0.25cm}
\begin{enumerate}
    \item[i)] $\{u_{n,t}: 1 \le t \le n \}$ are i.i.d. random variables with mean zero and variance $\sigma_n^2$. 

    \item[ii)] ${E}_n[u_{n,t}^4]< K_4$ and  $\sigma_n^2 \in [\underline{\sigma},\:\overline{\sigma}]$.
     
\end{enumerate}
\end{assumption}

We denote by $\mathbf{P}_{n,0} $ the set of all distributions $P_n$ that verify Assumption \ref{appendix:AV_assumptions}.
Theorem \ref{thm:sequence_models} below shows that the results presented in \cite{xu2022} and \cite{MO-PM2020} also hold for sequences of AR(1) models with i.i.d. shocks.
We adapt their proof and simplify some steps based on our stronger assumptions over the serial dependence of the shocks. For instance, we assume only bounded 4th moments, while they assume bounded at least 8th bounded moments. One remarkable difference is that we do not need to assume a high-level assumption such as Assumption \ref{A3:MO-PM} since this can be verified using Assumption \ref{appendix:AV_assumptions}; we present the claim of this result in the next proposition.

\begin{proposition}\label{proposition:A3}
Suppose Assumption \ref{appendix:AV_assumptions} holds. Then, we have
\begin{equation*}
 \lim_{K \to \infty} ~ \lim_{n \to \infty} ~ \inf_{P_n \in \mathbf{P}_{n,0}}  ~ \inf_{|\rho_n| \le 1+M/n} 
 ~ P_n \left(~  g(\rho,n)^{-2} ~ n^{-1} \sum_{t=1}^{n}y_{n,t-1}^2  \ge 1/K ~\right) = 1~,  
\end{equation*}
where $g(\rho,k) = \left(\sum_{\ell=0}^{k-1} ~ \rho^{2 \ell}\right)^{1/2}$.    
\end{proposition}

\begin{proof}
    See Section \ref{sec:proof_prop_b1} in \textcolor{black}{Online Supplemental} Appendix \ref{sec:appendix_online1}.
\end{proof} 

\begin{theorem}\label{thm:sequence_models}
Suppose Assumption \ref{appendix:AV_assumptions} holds. Then, for any sequence $h_n \le n$ such that $h_n = o\left(n\right)$, we have
    \begin{equation*}
      \sup_{h \le h_n} ~  \sup_{P_n \in \mathbf{P}_{n,0}}   ~ \sup_{|\rho| \le 1+M/n} ~ \sup_{x \in \mathbf{R}} ~ \left|  J_n(x,h,P_n,\rho) - \Phi(x)   \right| \to 0~, \quad \text{as} \quad n \to \infty,
    \end{equation*}
    where $J_n(\cdot,h,P_n,\rho)$ is as in \eqref{eq:cdf_root} and $\Phi(x)$ is the cdf of the standard normal distribution.
\end{theorem}

\begin{proof}
    See Section \ref{sec:proof_thm_b1} in \textcolor{black}{Online Supplemental} Appendix  \ref{sec:appendix_online1}.
\end{proof}

\begin{proposition}\label{prop:CI_LA_AR}
Suppose Assumption \ref{appendix:AV_assumptions} holds. In addition, assume $\rho_n = 1 - c_1/n$ and $h_n$ is such that $h_n \le n$ and $h_n/\sqrt{n} \to c_2 $ as $n \to \infty$ where $c_1,c_2>0$. Then,
$$ \liminf_{n \to \infty} P_n \left( [1/L,L] \subseteq  C_{la-ar}^*(h_n,1-\alpha) \right) \ge 1-\alpha $$
for any $L > 1 $, where $C_{la-ar}^*(h,1-\alpha)$ is defined in Remark \ref{rem:other-valid-bootstraps}, and presented below
$$ C_{la-ar}^*(h,1-\alpha) = \left[ (\hat{\beta}_n(1) - \hat{s}_n(1) c_n^*(1, 1-\alpha))^{h},~(\hat{\beta}_n(1) + \hat{s}_n(1) c_n^*(1, 1-\alpha))^{h} \right]~.$$
\end{proposition} 

\begin{proof}
    See Section \ref{sec:proof_prop_b2} in \textcolor{black}{Online Supplemental} Appendix \ref{sec:appendix_online1}.
\end{proof}

\subsection{Asymptotic Refinements}\label{sec:appendix_asympt_ref}

Consider the sequence $\{z_t: 1 \le t \le n\}$ defined as 
$$ z_t = \rho z_{t-1} + u_t~, \quad \text{and} \quad z_0 =  \sum_{\ell = 0}^{\infty} \rho^{\ell} u_{-\ell}~,$$
where $\{u_{-\ell} : \ell \ge 0 \}$ is an i.i.d. sequence with the same distribution as $u_1$. This appendix presents asymptotic expansion results for distributions of real value functions based on sample averages of the sequence $\{X_t = F(z_{t-1},z_t,z_{t+h}) : 1 \le t \le n-h \}$, where $F$ is a function that we define below. Our approach in this section relies on the framework and results presented in \cite{gotze1994asymptotic} and \cite{bhattacharya1978validity}.

Let $F(\cdot~; \sigma^2, V, \rho): \mathbf{R}^3 \to \mathbf{R}^8$ be a function defined at $(x,y,z)$ equal to
\begin{align}
      \Big( &(z-\rho^h y) (y-\rho x),~ (y-\rho x)^2-\sigma^2, ~((z-\rho^h y) (y-\rho x))^2-V,~(z-\rho^h y)  (y-\rho x)^3,~ \notag \\
             &(y-\rho x) x,~ (z-\rho^h y)  x,~ (z-\rho^h y) ^2 (y-\rho x) x,~ (z-\rho^h y) (y-\rho x)^2 x  \Big)~,\label{eq:def_F_appendix}
\end{align}
where $\sigma^2 = \sigma^2(P) = E_P[u_1^2]$, $V = V(\rho, h, P) = E_P[\xi_1^2 u_1^2]$, $\xi_1 = \xi_1(\rho,h) \equiv \sum_{\ell=1}^h \rho^{h-\ell} u_{1+\ell}$, and $P$ is the distribution of the shocks that verified Assumption \ref{appendix:AV_assumptions2} that we define below. Using that $u_t = z_t - \rho z_{t-1}$, $\xi_t = z_{t+h} - \rho^h z_t$, and the definition of $F$ in \eqref{eq:def_F_appendix}, we can write the sequence of random vectors $\{X_t = F(z_{t-1},z_t,z_{t+h}; \sigma^2, V, \rho)): 1 \le t \le n-h \}$ as follows
\begin{equation}\label{eq:Xi-vec}
    X_t = (\xi_t u_t, u_t^2-\sigma^2,(\xi_tu_t)^2-V,\xi_t u_t^3, u_t z_{t-1}, \xi_t z_{t-1}, \xi_t^2 u_t z_{t-1}, \xi_t u_t^2 z_{t-1})~.
\end{equation}

We assume in this section that $h \in \mathbf{N}$ is fixed and $|\rho|<1$. Moreover, for any fixed positive constants $C_{18}>0$ and $C_{\sigma}>0$, we consider the next assumption that imposes restrictions on the distribution of the shocks $P$.

\begin{assumption}\label{appendix:AV_assumptions2}
\hspace{1cm}\vspace{-0.25cm}

\begin{enumerate}
    \item[i)] $\{u_t: 1 \le t \le n\}$ is independent and identically distributed with $E[u_t] =0$.

    \item[ii)] $u_t$ has a positive continuous density.

    \item[iii)] $E[u_t^{18}] \le C_{18} < \infty$ and $E[u_t^2] \ge C_{\sigma} $.    
\end{enumerate}

\end{assumption}

Assumption \ref{appendix:AV_assumptions2} implies that the sequence $\{z_t: 1 \le t \le n\}$ is strictly stationary. By construction, $E[X_t] = \mathbf{0} \in \mathbf{R}^8$. Define 
\begin{equation}
    \Sigma = \lim_{n \to \infty} Cov\left(  (n-h)^{-1/2}  \sum_{t=1}^{n-h} X_t \right)~.
\end{equation}
The asymptotic covariate matrix $\Sigma$ is non-singular due to Lemma 2.1 in \cite{gotze1994asymptotic}, Assumption \ref{appendix:AV_assumptions2}, and how we defined the sequence $\{X_t: 1 \le t \le n-h\}$. Let $\mathcal{T}: \mathbf{R}^8 \to \mathbf{R}$ be a polynomial with coefficients depending on $\rho$, $E_P[u_1^2]$, and $E_P[u_1^4]$ such that $\mathcal{T}(\mathbf{0}) = 0$. Define
\begin{equation}\label{eq:tilde_J}
   \Tilde{J}_n(x,h,P,\rho) \equiv P_{\rho} \left( \frac{(n-h)^{1/2}}{ \tilde{\sigma} } \mathcal{T} \left( \frac{1}{n-h} \sum_{t=1}^{n-h} X_t \right) \le x \right) ~,
\end{equation} 
where $\Tilde{\sigma}^2$ is the asymptotic variance of $(n-h)^{1/2} \mathcal{T}( (n-h)^{-1} \sum_{t=1}^{n-h} X_t )$. The next theorem shows that the distribution $ \Tilde{J}_n(\cdot,h,P,\rho) $ admits a valid Edgeworth expansion.

\begin{theorem}\label{thm:edgeworth_approx_Jtilde}
Suppose Assumption \ref{appendix:AV_assumptions2} holds. Fix a given $h \in \mathbf{N}$ and $a \in (0,1)$. Then, for any $\rho \in [-1+a, 1-a]$, we have
$$ \sup_{x \in \mathbf{R}} \left | \Tilde{J}_n(x,h,P,\rho) - \left( \Phi(x) + \sum_{j=1}^2 n^{-j/2}q_j(x,h,P,\rho) \phi(x)  \right) \right| = O \left( n^{-3/2} \right)~,$$ 
where $\Tilde{J}_n(x,h,P,\rho)$ is as in \eqref{eq:tilde_J}, $\Phi(x)$ and $\phi(x)$ are the cdf and pdf of the standard normal distribution, and $q_1(x,h,P,\rho)$ and $q_2(x,h,P,\rho)$ are polynomials on $x$ with coefficients that are continuous function of moments of $P$ (up to order 12) and $\rho$. Furthermore, we have $q_1(x,h,P,\rho) = q_1(-x,h,P,\rho)$ and  $q_2(x,h,P,\rho) = -q_2(-x,h,P,\rho)$.
\end{theorem}

The proof of Theorem \ref{thm:edgeworth_approx_Jtilde} is presented in Section \ref{sec:proof_thm_b2} in \textcolor{black}{Online Supplemental} Appendix \ref{sec:appendix_online2}. It relies on \cite{GotzeHipp1983,gotze1994asymptotic} to guarantee the existence of Edgeworth expansion for sample averages and in the results of \cite{bhattacharya1978validity} to complete the proof. 

For the empirical distribution $\hat{P}_n$ defined in \eqref{eq:centered_residuals} and the estimator $\hat{\rho}_n$ defined in \eqref{eq:rho_hat}, we consider the bootstrap sequence $\{z_{b,t}^* : 1 \le t \le n \}$ defined as
$$ z_{b,t}^* = \hat{\rho}_n z_{b,t-1}^* + u_{b,t}^*~, \quad \text{and} \quad z_{b,0}^* =  \sum_{\ell = 0}^{\infty} \hat{\rho}_n^{\ell} u_{b,-\ell}^*~,$$
where $\{u_{b,j}^* : j \le n \}$ is an i.i.d. sequence draw from the distribution $\hat{P}_n$. Define the sequence of random vectors $\{X_{b,t}^* = F(z_{b,t-1}^*,z_{b,t}^*,z_{b,t+h}^*; \hat{\sigma}_n^2, \hat{V}_n, \hat{\rho}_n) : 1 \le t \le n-h \} $, where $F(\cdot)$ is as in \eqref{eq:def_F_appendix} and $\hat{\sigma}_n^2, \hat{V}_n, \hat{\rho}_n$ are the defined using $\hat{P}_n$ and $ \hat{\rho}_n$. 
 
\begin{theorem}\label{thm:edgeworth_approx_Jtilde_bootstrap}
Suppose Assumption \ref{AV_assumptions} holds. Fix a given $h \in \mathbf{N}$ and $a \in (0,1)$. Then, for any $\rho \in [-1+a, 1-a]$ and $\epsilon \in (0,1/2)$, there exist constants $C_1$ and $C_2$ such that  
$$ P \left( \sup_{x \in \mathbf{R}} \left | \Tilde{J}_n(x,h,\hat{P}_n,\hat{\rho}_n) - \left( \Phi(x) + \sum_{j=1}^2 n^{-j/2}q_j(x,h,\hat{P}_n,\hat{\rho}_n) \phi(x)  \right) \right| >   C_1 n^{-3/2} \right)~,$$  
is lower than $C_2 n^{-1-\epsilon}$, where $\Tilde{J}_n(x,h,\cdot,\cdot)$ is as in \eqref{eq:tilde_J} and $X_{b,t}^*$ is replacing $X_{i}$, $\Phi(x)$ and $\phi(x)$ are the cdf and pdf of the standard normal distribution, and $q_1(x,h,\hat{P}_n,\hat{\rho}_n)$ and $q_2(x,h,\hat{P}_n,\hat{\rho}_n)$ are polynomials on $x$ with coefficients that are continuous function of moments of $\hat{P}_n$ (up to order 12) and $\hat{\rho}_n$. Furthermore, we have $q_1(x,h,\hat{P}_n,\hat{\rho}_n) = q_1(-x,h,\hat{P}_n,\hat{\rho}_n)$ and  $q_2(x,h,\hat{P}_n,\hat{\rho}_n) = -q_2(-x,h,\hat{P}_n,\hat{\rho}_n)$.
\end{theorem}

The proof of Theorem \ref{thm:edgeworth_approx_Jtilde_bootstrap} is presented in Section \ref{sec:proof_thm_b3} in \textcolor{black}{Online Supplemental} Appendix \ref{sec:appendix_online2}. It relies on \cite{GotzeHipp1983,gotze1994asymptotic}, \cite{bhattacharya1978validity}, and Lemma \ref{lemma:high-prob-event}.

\bibliographystyle{ecta.bst}

\bibliography{references.bib}


\end{document}


\relax
	\hypersetup{pageanchor=false}
	\hypersetup{pageanchor=true}
    
\author{Amilcar Velez\\
Department of Economics \\ 
Cornell University \\ 
\url{amilcare@cornell.edu}
}

\bigskip

\title{Online Supplementary Material to ``The Local Projection Residual Bootstrap for AR(1) Models''%
}
\date{August 27, 2025}
\maketitle

\setcounter{page}{1}
\onehalfspacing


\counterwithin{figure}{section}
\counterwithin{table}{section}
\renewcommand\thesection{C}
\renewcommand\thesubsection{C.\arabic{subsection}}
\renewcommand{\thetable}{C.\arabic{table}} 
\renewcommand{\thelemma}{C.\arabic{lemma}}
\renewcommand{\theequation}{C.\arabic{equation}}
\setcounter{equation}{0}
 
\section{Proof of Auxiliary Results: Uniform Inference}\label{sec:appendix_online1} 

\subsection{Proof of the Lemma B.1}\label{sec:proof_lemma_b1}
 
\begin{proof}

\noindent \textbf{Notation:} We say a sequence of random variables $Z_n$ is uniformly $O_p\left(1 \right)$ if $\forall \epsilon>0$ there exists $M>0$ and $n_0 \in \mathbf{N}$ such that $P_{\rho}(|Z_n| >M) < \epsilon $ for any $\rho \in [-1,1]$ and $n \ge n_0$. Similarly, $Z_n$ is uniformly $o_p\left(1\right)$ if $\forall \epsilon, \delta >0$  there exists $n_0 \in \mathbf{N}$ such that $P_{\rho}(|Z_n| >\delta) < \epsilon $ for any $\rho \in [-1,1]$ and $n \ge n_0$. 

\noindent \textbf{Item 1:} Consider the following derivation:
\begin{align*}
     g(\rho,n) ~n^{1/2}~ \left(\hat{\rho}_n - \rho \right) = \left( \frac{g(\rho,n)^{-2} \sum_{t=1}^n y_{t-1}^2 }{n}  \right)^{-1} \left( \frac{\sum_{t=1}^n u_t y_{t-1} }{g(\rho,n) n^{1/2}} \right)~,
\end{align*}
where the first term is uniformly $O_p\left(1\right)$ due to Assumption \ref{A3:MO-PM}. The second term is also uniformly $O_p\left(1\right)$ due to the following derivation:
\begin{align*}
  E\left [ \left(  \frac{\sum_{t=1}^n u_t y_{t-1} }{g(\rho,n) n^{1/2}} \right)^2 \right]  = \frac{1}{g(\rho,n)^2 n} \sum_{t=1}^n E[u_t^2y_{t-1}^2]
   \le E[u_t^4]^{1/2} \max_{1\le i \le n} \left (\frac{E[y_{t-1}^4]}{g(\rho,n)^4} \right)^{1/2}
   \le C_8^{1/(4\zeta)} C_{y4}^{1/2}~,
\end{align*}
where the first inequality follows by Cauchy's and algebra manipulation and the second inequality follows by Assumption \ref{xu_assumptions} and part(i) of Lemma MOMT-Y in \cite{xu2022}. The constant $C_{y4}$ depends on the distribution of the sequence $\{u_t: t \ge 1\}$ but does not depend on $\rho$. Therefore, we conclude $g(\rho,n) ~n^{1/2}~ \left(\hat{\rho}_n - \rho \right)$ is uniformly $O_p\left(1\right)$ for any $\rho \in [-1,1]$, which conclude the proof of the lemma.

\textbf{Item 2:}  Recall that  $\tilde{u}_t = \hat{u}_t - n^{-1} \sum_{t=1}^n \hat{u}_t$, where $\hat{u}_t = y_t - \hat{\rho}_n y_{t-1}$ and $\hat{\rho}_n $ is defined in \eqref{eq:rho_hat}. By Bonferroni's inequality, it is sufficient to prove that there exists $N_0 = N_0(\eta)$ such that
\begin{equation}\label{eq:aux1_lemma2}
    P_{\rho} \left( \left | n^{-1}\sum_{t=1}^n \hat{u}_t^2 - \sigma^2  \right| >\sigma^2/4    \right) < \eta/2 
\end{equation}
and
\begin{equation}\label{eq:aux2_lemma2}
    P_{\rho} \left( \left  (  n^{-1}\sum_{t=1}^n \hat{u}_t\right )^2  >\sigma^2/4    \right) < \eta/2  
\end{equation}
for any $n \ge N_0$ and any $\rho \in [-1,1]$. Lemma SIG in \cite{xu2022} adapted for the case of the AR(1) model implies \eqref{eq:aux1_lemma2}. To prove \eqref{eq:aux2_lemma2}, we derive the following inequality
\begin{align*}
\left  (  n^{-1}\sum_{t=1}^n \hat{u}_t\right )^2 &= \left  ( n^{-1} \sum_{t=1}^n u_t + (\rho-\hat{\rho}_n) n^{-1} \sum_{t=1}^n y_{t-1} \right)^2\\
    &\le  2 \left(n^{-1} \sum_{t=1}^n u_t \right)^2 + 2 g(\rho,n)^2(\hat{\rho}_n-\rho)^2 \left( n^{-1} \sum_{t=1}^n \frac{y_{t-1}^2}{g(\rho,n)^2} \right)~,             
\end{align*}
where we used Loeve's inequality (see Theorem 9.28 in \cite{Davidson1994}) in the inequality above.
Note that the first term is uniformly $o_p\left(1\right)$ due to the law of large numbers for $\alpha$-mixing sequences (see Corollary 3.48 in \cite{white2014asymptotic}) and Assumption \ref{xu_assumptions}. Since $ g(\rho,n)^2(\hat{\rho}_n-\rho)^2$ is uniformly $o_p\left(1\right)$ due to Part 1, it is sufficient to prove that $ n^{-1} \sum_{t=1}^n \frac{y_{t-1}^2}{g(\rho,n)^2} $ is uniformly $O_p\left(1\right)$. The last claim follows by the next inequality
$$ E \left [  n^{-1} \sum_{t=1}^n \frac{y_{t-1}^2}{g(\rho,n)^2} \right] \le n^{-1} \sum_{t=1}^n \left(E \left[  \frac{y_{t-1}^4}{g(\rho,n)^4}  \right]\right)^{1/2} \le C_{y4}^{1/2}~,$$
where the last inequality follows by part(i) of Lemma MOMT-Y in \cite{xu2022}. The constant $C_{y4}$ depends on the distribution of the sequence $\{u_t: t \ge 1\}$ but does not depend on $\rho$. Therefore, $ n^{-1} \sum_{t=1}^n \frac{y_{t-1}^2}{g(\rho,n)^2} $ is uniformly $O_p\left(1\right)$, which concludes the proof of the lemma.

\textbf{Item 3:} Recall that  $\tilde{u}_t = \hat{u}_t - n^{-1} \sum_{t=1}^n \hat{u}_t$, where $\hat{u}_t = y_t - \hat{\rho}_n y_{t-1}$ and $\hat{\rho}_n $ is defined in \eqref{eq:rho_hat}. By Loeve's inequality (see Theorem 9.28 in \cite{Davidson1994}), we obtain 
$$  (\hat{u}_t - n^{-1} \sum_{t=1}^n \hat{u}_t)^4 = ((\hat{u}_t - u_t) + n^{-1} \sum_{t=1}^n \hat{u}_t +u_t )^4 \le 3^3 \left((\hat{u}_t - u_t )^4 + \left( n^{-1} \sum_{t=1}^n \hat{u}_t \right)^4 +  u_t^4 \right)~.  $$
Therefore, it is sufficient to prove that there exists $N_0 = N_0(\eta)$ and $\Tilde{K}_4$ such that
\begin{align}
    P \left(   n^{-1} \sum_{t=1}^n (\hat{u}_t - u_t )^4   >\Tilde{K}_4/81    \right) &< \eta/3 ~,\label{eq:aux1_lemma3}\\
    P \left(    \left(n^{-1} \sum_{t=1}^n \hat{u}_t \right)^4   >\Tilde{K}_4/81    \right) &< \eta/3 ~,\label{eq:aux2_lemma3} \\
    P \left(   n^{-1}  \sum_{t=1}^n u_t^4  >\Tilde{K}_4/81    \right)  &< \eta/3 ~.\label{eq:aux3_lemma3}   
\end{align}
%
To prove \eqref{eq:aux1_lemma3}, we use  $\hat{u}_t - u_t = (\rho - \hat{\rho}_n) y_{t-1}$, the following equality
\begin{align*}
     n^{-1} \sum_{t=1}^n (\hat{u}_t - u_t )^4 &= (\hat{\rho}_n-\rho)^4~g(\rho,n)^4  ~n^{-1}\sum_{t=1}^n \frac{y_{t-1}^4}{g(\rho,n)^4}  ~,
\end{align*}
Markov inequality, and part (i) of Lemma MOMT-Y in \cite{xu2022}. To verify \eqref{eq:aux2_lemma3}, we use \eqref{eq:aux2_lemma2} from the proof of Part 2. Finally, Markov's inequality and  Assumption \ref{xu_assumptions} implies \eqref{eq:aux3_lemma3}.  
\end{proof}

\subsection{Proof of the Lemma B.2}\label{sec:proof_lemma_b2} 

\begin{proof}
We first prove that there exists $\tilde{M} = \tilde{M}(M)>0$ and $N_0=N_0(M)>0$ such that
\begin{equation}\label{eq:aux1_algebra_step1}
    \rho +  \frac{M}{n^{1/2}g(\rho,n)} \le 1 + \frac{\tilde{M}}{n}~,
\end{equation}
for all $n \ge N_0$. Note that \eqref{eq:aux1_algebra_step1} is sufficient to conclude that $\hat{\rho}_n \le 1 + \tilde{M}/n$ since $\hat{\rho}_n \le \rho + Mn^{-1/2}/g(\rho_n,n)$.

Let us prove \eqref{eq:aux1_algebra_step1} by contradiction. That is: suppose that there exist sequences $\rho_k$, $\tilde{M}_k \to \infty$, and $n_k \to \infty$ such that 
$     \rho_k +  M / (n_k^{1/2}g(\rho_k,n_k)) > 1 + \tilde{M}_k/n_k,$ 
for all $k$. The previous expression is equivalent to
\begin{equation}\label{eq:aux2_algebra_step1}
     M > n_k^{1/2}g(\rho_k,n_k) (1-\rho_k) + \tilde{M}_k\frac{g(\rho_k,n_k)}{n_k^{1/2}}~.
\end{equation}
Define $a_k = n_k(1-|\rho_{n_k}|)$. Consider the derivation to get a lower bound for $g(\rho_k,n_k)$:
\begin{align*}
    g(\rho_k,n_k)^2 &= 1 + (1-a_k/n_k)^2 + ... + (1-a_k/n_k)^{2(n_k-1)} \\
    &= \frac{\{1-(1-a_k/n_k)^{n_k} \} \{1+(1-a_k/n_k)^{n_k}\}  }{a_k/n_k \{2-a_k/n_k\}}\\
    &\ge \frac{n_k}{a_k} \times \frac{1-e^{-a_k}}{2}~,
\end{align*}
where the last inequality use that $(1-a_k/n_k)^{n_k} = \exp( n_k \log(1-a_k/n_k)) \le \exp(-a_k)$. Without loss of generality, suppose that $a_k \to a \in [0,+\infty]$; otherwise, we can use a subsequence. We now consider two cases. For the first case, suppose $a_k \to +\infty$. This implies that
$$ n_k^{1/2}g(\rho_k,n_k) (1-\rho_k) \ge   \left(  \frac{1-e^{-a_k}}{2}\right)^{1/2}  a_k^{1/2} \to \infty ~,$$
which contradicts \eqref{eq:aux2_algebra_step1}. For the second case, suppose $a_k \to a$. This implies that
$$ \frac{\tilde{M}_k}{\sqrt{n_k}} g(\rho_k,n_k)  \ge \left(  \frac{1-e^{-a_k}}{2a_k}\right)^{1/2} \tilde{M}_k \to \infty~, $$
which contradicts \eqref{eq:aux2_algebra_step1}. Therefore, there exists $N_0 = N_0(M)$ and $\tilde{M} = \tilde{M}(M)$ such that \eqref{eq:aux1_algebra_step1} holds for $n \ge N_0$. We can adapt the proof to conclude that $\hat{\rho}_n > -1 - \tilde{M}/n$  for all $n \ge N_0$.
\end{proof}


\subsection{Proof of the Lemma B.3}\label{sec:proof_lemma_b3} 
 
\begin{proof}
We prove only item 1 since the proof of item 2 is analogous. The proof of item 1 has three steps. First, we can write
$ P_{\rho}\left( |R_{n,b}^*(h)| \le x \mid Y^{(n)} \right) - (2\Phi(x)-1) = I_1 + I_2 + I_3,$
where $I_1 = J_n(x,h,\hat{P}_n,\hat{\rho}_n) - \Phi(x) $, $I_2  = \Phi(-x)  - J_n(-x,h,\hat{P}_n,\hat{\rho}_n) $, and
$$ I_3 = P_{\rho}(R_{n,b}^*(h)=-x \mid Y^{(n)}) \le P_{\rho}(R_{n,b}^*(h) \in (-x-\epsilon/2,-x+\epsilon/2] \mid Y^{(n)})~.$$
Second, conditional on the event $E_n$ defined in the proof of Theorem \ref{thm:consistency}, the inequality \eqref{eq:aux0_proof_thm3.1} in the proof of Theorem \ref{thm:consistency} implies that $|I_1| < \epsilon/2$ and $|I_2| < \epsilon/2$ for any $n \ge N_2= N_2(\epsilon,\eta)$ and any $h \le h_n$ such that $h_n \le n$ and $h_n = o\left(n\right)$. Also, the inequality \eqref{eq:aux0_proof_thm3.1} and algebra manipulation implies $|I_3|<2\epsilon$. Therefore, we conclude
\begin{equation}\label{eq:aux1_lemma_quantiles}
    \sup_{h \le h_n} \: \sup_{x \in \mathbf{R}} ~ \left|  P_{\rho} \left(|R_{n,b}^*(h)|\le x \mid Y^{(n)} \right) - (2\Phi(x)-1)   \right| < 3\epsilon~,
\end{equation}
for any $n \ge N_2$ and any $h_n\le n$ such that $h_n = o\left(n\right)$. Third, taking $x=z_{1-\alpha/2-3\epsilon/2}$ in \eqref{eq:aux1_lemma_quantiles}, it follows that
$\left|  P_{\rho} \left(|R_{n,b}^*(h)|\le z_{1-\alpha/2-3\epsilon/2} \mid Y^{(n)} \right) - (1-\alpha-3\epsilon)   \right| < 3\epsilon, $
which implies
$$   P_{\rho} \left(|R_{n,b}^*(h)|\le z_{1-\alpha/2-3\epsilon/2} \mid Y^{(n)} \right)  \le 1-\alpha~.$$
By definition of  $c_n^*(h,1-\alpha) $ as in \eqref{eq:cv_bootstrap}, it follows that $c_n^*(h,1-\alpha) \ge z_{1-\alpha/2-3\epsilon/2}$ holds conditional on the event $E_n$. We similarly obtain that $c_n^*(h,1-\alpha) \le z_{1-\alpha/2+3\epsilon/2}$ holds conditional on $E_n$.
\end{proof}

\subsection{Proof of Proposition \ref{proposition:A3} }\label{sec:proof_prop_b1}

\begin{proof}[Proof of Proposition \ref{proposition:A3}]

We use the general subsequence approach of \cite{andrews2020generic} to show that the uniform result in the proposition holds. We prove that for any sequence $\{ \rho_n: n \ge 1\}$ such that $|\rho_n|\le 1+M/n$ and any sequence $\{\sigma_n^2: n \ge 1\} \subset [\underline{\sigma},\:\overline{\sigma}]$, there exists subsequences $\{ \rho_{n_k}: k \ge 1\}$ and $\{\sigma_{n_k}^2: k \ge  1 \}$ such that 
 \begin{equation}\label{eq:aux1_lemma1}
 \lim_{K \to \infty} ~ \lim_{k \to \infty}  ~ P \left(~  g(\rho_{n_k},n_k)^{-2} n_k^{-1} \sum_{i=1}^{n_k}y_{n_k,i-1}^2  \ge 1/K ~\right) = 1~.  
\end{equation}

We consider two cases to prove \eqref{eq:aux1_lemma1}. The first case is $n_k(1-|\rho_{n_k}|) \to \infty$ for some subsequence $\{n_k: k \ge 1\}$, which considers the subsequence of $\rho_{n}$ that stay on the stationary region or go to the boundary at slower rates. The second case is $n_k(1-|\rho_{n_k}|) \to c  \in [-M,+\infty)$ for some subsequence $\{n_k: k \ge 1\}$, which  considers the subsequence of $\rho_{n}$ that goes to the boundary (local-unit-model) or are on it (unit-root model). For both cases, we assume $\sigma_{n_k}^2 \to \sigma_0^2$ since any sequence $\{\sigma_n^2: n \ge 1\} \subset [\underline{\sigma},\:\overline{\sigma}]$ has always a convergent subsequence. 
To avoid complicated sub-index notation, we present the algebra derivation using the original sequence. 

\noindent \textbf{Case 1:} Suppose $n(1-|\rho_{n}|) \to \infty$  and $\sigma_{n}^2 \to \sigma_0^2$ as $n \to \infty$. This condition implies that there  exists $N_0$ such that $|\rho_n|\le 1$ for all $n \ge N_0$, otherwise there is a subsequence $|\rho_{n_k}|$  in $(1,1+M/{n_k}]$ but this cannot occur since $n_k(1-|\rho_{n_k}|) \in [-M,0]$. As a result, we have $g(\rho_{n},n)^2 = \sum_{\ell=0}^{n-1} ~ \rho_n^{2 \ell} \le (1-\rho_{n}^2)^{-1}$ that implies
$$  P \left(~  g(\rho_{n},n)^{-2} ~ n^{-1} \sum_{t=1}^{n}y_{n,t-1}^2  \ge 1/K ~\right) \ge P \left(~ \frac{(1-\rho_n^2)}{n} \sum_{t=1}^{n} y_{n,t-1}^2~ \ge 1/K ~\right)~. $$
Therefore, to verify \eqref{eq:aux1_lemma1} is sufficient to prove that
 \begin{equation*}
 \lim_{K \to \infty} ~ \lim_{n \to \infty}  ~ P \left(~ \frac{(1-\rho_n^2)}{n} \sum_{t=1}^{n} y_{n,t-1}^2~ \ge 1/K ~\right) = 1~,  
\end{equation*}
which follows if we prove
\begin{equation}\label{eq:aux2_lemma1}
     \frac{(1-\rho_{n}^2)  }{ n} \sum_{t=1}^{n} y_{n,t-1}^2\overset{p}{\to} \sigma_0^2~.
\end{equation}
 
We prove \eqref{eq:aux2_lemma1} in two steps. 

\noindent \textit{Step 1:} Using Assumption \ref{appendix:AV_assumptions} and $y_{n,t-1}= \sum_{\ell=1}^{i-1} \rho_n^{i-1-\ell} u_{n,\ell}$, we derive the following:
\begin{align*}
    E\left[\frac{(1-\rho_n^2)}{n} \sum_{t=1}^n y_{n,t-1}^2 \right] &= \frac{(1-\rho_n^2)}{n}\sum_{t=1}^n \sum_{\ell = 1}^{n-1} E[u_{n,\ell}^2] ~\rho_n^{2(t-1-\ell)} ~ I\{ 1 \le \ell \le i-1 \} \\
    %
    %
    &= \frac{\sigma_n^2(n-1)}{n} - \frac{\sigma_n^2}{n} \sum_{\ell = 1}^{n-1} \rho_n^{2(n-\ell)}~.
\end{align*}
We conclude the right-hand side of the previous display converges to $\sigma_0^2$ since  $\sigma_n^2 \to \sigma_0^2$, $n^{-1}\sum_{\ell = 1}^n \rho_n^{2(n-\ell)} \le \{n(1-|\rho_n|)(1+|\rho_n|)\}^{-1}$, and $n(1-|\rho_n|) \to \infty$.

\noindent \textit{Step 2:} We use $E[y_{n,t-1}^2] = g(\rho_n,t-1)^2\sigma_n^2$ to derive the following decomposition 
\begin{align*}
   \sum_{i=1}^n y_{n,t-1}^2 - E[y_{n,t-1}^2]  
   = \sum_{\ell=1}^{n-1} \left( u_{n,\ell}^2 - \sigma_n^2 \right) b_{n,\ell}
   + 2  \sum_{\ell=1}^{n-1} u_{n,\ell} ~ d_{n,\ell} ~,
\end{align*}
where $ b_{n,\ell} =  \sum_{i=1+\ell}^n  \rho_n^{2(i-1-\ell)}$ and $d_{n,\ell} =  b_{n,\ell}\sum_{\ell_2 = 1}^{\ell-1}  u_{n,\ell_2} \rho_n^{\ell-\ell_2}$. Note $d_{n,\ell}$ is measurable with respect to the $\sigma$-algebra defined by $\{u_{n,k}: 1 \le k \le \ell-1\}$.

The decomposition above, Loeve's inequality (see Theorem 9.28 in \cite{Davidson1994}), and Assumption \ref{appendix:AV_assumptions} imply that the variance of $(1-\rho_{n}^2)n^{-1}\sum_{t=1}^{n} y_{n,t-1}^2 $ is lower than
\begin{align*}
    \frac{2(1-\rho_{n}^2)^2}{n^2}  \left( \sum_{\ell=1}^{n-1} E \left[ \left( u_{n,\ell}^2 - \sigma_n^2 \right)^2 \right]  b_{n,\ell}^2
    +
    4 \sum_{\ell=1}^{n-1} E \left[ u_{n,\ell}^2 \right] E \left[d_{n,\ell}^2 \right]  \right)~.
\end{align*}
Since $b_{n,\ell}^2 \le (1-\rho_n^2)^{-2} $ and $E[d_{n,\ell}^2] = b_{n,\ell}^2 \sigma_n^2 \sum_{\ell_2=1}^{\ell-1} \rho_n^{2(\ell-\ell_2)} \le \sigma_n^2 (1-\rho_n^2)^{-3} $ for $\ell=1,\ldots,n-1$, and $E[(u_{n,\ell}^2-\sigma_n^2)^2] \le E[u_{n,\ell}^4] \le K_4$ by Assumption \ref{appendix:AV_assumptions}, the previous display is lower than
$$ \frac{2(n-1)K_4}{n^2} + \frac{8(n-1)\sigma_n^4}{n^2(1-\rho_n^2)}~,$$
which goes to 0 since $\sigma_n^4 = E[u_{n,\ell}^2]^2 \le E[u_{n,\ell}^4] \le K_4$ and $n(1-\rho_n^2) = n(1-|\rho_n|)(1+|\rho_n|) \to \infty$ as $n \to \infty$. This proves that the variance of the left-hand side on \eqref{eq:aux2_lemma1} goes to zero, which proves \eqref{eq:aux2_lemma1} due to step 1.

\noindent \textbf{Case 2:} Suppose $n(1-|\rho_{n}|) \to c \in [-M,+\infty)$ and $\sigma_{n}^2 \to \sigma_0^2$ as $n \to \infty$. We first observe that $g(\rho_n,n)^2 \le n \exp(2M)$ due to $|\rho_n| \le 1 + M/n$ and the following derivation:
\begin{align*}
    g(\rho_n,n)^2 &= \sum_{\ell=0}^{n-1} \rho_n^{2\ell}  \le n (1+M/n)^{2n} = n \exp \left( 2n \log \left( 1+M/n \right)  \right) \le n \exp \left( 2M  \right)~,
\end{align*}
where we used that $\log(1+x) \le x$ for all $x>-1$. By the previous observation
$$  P \left(~  g(\rho_{n},n)^{-2} n^{-1} \sum_{t=1}^{n}y_{n,t-1}^2  \ge 1/K ~\right) \ge P \left(~ \frac{\exp(-2M)}{n^2} \sum_{t=1}^{n} y_{n,t-1}^2~ \ge 1/K ~\right)~, $$
where $\exp(-2M)$ is a constant that does not change as $n\to \infty$ and $K \to \infty$. Therefore, it is sufficient to prove that
$
 \lim_{K \to \infty} ~ \lim_{n \to \infty}  ~ P \left(~ \frac{1}{n^2} \sum_{t=1}^{n} y_{n,t-1}^2~ \ge 1/K ~\right) = 1~
$
to verify \eqref{eq:aux1_lemma1},
which follows if we prove
\begin{equation}\label{eq:aux3_lemma1}
     \frac{1}{n^2} \sum_{t=1}^{n} y_{n,t-1}^2 \overset{d}{\to} \sigma_0^2 \int_0^1 J_{-c}(r)^2dr~,
\end{equation} 
where $J_{-c}(r) = \int_0^r e^{-(r-s)c} dW(s)$ and $W(s)$ is a standard Brownian motion.  

To prove \eqref{eq:aux3_lemma1}, we rely on the results and techniques presented in \cite{phillips1987towards}. Specifically, we adapt his Lemma 1 part (c) for the sequence of models and the drifting parameter that we consider in this paper. We proceed in two steps. First, we construct a triangular array $\{\tilde{y}_{n,t}: 1 \le t \le n, n \ge 1\}$ that verify \eqref{eq:aux3_lemma1}. Then, we prove that the constructed sequence verifies that
$  n^{-2} \sum_{t=1}^{n} y_{n,t-1}^2 -  n^{-2} \sum_{t=1}^{n} \tilde{y}_{n,t-1}^2 = o_p(1)$.

\noindent \textit{Step 1:} Define $\tilde{u}_{n,t} = u_{n,t} I\{ \rho_n \ge 0 \} +(-1)^t u_{n,t} I\{ \rho_n < 0 \} $ for all $t=1,\ldots,n$. Note that the sequence $\{\tilde{u}_{n,t}: 1 \le t \le n \}$ defines a martingale difference array with the same variance $E[\tilde{u}_{n,t}^2] = \sigma_n^2$ and satisfies that $E[\tilde{u}_{n,t}^2] \in [\underline{\sigma},\overline{\sigma}]$, and $E[\tilde{u}_{n,t}^4] \le K_4$. Using this notation, we construct the following triangular array:
$$ \tilde{y}_{n,t} = e^{-c/n} \tilde{y}_{n,t-1} + \tilde{u}_{n,t}~, \quad \tilde{y}_{n,0} = 0~,  $$
where $c = \lim_{n\to \infty} n(1-|\rho_n|)$. Denote the sequence of partial sums by $S_{n,j} = \sum_{t=1}^j u_{n,t}~$ for any $j=1,\ldots,n$ and $S_{n,0}=0$. Let us define the following random process
\begin{equation*}
    X_n(r) = \frac{1}{\sqrt{n}} \frac{1}{\sigma_n} S_{n,[nr]} = \frac{1}{\sqrt{n}} \frac{1}{\sigma_0} S_{n,j-1} ~\text{ if } (j-1)/n \le r < j/n~,
\end{equation*}
and $X_n(1) = \frac{1}{\sqrt{n}} \frac{1}{\sigma_n} S_{n,n}$. By a functional central limit theorem for martingale difference arrays (see Theorem 27.14 in \cite{Davidson1994}), we claim that $\{X_n(r): r \in [0,1]\}$ converges to the standard Brownian motion process $\{W(r): r \in [0,1]\}$. To use this result, we prove
\begin{align*}
    \text{(a)}\quad \sum_{t=1}^{n} \frac{\tilde{u}_{n,t}^2}{n \sigma_n^2} \overset{p}{\to} 1~,\quad %
    \text{(b)}\quad \max_{t \le n} \left|\frac{\tilde{u}_{n,t}}{ \sqrt{n} \sigma_n} \right| \overset{p}{\to} 0~,\quad
    \text{(c)}\quad\lim_n \sum_{t=1}^{[nr]} \frac{E[\tilde{u}_{n,t}^2]}{n \sigma_n^2} = r~.
\end{align*}
%
%
We can verify condition (a) using $\tilde{u}_{n,t}^2 = {u}_{n,t}^2$, Chebyshev's inequality and Assumption \ref{appendix:AV_assumptions}:
\begin{align*}
    P\left( \left | \sum_{t=1}^{n} \frac{\tilde{u}_{n,t}^2}{n \sigma_n^2} - 1 \right| > \epsilon \right) \le \frac{n^{-2}}{\epsilon^2}   \sum_{i=1}^n E[{u}_{n,t}^4] \le \frac{K_4}{\epsilon^2 n} \to 0 \quad \text{as }n \to \infty~,
\end{align*}
for any $\epsilon >0$. To verify condition (b) holds is sufficient to show that
$$ n E \left[\frac{\tilde{u}_{n,t}^2}{ n\sigma_n^2} I \left\{\left|\frac{\tilde{u}_{n,t}}{ \sqrt{n} \sigma_n} \right| > c \right\} \right] \to 0~,$$
for any $c>0$, where $I\{ \cdot \}$ is the indicator function. If the previous display holds, then condition (b) follows by theorem 23.16 in \cite{Davidson1994}. To verify the previous condition, note that
\begin{align*}
    n E \left[\frac{\tilde{u}_{n,t}^2}{ n\sigma_n^2} I \left\{\left|\frac{\tilde{u}_{n,t}}{ \sqrt{n} \sigma_n} \right| > c \right\} \right] 
    \le
    n E \left[\frac{\tilde{u}_{n,t}^4}{ n^2\sigma_n^4c^2} I \left\{\left|\frac{\tilde{u}_{n,t}}{ \sqrt{n} \sigma_n} \right| > c \right\} \right]  
    \le \frac{E[\tilde{u}_{n,t}^4]}{ n \sigma_n^4c^2}   ~\le ~ \frac{K_4}{n \underline{\sigma}^2 c^2}~,
\end{align*}
where the last inequality uses $\tilde{u}_{n,t}^4 = {u}_{n,t}^4$ and Assumption \ref{appendix:AV_assumptions}. Finally, condition (c) holds since $E[\tilde{u}_{n,t}^2] = \sigma_n^2$. 

Using the functional central limit theorem, the continuous mapping theorem, and $\sigma_n \to \sigma_0$, we can repeat the arguments presented in the proof of Lemma 1 in \cite{phillips1987towards} to conclude that
$ n^{-2}\sum_{t=1}^{n} \tilde{y}_{n,t-1}^2 \overset{d}{\to} \sigma_0^2 \int_0^1 J_{-c}(r)^2dr~.$

\noindent \textit{Step 2:} Define $a_n = |\rho_n|e^{c/n}$. We know $y_{n,t} =  \sum_{\ell=1}^i  \rho_n^{i-\ell} u_{n,t}$ and $\tilde{y}_{n,t} =  \sum_{\ell=1}^i   e^{-c(i-\ell)/n} \tilde{u}_{n,t}$; therefore, $\tilde{y}_{n,t} = y_{n,t}-R_{n,t}$ if $\rho_n \ge 0$, and $\tilde{y}_{n,t} = (-1)^t y_{n,t}-R_{n,t}$ if $\rho_n < 0$,
%
            %
%
where
$$R_{n,t} =  \sum_{\ell=1}^t \left(  a_n^{t-\ell} -1 \right) e^{-c(t-\ell)/n} \tilde{u}_{n,\ell}~.$$
Therefore, we conclude that $y_{n,t}^2 = \tilde{y}_{n,t}^2 + 2\tilde{y}_{n,t} R_{n,t} + R_{n,t}^2$. This implies that
\begin{equation*}
    \left | \frac{1}{n^2}\sum_{t=1}^n y_{n,t-1}^2 - \frac{1}{n^2}\sum_{t=1}^n \tilde{y}_{n,t-1}^2 \right| \le  \left| \frac{2}{n^2}\sum_{t=1}^n \tilde{y}_{n,t-1} R_{n,t} \right| + \left| \frac{1}{n^2}\sum_{t=1}^n   R_{n,t}^2 \right|~.
\end{equation*}
By Cauchy--Schwartz's inequality, the right-hand side of the previous expression is lower than or equal to
\begin{equation}\label{eq:aux4_lemma1}
    2 \left( \frac{1}{n^2}\sum_{i=1}^n \tilde{y}_{n,t-1}^2 \right)^{1/2} \left(  \frac{1}{n^2}\sum_{i=1}^n   R_{n,t}^2\right)^{1/2} + \left| \frac{1}{n^2}\sum_{i=1}^n   R_{n,t}^2 \right|~.
\end{equation}
By the result at the end of Step 1, we have $n^{-2}\sum_{i=1}^n \tilde{y}_{n,t-1}^2$ is $O_p\left(1\right)$. Therefore, it is sufficient to show 
$n^{-2} \sum_{i=1}^n   R_{n,t}^2 \overset{p}{\to} 0~$
to conclude that \eqref{eq:aux4_lemma1} converges to zero in probability.

To verify the claim, we first observe that  $a_n^je^{-(i-\ell)c/n} = |\rho_n|^j e^{-(i-\ell-j)c/n} \le |\rho_n|^j$ for all $j=0,\ldots,i-\ell-1$, which implies that
\begin{align*}
    |R_{n,t}| &= \left |\sum_{\ell=1}^t (a_n-1)  (1+a_n+...+a_n^{i-\ell-1}) e^{-(t-\ell)c/n}  \tilde{u}_{n,\ell} \right| \\
    &\le  |a_n-1| \sum_{\ell=1}^t   (1+|\rho_n|+...+|\rho_n|^{t-\ell-1})  |\tilde{u}_{n,\ell}|~.
\end{align*}
Using the previous inequality and $|\rho_n|^j \le (1+M/n)^j \le (1+M/n)^n \le e^M$, we obtain
$$ |R_{n,t}| \le e^M |a_n-1|  \sum_{\ell=1}^t  (t-\ell) |\tilde{u}_{n,\ell}| \le e^M |n(a_n-1)| \sum_{\ell=1}^n  |u_{n,\ell}|~,$$
for all $t=1,\ldots,n$, where we used that $|\tilde{u}_{n,\ell}| = |u_{n,\ell}|$ in the last inequality. Then, we derive
$$ n^{-2} \sum_{t=1}^n   R_{n,t}^2 \le e^{2M} |n(a_n-1)|^2 \left(n^{-1}\sum_{\ell=1}^n  |u_{n,\ell}|\right)^2~.$$
By Markovs's inequality and Assumption \ref{appendix:AV_assumptions}, we obtain that $n^{-1}\sum_{\ell=1}^n  |u_{n,\ell}|$ is $O_p\left(1\right)$. Analyzing $a_n-1 = e^{c/n} \left(|\rho_n|-e^{-c/n} \right)$, we can conclude that $n(a_n-1) = o\left(1 \right)$, which implies that the right-hand side of the previous display converges to zero in probability. As a result, we conclude that \eqref{eq:aux4_lemma1} converges to zero in probability, which implies that
$ n^{-2}\sum_{i=1}^n y_{n,t-1}^2 - n^{-2}\sum_{i=1}^n \tilde{y}_{n,t-1}^2 = o_p\left(1 \right),$
and by the result at the end of Step 1, we conclude \eqref{eq:aux3_lemma1}.
\end{proof}

\subsection{Proof of Theorem \ref{thm:sequence_models}}\label{sec:proof_thm_b1}

\textbf{Additional Notation:} 
Define $\xi_{n,t}(\rho_n,h_n) =   y_{n,t+h_n} - \beta(\rho_n,h_n) y_{n,t}$ and recall $\beta(\rho,h) \equiv \rho^{h}$. Algebra shows
\begin{equation}\label{eq:appendix_A_xi}
    \xi_{n,t}(\rho_n,h_n) =  \sum_{\ell=1}^{h_n} \rho_n^{h_n-\ell} u_{n,t+\ell}~~.
\end{equation}


\begin{proof}
The derivations presented on pages 1811 and 1812 in \cite{MO-PM2020} imply  
$$ R_n(h_n) = \frac{ \sum_{t=1}^{n-h_n} \xi_{n,t}(\rho_n,h_n)\hat{u}_{n,t}(h_n)}{\left(\sum_{t=1}^{n-h_n}\hat{\xi}_{n,t}(h_n)^2 \hat{u}_{n,t}(h_n)^2\right)^{1/2}}~,$$
which is equal to 
$$  \left( \frac{\sum_{t=1}^{n-h_n} \xi_{n,t}(\rho_n, h_n) u_{n,t}}{(n-h_n)^{1/2}V(\rho_n,h_n)^{1/2}} + \frac{\sum_{t=1}^{n-h} \xi_{n,t}(\rho_n,h_n)  (\hat{u}_{n,t}(h_n) -u_{n,t}) }{(n-h_n)^{1/2} V(\rho_n,h_n)^{1/2}} \right) \times \left(\frac{(n-h_n)  V(\rho_n,h_n) }{\sum_{t=1}^{n-h_n}\hat{\xi}_{n,t}(h_n)^2 \hat{u}_{n,t}(h_n)^2 } \right)^{1/2},$$
where $V(\rho,h) = E[\xi_{n,t}(h)^2 u_{n,t}^2]$. We then follow their approach and prove that under Assumption \ref{appendix:AV_assumptions}:  for any sequences $\{ \rho_n: n \ge 1\} \subset [-1 -M/n,1 + M/n]$, $\{\sigma_n^2: n \ge 1\} \subset [\underline{\sigma},\:\overline{\sigma}]$, and $\{h_n: n \ge 1\}$ satisfying $h_n = o(n)$ and $h_n \le n$, we have
\begin{align*}
  (i) & \quad \frac{\sum_{t=1}^{n-h_{n}} \xi_{n,t}(\rho_n, h_{n}) u_{n,t}}{(n-h_{n})^{1/2}V(\rho_{n},h_{n})^{1/2}} \overset{d}{\to} N(0,1)~,\\
  (ii) & \quad \frac{\sum_{t=1}^{n-h_n} \xi_{n,t}(\rho_n,h_n)  (\hat{u}_{n,t}(h_n) -u_{n,t}) }{(n-h_n)^{1/2} V(\rho_n,h_n)^{1/2}}  \overset{p}{\to} 0~,\\
  (iii) &  \quad \frac{\sum_{t=1}^{n-h_n} \hat{\xi}_{n,t}(h_n)^2 \hat{u}_{n,t}(h_n)^2}{(n-h_n) V(\rho_n,h_n)} \overset{p}{\to} 1~.
\end{align*}
Finally, Lemmas \ref{MO-PM:lemmaA1}, \ref{MO-PM:lemmaA4_E8}, and \ref{MO-PM:lemmaA2} imply (i), (ii), and (iii), respectively. 
\end{proof}


\begin{lemma}\label{MO-PM:lemmaA7}
Suppose Assumption \ref{appendix:AV_assumptions} holds. Then, for any $(\rho_n ,\sigma_n,h_n)$ such that $ |\rho_n| \le  1+M/n$, $\sigma_n^2 \in [\underline{\sigma},\:\overline{\sigma}]$, and $h_n \le n$, we have
$$ E\left[ \xi_{n,t}(\rho_n,h_n)^4\right] \le  4g(\rho_n,h_n)^4 K_4~,$$
where $g(\rho,k) = \left(\sum_{\ell=0}^{k-1} ~ \rho^{2 \ell}\right)^{1/2}$ and $\xi_{n,t}(\rho_n,h_n)$ is in \eqref{eq:appendix_A_xi}.
\end{lemma}

\begin{proof}
It follows from the proof of Lemma A.7 in \cite{MO-PM2020}.
\end{proof}

\begin{lemma}\label{MO-PM:lemmaA4}
Suppose Assumption  \ref{appendix:AV_assumptions} holds. Then, for any $(\rho_n ,\sigma_n,h_n)$ such that $ |\rho_n| \le  1+M/n$, $\sigma_n^2 \in [\underline{\sigma},\:\overline{\sigma}]$, and $h_n \le n$, we have
\begin{align*}
E\left[ \left( \frac{ \sum_{t=1}^{n-h_n} \xi_{n,t}(\rho_n,h_n) y_{n,t-1}}{(n-h_n)g(\rho_n,n-h_n)g(\rho_n,h_n)\sigma_n^2 } \right)^2 \right]  \le \frac{n }{n-h_n} \times \frac{h_n}{n-h_n} \times  \frac{\sqrt{4 K_4}}{\underline{\sigma}} ~,
\end{align*}
where $g(\rho,k) = \left(\sum_{\ell=0}^{k-1} ~ \rho^{2 \ell}\right)^{1/2}$ and $\xi_{n,t}(\rho_n,h_n)$ is in \eqref{eq:appendix_A_xi}.
\end{lemma}

\begin{proof}
The definition of $\xi_{n,t}(\rho_n,h_n)$ in \eqref{eq:appendix_A_xi} implies
$ \sum_{t=1}^{n-h_n} \xi_{n,t}(\rho_n,h_n) y_{n,t-1} = \sum_{j=1}^{n}  u_{n,j} b_{n,j},$
%
    %
    %
%
where
$ b_{n,j} =  \sum_{t=j-h_n}^{j-1} \rho_n^{t+h_n-j}~ y_{n,t-1} I\{ 1 \le t \le n-h \}.$
Note that $b_{n,j}$ is measurable with respect to the $\sigma$-algebra defined by $\{u_{n,k}: 1 \le k \le j-2\}$. Using Assumption  \ref{appendix:AV_assumptions}, we obtain
$$ E \left[ \left(\sum_{j=1}^{n}  u_{n,j} b_{n,j}\right)^2\right] = \sum_{j=1}^{n}  E[ u_{n,j}^2 b_{n,j}^2] = \sigma_n^2 \sum_{j=1}^{n}  E[ b_{n,j}^2]~.$$
Therefore, the derivation above implies
$ E \left[ \left( \sum_{t=1}^{n-h_n} \xi_{n,t}(\rho_n,h_n) y_{n,t-1}   \right)^2\right] = \sigma_n^2 \sum_{j=1}^{n}  E[ b_{n,j}^2].$
We claim that
\begin{equation}\label{eq:aux1_lemmaA4}
    E[ b_{n,j}^2] \le h_n g(\rho_n,h_n)^2 g(\rho_n,n-h_n)^2 \sqrt{4 K_4} ~,
\end{equation}
for any $j=1,\ldots,n$. The previous claim and Assumption  \ref{appendix:AV_assumptions} imply
$$  E\left[ \left( \frac{ \sum_{t=1}^{n-h_n} \xi_{n,t}(\rho_n,h_n) y_{n,t-1}}{(n-h_n)g(\rho_n,n-h_n)g(\rho_n,h_n)\sigma_n^2 } \right)^2 \right] \le \frac{n h_n  }{(n-h_n)^2\sigma_n^2} \times \sqrt{4 K_4}  \le \frac{n h_n  }{(n-h_n)^2\underline{\sigma}} \times \sqrt{4 K_4}  ~. $$

To verify \eqref{eq:aux1_lemmaA4}, we consider three cases. The first case is $j \le h_n$, in which we derive
\begin{align*}
    E[b_j^2] = E[  (\sum_{t=1}^{j-1} \rho^{t+h-j} y_{n,t-1} )^2] 
    &\le (j-1) E[ \sum_{t=1}^{j-1} \rho^{2(t+h-j)} y_{n,t-1}^2 ] \\
    %
    %
    %
    &\le h_n \left(\sum_{t=1}^{j-1} \rho^{2(t+h-j)} \right) g(\rho_n,n-h_n)^2 \sqrt{4 K_4} ~\\
    &\le h_n g(\rho_n,h_n)^2 g(\rho_n,n-h_n)^2 \sqrt{4 K_4} ~,
\end{align*}
where we use Loeve's inequality (see Theorem 9.28 in \cite{Davidson1994}) in the first inequality above. In the second inequality, we use $E[y_{n,t-1}^2] \le E[y_{n,t-1}^4]^{1/2}$, $y_{n,t-1} = \xi_{n,0}(\rho_n,t-1)$, Lemma \ref{MO-PM:lemmaA7}, and $g(\rho_n,t-1)^2 \le g(\rho_n,n-h_n)^2$ for all $t=1,\ldots,n-h_n$. Note that we also use $j \le h_n$, which also implies the last inequality above. The second case is $ h_n + 1 \le j \le n - h_n + 1  $. We follow the same approach as before and conclude 
$E[b_j^2] \le h_n g(\rho_n,h_n)^2 g(\rho_n,n-h_n)^2 \sqrt{4 K_4}.$
%
    %
    %
    %
    %
    %
 %
 In the final case, we have $j \ge n-h_n + 2$. Similarly, we obtain
 $E[b_j^2] \le h_n g(\rho_n,h_n)^2 g(\rho_n,n-h_n)^2 \sqrt{4 K_4}.$
%
    %
    %
    %
    %
    %
 %
\end{proof}

\begin{lemma}\label{MO-PM:lemmaE8}
Suppose Assumption \ref{appendix:AV_assumptions} holds. Then, for any $(\rho_n ,\sigma_n,h_n)$ such that $ |\rho_n| \le  1+M/n$, $\sigma_n^2 \in [\underline{\sigma},\:\overline{\sigma}]$, and $h_n \le n$, we have
$$  E\left[ \left(\frac{  \sum_{t=1}^{n-h_n} u_{n,t} y_{n,t-1}}{ (n-h_n)^{1/2}g(\rho_n,n-h_n)}\right)^2 \right] \le 2 K_4~,$$ 
where $g(\rho,k) = \left(\sum_{\ell=0}^{k-1} ~ \rho^{2 \ell}\right)^{1/2}$. 
\end{lemma}

\begin{proof}
It follows from the proof of Lemma E.8 in \cite{MO-PM2020}.
\end{proof}


\begin{lemma}\label{MO-PM:lemmaA1}
Suppose Assumptions \ref{appendix:AV_assumptions} hold. Then, for any sequences $\{ \rho_n: n \ge 1\} $ such that $|\rho_n| \le 1+M/n$, $\{\sigma_n^2: n \ge 1\} \subset [\underline{\sigma},\:\overline{\sigma}]$, and $\{h_n: n \ge 1\}$ satisfying $h_n = o(n)$ and $h_n \le n$, we have
\begin{equation}\label{eq:MO-PM:lemmaA1}
    \frac{\sum_{t=1}^{n-h_{n}} \xi_{n,t}(\rho_n,h_{n}) u_{n,t}}{(n-h_{n})^{1/2}g(\rho_{n},h_{n}) \sigma_n^2} \overset{d}{\to} N(0,1)~,
\end{equation} 
where $\xi_{n,t}(\rho_n,h_{n})$ is in \eqref{eq:appendix_A_xi} and $g(\rho,h)^2 = \sum_{\ell=1}^h \rho^{2h}$.
\end{lemma}

\begin{proof}
We adapt the proof of Lemma A1 in \cite{MO-PM2020}. We start by writing the left-hand side term in \eqref{eq:MO-PM:lemmaA1} as follows
$$ \sum_{t=1}^{n-h_{n}} \chi_{n,t}~, $$
where 
\begin{equation*}
    \chi_{n,t} = \frac{\xi_{n,n-h_n+1-t}(\rho_n,h_{n}) u_{n,n-h_n+1-t}}{(n-h_{n})^{1/2}g(\rho_{n},h_{n}) \sigma_n^{2}}~,
\end{equation*}
for $t = 1,\ldots,n-h_n$. Define the $\sigma$-algebra $\mathcal{F}_{n,t} = \sigma \left(u_{n-h_n+j-t}: j \ge 1 \right)$. Note that for any $t = 1,\ldots,n-h_n$, $\chi_{n,t}$ is measurable with respect to $\mathcal{F}_{n,t}$. Therefore, the sequence $\{\chi_{n,t} : 1 \le t \le n-h_n\}$ is adapted to the filtration $\{ \mathcal{F}_{n,t} : 1 \le t \le n-h_n \}$. Moreover, $\xi_{n,n-h_n+1-t}(\rho_n,h_{n})$ is measurable with respect to $\mathcal{F}_{n,t-1}$ since it is a function of $\{ u_{n,n-h_n+j-(t-1)} : 1 \le j \le h_n \}$. This implies that $E[\chi_{n,t}|\mathcal{F}_{n,t-1}] = 0$ since
$$ E[\chi_{n,t}|\mathcal{F}_{n,t-1}] = \frac{\xi_{n,n-h_n+1-t}(\rho_n,h_{n})}{(n-h_{n})^{1/2}g(\rho_{n},h_{n}) \sigma_n^2} E[ u_{n,n-h_n+1-t}|\mathcal{F}_{n,t-1}]$$
and by Assumptions \ref{appendix:AV_assumptions} we conclude $E[ u_{n,n-h_n+1-t}|\mathcal{F}_{n,t-1}] = E[u_{n,n-h_n+1-t}] =0$. 

The derivation presented above proves that the sequence $\{\chi_{n,t} : 1 \le t \le n-h_n\}$ is a martingale difference array with respect to the filtration $\{ \mathcal{F}_{n,t} : 1 \le t \le n-h_n \}$. The result in \eqref{eq:MO-PM:lemmaA1} then follows by a martingale central limit theorem (Theorem 24.3 in  \cite{Davidson1994}), which requires 
\begin{align*}
   (i) \quad \sum_{t=1}^{n-h_n} E[\chi_{n,t}^2] = 1~,\quad
   (ii) \quad \sum_{t=1}^{n-h_{n}} \chi_{n,t}^2 \overset{p}{\to} 1~, \quad
  (iii) \quad \max_{1 \le t \le n-h_n} |\chi_{n,t}| \overset{p}{\to} 0~.
\end{align*}
The condition $(i)$ follows by using that $E[\xi_{n,t}(\rho_n,h_n)^2u_{n,t}^2] = g(\rho_n,h_n)^2\sigma_n^4$. To prove the condition $(ii)$ is sufficient to show
\begin{equation}\label{eq:aux1_lemmaA1}
     \text{Var} \left( \sum_{t=1}^{n-h_{n}} \chi_{n,t}^2 \right) \to 0~.
\end{equation}
To prove \eqref{eq:aux1_lemmaA1}, we first recall that 
$$ \sum_{t=1}^{n-h_{n}} \chi_{n,t}^2 = \frac{1}{(n-h_{n})g(\rho_{n},h_{n})^2 \sigma_n^4} \times \sum_{t=1}^{n-h_n} \xi_{n,t}(\rho_n,h_n)^2 u_{n,t}^2~,$$
where the second term of the right-hand side of the previous display can be decomposed into the sum of its expected value and another three zero mean terms:
$$ (n-h_n)  g(\rho_n,h_n)^2\sigma_n^4  + \sum_{j=1}^n (u_{n,j}^2-\sigma_n^2) b_{n,j} + \sum_{j=1}^{n} u_{n,j}  d_{n,j} + \sum_{j=1}^{n} (u_{n,j}^2-\sigma_n^2) g(\rho_n,h_n)^2\sigma_n^2~,$$
where $b_{n,j} = \sum_{t=j-h_n}^{j-1} \rho_n^{2(h_n+t-j)} u_{n,t}^2 I\{ 1 \le t \le n-h_n \}$ and 
$$d_{n,j}  =  \sum_{t=j-h_n}^{j-1} \sum_{\ell_2=1}^{j-t-1}  u_{n,t+\ell_2}\rho_n^{2h-j+t-\ell_2}u_{n,t}^2 I\{ 1 \le t \le n-h_n\}~.$$
%
    %
%
Note that $b_{n,j}$ and $d_{n,j}$ are measurable with respect to the $\sigma-$algebra $\sigma( u_{n,k} : 1 \le k \le j-1\}$. By Assumptions \ref{appendix:AV_assumptions} and Loeve's inequality (Theorem 9.28 in \cite{Davidson1994}), we conclude
\begin{align*}
    E[ b_{n,j}^2]
    %
    &\le h_n \sum_{t=j-h_n}^{j-1} E[\rho_n^{4(h_n+t-j)} u_{n,t}^4 ]
     \le h_n g(\rho_n,h_n)^4 K_4
\end{align*}
and
\begin{align*}
    E[d_{n,j}^2]  
    %
    &\le h_n  E \left[ \left( \sum_{\ell_2=1}^{j-i-1}  u_{n,i+\ell_2}\rho_n^{2h-j+i-\ell_2}u_{n,t}^2 \right)^2  \right] 
     \le h_n   g(\rho_n,h_n)^4  \sigma_n^2   K_4 
\end{align*}
for all $j=1,\ldots,n$. 

We use the decomposition presented above, Assumptions \ref{appendix:AV_assumptions}, and Loeve's inequality (see Theorem 9.28 in \cite{Davidson1994}) imply that the left-hand side of \eqref{eq:aux1_lemmaA1} is lower than or equal to
$$  \frac{ 3 \sum_{j=1}^n E[(u_{n,j}^2-\sigma_n^2)^2  b_{n,j}^2] + 3 \sum_{j=1}^{n} E[ u_{n,j}^2 d_{n,j}^2] + 3 g(\rho_n,h_n)^4\sigma_n^4\sum_{j=1}^{n} E[(u_{n,j}^2-\sigma_n^2)^2] }{(n-h_{n})^2g(\rho_{n},h_{n})^4 \sigma_n^8}~. $$
By Assumptions \ref{appendix:AV_assumptions} and the upper bounds that we found for $E[ b_{n,j}^2]$ and $E[d_{n,j}^2]$, the previous expression is lower than or equal to
%
%
%
%
%
$$ \frac{3  n h_n   K_4^2}{(n-h_{n})^2 \underline{\sigma}^4} 
+ \frac{3  n h_n        K_4 }{(n-h_{n})^2  \underline{\sigma}^2} 
+ \frac{3 n K_4}{{(n-h_{n})^2 \underline{\sigma}^2}}~.$$
The previous expression is $o\left(1\right)$ since $h_n = o\left(n\right)$ as $n \to \infty$. This implies that \eqref{eq:aux1_lemmaA1} holds. 

Finally, to verify that condition (iii) holds is sufficient to show
\begin{equation}\label{eq:aux2_lemmaA1}
    (n-h_n) E\left[ \chi_{n,t}^2  I \{ |\chi_{n,t}^2 |> c\}  \right] \to 0~,
\end{equation} 
for any $c>0$, where $I\{ \cdot \}$ is the indicator function. If the condition in \eqref{eq:aux2_lemmaA2} holds, then condition (iii) follows by Theorem 23.16 in \cite{Davidson1994}. To verify \eqref{eq:aux2_lemmaA1}, note that
\begin{align*}
    (n-h_n)E\left[ \chi_{n,t}^2  I \{ |\chi_{n,t}^2 |> c\}  \right]  &\le ~(n-h_n) E\left[ \frac{\chi_{n,t}^4}{c^2}  I \{ |\chi_{n,t}^2 |> c\}  \right] \\
    &\le ~ (n-h_n) \frac{E\left[ \chi_{n,t}^4 \right]}{c^2}   \\
    %
    %
    &= \frac{  E\left[\xi_{n,n-h_n+1-t}(\rho_n,h_{n})^4 \right] E\left[u_{n,n-h_n+1-t}^4\right]}{(n-h_{n})\sigma_n^8 g(\rho_n,h_n)^4 c^2}\\
    &\le  \frac{ 4 K_4 E\left[u_{n,n-h_n+1-t}^4\right]}{(n-h_{n})\sigma_n^8 c^2}~,
\end{align*}
where the equality above uses Assumption \ref{appendix:AV_assumptions}, and the last inequality follows by Lemma \ref{MO-PM:lemmaA7}. By Assumption \ref{appendix:AV_assumptions} we obtain
$ (n-h_n)E\left[ \chi_{n,t}^2  I \{ |\chi_{n,t}^2 |> c\}  \right]  \le 4K_4^2/((n-h_n) \underline{\sigma}^4 c^2),$
which is sufficient to conclude \eqref{eq:aux2_lemmaA1}. 
\end{proof}


\begin{lemma}\label{MO-PM:lemmaA4_E8}
Suppose Assumption \ref{appendix:AV_assumptions} holds. Then, for any sequences $\{ \rho_n: n \ge 1\} $ such that $|\rho_n| \le 1+M/n$, $\{\sigma_n^2: n \ge 1\} \subset [\underline{\sigma},\:\overline{\sigma}]$, and $\{h_n: n \ge 1\}$ satisfying $h_n = o(n)$ and $h_n \le n$, we have
\begin{equation}\label{eq:MO-PM:lemmaA4_E8}
    \frac{\sum_{t=1}^{n-h} \xi_{n,t}(\rho_n,h_n)  (\hat{u}_{n,t}(h_n) -u_{n,t}) }{(n-h_n)^{1/2} g(\rho_n,h_n) \sigma_n^2}  \overset{p}{\to} 0~,
\end{equation}
where $\hat{u}_{n,t}(h_n) = y_{n,t}-\hat{\rho}_n(h_n) y_{n,t-1}$, $\hat{\rho}_n(h_n)$ is defined in \eqref{eq:appendix_A_OLS}, and $g(\rho,h)^2 = \sum_{\ell=1}^h \rho^{2h}$.
\end{lemma}

\begin{proof}
A proof can be adapted from the proof of Lemma A.4 and Lemma E.8 in \cite{MO-PM2020}. Importantly, their Assumption 3 (relevant for the proof) holds due to Proposition B.1 in Appendix B. 
\end{proof}

\begin{lemma}\label{MO-PM:lemmaA3}
Suppose Assumption \ref{appendix:AV_assumptions} holds. Then, for any sequences $\{ \rho_n: n \ge 1\} $ such that $|\rho_n| \le 1+M/n$, $\{\sigma_n^2: n \ge 1\} \subset [\underline{\sigma},\:\overline{\sigma}]$, and $\{h_n: n \ge 1\}$ satisfying $h_n = o(n)$ and $h_n \le n$, we have
\begin{align*}
    (i)& \quad \frac{\hat{\beta}_n(h_n)-\beta(\rho_n,h_n)}{g(\rho_n,h_n)} \overset{p}{\to} 0~,\\
    (ii)& \quad \frac{g(\rho_n,n-h_n) \left( \hat{\eta}(\rho_n,h_n) - \eta(\rho_n,h_n) \right)}{g(\rho_n,h_n)} \overset{p}{\to} 0~,\\
    (iii)& \quad  (n - h_n)^{1/2} \times g(\rho_n,n-h_n) \times \left( \hat{\rho}_n(h_n) - \rho_n \right)  = O_p(1)~,
\end{align*}
where $\hat{\eta}(\rho_n,h_n) = \rho_n \hat{\beta}_n(h_n) + \hat{\gamma}_n(h_n)$, $\eta(\rho_n,h_n) = \rho_n \beta(\rho_n,h_n)=\rho_n^{h_n+1}$, $\hat{\beta}_n(h_n)$ and $\hat{\gamma}_n(h_n)$ are defined in \eqref{eq:appendix_A_LP}, $\hat{\rho}_n(h_n)$ is in \eqref{eq:appendix_A_OLS}, and $g(\rho,h)^2 = \sum_{\ell=1}^h \rho^{2h}$.
\end{lemma}

\begin{proof}
A proof can be adapted from the proof of Lemma A.3 in \cite{MO-PM2020}. Importantly, their Assumption 3 (relevant for the proof) holds due to Proposition \ref{proposition:A3} in Appendix \ref{sec:appendix_b}. 
\end{proof}
 
\begin{lemma}\label{MO-PM:lemmaA2}
Suppose Assumptions \ref{appendix:AV_assumptions} holds. Then, for any sequences $\{ \rho_n: n \ge 1\} $ such that $|\rho_n| \le 1+M/n$, $\{\sigma_n^2: n \ge 1\} \subset [\underline{\sigma},\:\overline{\sigma}]$, and $\{h_n: n \ge 1\}$ satisfying $h_n = o(n)$ and $h_n \le n$, we have
$$ \frac{\sum_{t=1}^{n-h_n} \hat{\xi}_{n,t}(h_n)^2 \hat{u}_{n,t}(h_n)^2}{(n-h_n) g(\rho_n,h_n)^2 \sigma_n^4} \overset{p}{\to} 1~,$$
where $\hat{\xi}_{n,t}(h_n) = y_{n,t+h_n}-\hat{\beta}_n(h_n) y_{n,t}-\hat{\gamma}_n(h_n)y_{n,t-1}$,  $\hat{u}_{n,t}(h_n) = y_{n,t}-\hat{\rho}_n(h_n) y_{n,t-1}$, $\hat{\beta}_n(h_n)$ and $\hat{\gamma}_n(h_n)$ are defined in \eqref{eq:appendix_A_LP}, $\hat{\rho}_n(h_n)$ is defined in \eqref{eq:appendix_A_OLS}, and $g(\rho,h)^2 = \sum_{\ell=1}^h \rho^{2h}.$
\end{lemma}

\begin{proof}
We adapt the proof of Lemma A.2 in \cite{MO-PM2020} presented in their Supplemental Appendix E.2. They claim that is sufficient to prove
 \begin{equation}\label{eq:aux1_lemmaA2}
     \frac{\sum_{t=1}^{n-h_n} \hat{\xi}_{n,t}(h_n)^2 \hat{u}_{n,t}(h_n)^2}{(n-h_n) g(\rho_n,h_n)^2 \sigma_n^4}- \frac{\sum_{t=1}^{n-h_n}  {\xi}_{n,t}(\rho_n,h_n)^2 {u}_{n,t}^2}{(n-h_n) g(\rho_n,h_n)^2 \sigma_n^4} \overset{p}{\to} 0~,
 \end{equation}
 since they then can conclude using their Lemma A6, which implies
 $$ \frac{\sum_{t=1}^{n-h_n}  {\xi}_{n,t}(\rho_n,h_n)^2 {u}_{n,t}^2}{(n-h_n) g(\rho_n,h_n)^2 \sigma_n^4} \overset{p}{\to} 1~.
 $$ 
We avoid using their Lemma A6 since its proof requires that the shocks have a finite 8th moment. Instead, we observe that \eqref{eq:aux1_lemmaA1} presented in the proof of Lemma \ref{MO-PM:lemmaA1} implies the previous claim.

To verify \eqref{eq:aux1_lemmaA2}, \citeauthor{MO-PM2020} prove that is sufficient to show that
\begin{equation}\label{eq:aux2_lemmaA2}
    \frac{\sum_{t=1}^{n-h_n} \left [ \hat{\xi}_{n,t}(h_n)^2 \hat{u}_{n,t}(h_n)^2 -  {\xi}_{n,t}(\rho_n,h_n)^2 {u}_{n,t}^2 \right ]^2 }{(n-h_n) g(\rho_n,h_n)^2 \sigma_n^4}
\end{equation}
converges in probability to zero. To prove that, they derive the following upper bound for \eqref{eq:aux2_lemmaA2}:
$$ 3[(\hat{R}_1)^{1/2} \times (\hat{R}_2)^{1/2} ] + 3[(\hat{R}_3)^{1/2} \times (\hat{R}_4)^{1/2} ] + 3[(\hat{R}_1)^{1/2} \times (\hat{R}_3)^{1/2} ]~,$$\vspace{-0.5cm}
where
\begin{align*}
    \hat{R}_1 &= \frac{\sum_{t=1}^{n-h_n} \left[\xi_{n,t}(\rho_n,h_n) - \hat{\xi}_{n,t}(h_n) \right ]^4}{(n-h_n)g(\rho_n,h_n)^4\sigma_n^8}~,   
    &\hat{R}_2 =& \frac{\sum_{t=1}^{n-h_n} u_{n,t}^4}{n-h_n} \\
    \hat{R}_3 &= \frac{\sum_{t=1}^{n-h_n} \left[ \hat{u}_{n,t}(h_n) - u_{n,t} \right]^4}{n-h_n}~,  
    &\hat{R}_4 =& \frac{\sum_{t=1}^{n-h_n}  \xi_{n,t}(\rho_n,h_n)^4}{(n-h_n)g(\rho_n,h_n)^4\sigma_n^8}~.
\end{align*}

In what follows, we use Assumptions \ref{appendix:AV_assumptions} to prove that (i) $\hat{R}_1$ and $\hat{R}_3$ are $o_p\left(1\right)$ and (ii) $\hat{R}_2$ and $\hat{R}_4$ are $O_p\left(1\right)$, which are sufficient to conclude that \eqref{eq:aux2_lemmaA2} converges to zero in probability. 

To verify $\hat{R}_1$ is $o_p\left(1\right)$, let us first observe that 
$$ \xi_{n,t}(\rho_n,h_n) -  \hat{\xi}_{n,t}(h_n) =  [\hat{\beta}_n(h_n) - \beta(\rho_n,h_n) ] u_{n,t} 
 + [ \hat{\eta}_n(\rho_n,h_n) - \eta(\rho_n,h_n)] y_{n,t-1} ~,$$
where  $\hat{\eta}(\rho_n,h_n) = \rho_n \hat{\beta}_n(h_n) + \hat{\gamma}_n(h_n)$ and $\eta(\rho_n,h_n) = \rho_n \beta(\rho_n,h_n)$. Then, using Loeve's inequality (see Theorem 9.28 in \cite{Davidson1994}), we obtain
\begin{align*}
    \hat{R}_1 &\le 8 \left(\frac{[\hat{\beta}_n(h_n) - \beta(\rho_n,h_n) ]}{g(\rho_n,h_n)} \right)^4\left(\frac{ \sum_{t=1}^{n-h_n} u_{n,t}^4}{(n-h_n)\sigma_n^8}\right) \\
    &~+8 \left(\frac{g(\rho_n,n-h_n)[ \hat{\eta}_n(\rho_n,h_n) - \eta(\rho_n,h_n)]}{g(\rho_n,h_n)}  \right)^4 \left(\frac{ \sum_{t=1}^{n-h_n} y_{n,t-1}^4}{(n-h_n)g(\rho_n,n-h_n)^4\sigma_n^8} \right)~.
\end{align*}
Note that the first term on the right-hand side in the previous expression goes to zero in probability due to part (i) in Lemma \ref{MO-PM:lemmaA3}, Markov's inequality, and Assumptions \ref{appendix:AV_assumptions}. The second term on the right-hand side in the previous expression goes to zero in probability due to part (ii) in Lemma \ref{MO-PM:lemmaA3}, Markov's inequality, and using that 
\begin{equation}\label{eq:aux3_lemmaA2}
    E[y_{n,t-1}^4] = E[\xi_{n,0}(\rho_n,t-1)^4] \le g(\rho_n,n-h_n)^4  4K_4~, 
\end{equation}
where the inequality holds due to Lemma \ref{MO-PM:lemmaA7} and $g(\rho_n,t-1)^4 \le g(\rho_n,n-h_n)^4$ for all $i \le n-h_n$. This completes the proof of $\hat{R}_1$ is $o_p\left(1\right)$.

To prove that  $\hat{R}_3$ is $o_p\left(1\right)$, note that we can write:
$$ \hat{R}_1 = (g(\rho_n,h_n)[\hat{\rho}_n(h_n) - \rho])^4 \frac{\sum_{t=1}^{n-h_n} y_{t-1}^4 }{(n-h_n)g(\rho_n,h_n)^4} $$
since $ \hat{u}_{n,t}(h_n) - u_{n,t} = (\hat{\rho}_n(h_n) - \rho) y_{t-1}~.$ Note that the right-hand side in the previous expression goes to zero in probability due to part (iii) in Lemma \ref{MO-PM:lemmaA3}, Markov's inequality, and using \eqref{eq:aux3_lemmaA2}. This completes the proof of $\hat{R}_3$ is $o_p\left(1\right)$.
Finally, Markov's inequality and Assumptions \ref{appendix:AV_assumptions} implies that $\hat{R}_2$ is $O_p\left(1\right)$. While Markov's inequality and Lemma  \ref{MO-PM:lemmaA7} implies $\hat{R}_4$ is $O_p\left(1\right)$. 
\end{proof}

\subsection{Proof of Proposition \ref{prop:CI_LA_AR}}\label{sec:proof_prop_b2}
\begin{proof}
    For any $\epsilon>0$, define the event $E_{n,\epsilon} =  \{ |R_n(1)| \le c_n^*(1, 1-\alpha) - \psi(\epsilon) \} $, where $ \psi(\epsilon) = z_{1-\alpha/2-3\epsilon/2} - z_{1-\alpha/2-2\epsilon} $ and $z_\alpha$ is the $\alpha$-quantile of the standard normal distribution. 

    We will prove that for any (small) $\epsilon>0$ the following claims hold,
    \begin{align}
        \lim_{n \to \infty}  & P_n \left( [1/L,L] \subseteq  C_{la-ar}^*(h_n,1-\alpha)  \mid E_{n,\epsilon} \right)   = 1~,\label{eq:aux1_prop_b2} \\
        \liminf_{n \to \infty} & P_n (E_{n,\epsilon})  \ge 1-\alpha-4\epsilon~.\label{eq:aux2_prop_b2}
    \end{align}
    These two claims are sufficient to conclude that 
    $$\liminf_{n \to \infty}  P_n \left( [1/L,L] \subseteq  C_{la-ar}^*(h_n,1-\alpha) \right) \ge 1-\alpha -4\epsilon~.$$
    Since this holds for any (small) $\epsilon>)$, it implies the claim of the proposition.

    \noindent \textit{Claim 1:} \eqref{eq:aux1_prop_b2} holds. To verify this, we first rewrite the lower and upper bounds of $C_{la-ar}^*(h,1-\alpha)$ using the definition of $R_n(1)$ as in \eqref{eq:root}:
    \begin{align}
        \left( \hat{\beta}_n(1) - \hat{s}_n(1) c_n^*(1, 1-\alpha) \right)^{h_n} &= \left(1 - c_1/n - \zeta_{n,1} \right)^{h_n} \label{eq:aux3_prop_b2}\\
        \left( \hat{\beta}_n(1) + \hat{s}_n(1) c_n^*(1, 1-\alpha) \right)^{h_n} &= \left(1 - c_1/n + \zeta_{n,2} \right)^{h_n} \label{eq:aux4_prop_b2}
    \end{align}
    where $\zeta_{n,1} = \hat{s}_n(1) \{ -R_n(1)  + c_n^*(1, 1-\alpha) \}$ and $\zeta_{n,2} = \hat{s}_n(1) \{ R_n(1)  + c_n^*(1, 1-\alpha) \}$ . 
    
    Note that we have $\zeta_{n,1} \ge \hat{s}_n(1) \psi(\epsilon)$ and $\zeta_{n,2} \ge \hat{s}_n(1) \psi(\epsilon)$ conditional $E_{n,\epsilon}$. Additionally, we can obtain that
    \begin{equation}\label{eq:aux5_prop_b2}
         \hat{s}_n(1) = c_2^{1/2} n^{-1/4} (1 + o_p(1))~,
    \end{equation}
    which follows by the formula in \eqref{eq:HC_se}, Lemma \ref{MO-PM:lemmaA2} and part 2 of Lemma  \ref{AppendixA:lemma_event_En}(adapted for the i.i.d. case that we consider in this proposition), and because $h_n/n^{1/2} \to c_2$ as $n\to \infty$.

    Since $h_n/\sqrt{n} \to c_2>0$ as $n \to \infty$, it holds that $\lim_{n \to \infty} (1-C n^{-1/4} )^{h_n} \to 0$ for any positive constant $C$. This implies that conditional on $E_{n,\epsilon}$, we have that the lower bound of $C_{la-ar}^*(h,1-\alpha)$ goes to zero. To see this, consider the following derivations using \eqref{eq:aux3_prop_b2}, 
    \begin{align*}
        \left(1 - c_1/n - \zeta_{n,1} \right)^{h_n} &\overset{(1)}{\le} \left(1   -  c_2^{1/2} n^{-1/4} (1 + o_p(1))  \psi(\epsilon) \right)^{h_n}  
    \end{align*}
    where (1) holds by definition of $\zeta_{n,1}$ conditional on $E_n$, \eqref{eq:aux5_prop_b2}, and because $\psi(\epsilon)$ is positive by definition. Since $1+o_p(1)$ is larger than $1-\delta$ for any small $\delta>0$ with a high probability for any $n$ sufficiently larger, we can conclude that the right-hand side of the previous display goes to zero with a high probability conditional on $E_n$. 
    
    The previous derivation concludes that the lower bound of $C_{la-ar}^*(h,1-\alpha)$ goes to zero conditional on $E_n$, which implies that the lower bound is asymptotically lower than $1/L$ conditional on $E_n$. Now we will show that the upper bound of $C_{la-ar}^*(h,1-\alpha)$  is asymptotically larger than $L$ conditional on $E_n$. To see this, consider the following derivation using \eqref{eq:aux4_prop_b2},
    \begin{align*}
        \left(1 - c_1/n + \zeta_{n,2} \right)^{h_n} &\overset{(1)}{\ge} \left(1 - c_1/n + c_2^{1/2} n^{-1/4} (1+o_p(1)) \psi(\epsilon) \right)^{h_n}  \\
        &\overset{(2)}{\ge} 1 + ( - c_1/n + c_2^{1/2} n^{-1/4} (1+o_p(1)) \psi(\epsilon) ) h_n \\
        &= 1 + O(n^{-1/2}) + c_2^{1/2} n^{1/4} (1+o_p(1)) (h_n/\sqrt{n})
    \end{align*}
    where (1) holds by definition of $\zeta_{n,2}$ conditional on $E_n$, \eqref{eq:aux5_prop_b2}, and because $\psi(\epsilon)$ is positive by definition, and (2) holds by Bernoulli's inequality. Since $1+o_p(1)$ is larger than $1-\delta$ for any small $\delta>0$ with a high probability for any $n$ sufficiently larger, we can conclude that the right-hand side of the previous display goes to infinity with a high probability conditional on $E_n$. In particular, the upper bound of $C_{la-ar}^*(h,1-\alpha)$  is asymptotically larger than $L$ conditional on $E_n$. This completes the proof of claim 1.

    \noindent \textit{Claim 2:} \eqref{eq:aux2_prop_b2} holds. We first note that the following inclusion
    \begin{align*}
       \{ |R_n(1)| \le z_{1-\alpha/2-2\epsilon} \} \subseteq \{ |R_n(1)| \le c_n^*(1,1-\alpha) - \psi(\epsilon) \} = E_{n,\epsilon} 
    \end{align*}
    holds with a high probability due to part 1 of Lemma \ref{lemma:quantiles_high-prob} (adapted for the i.i.d. case that we consider in this proposition). We then observe that the probability of the left-hand side of the previous expression goes to $1-\alpha-4\epsilon$ due to Theorem \ref{thm:sequence_models}. This completes the proof of claim 2. 
\end{proof}
  


\counterwithin{figure}{section}
\counterwithin{table}{section}
\renewcommand\thesection{D}
\renewcommand\thesubsection{D.\arabic{subsection}}
\renewcommand{\thetable}{D.\arabic{table}} 
\renewcommand{\thelemma}{D.\arabic{lemma}}
\renewcommand{\theequation}{D.\arabic{equation}}
\setcounter{equation}{0}
  
\section{Proof of Auxiliary Results: Asymptotic Refinements}\label{sec:appendix_online2}

\subsection{Proof of the Lemma B.4}\label{sec:appendix_proof_lemma_b4}\label{sec:proof_lemma_b4}
\begin{proof}
    By Lemma \ref{lemma:approximation}, there exists a random variable $\Tilde{R}_n(h)$ such that
    $$ P_{\rho} \left(   \left |R_n(h) - \Tilde{R}_n(h) \right | > n^{-1-\epsilon} \right) \le C n^{-1-\epsilon} \left(E[|u_t|^{k}] + E[u_t^{2k}] + E[u_t^{4k}] \right)~,$$
    where 
   $$\Tilde{R}_n(h) \equiv  (n-h)^{1/2} \mathcal{T} \left( \frac{1}{n-h} \sum_{t=1}^{n-h} X_t ;\sigma^2,\psi_4^4,\rho \right) $$ 
and the sequence $\{ X_t : 1 \le t \le n-h \}$ is defined in \eqref{eq:Xi-vec}. Due to Lemma \ref{lemma:approximation}, we know $\mathcal{T}$ is a polynomial. Define $\tilde{J}_n(x,h,P,\rho)  = P_{\rho} \left(   \Tilde{R}_n(h)  \le x   \right)$.  Using Bonferroni's inequality, we conclude
    \begin{align*}
         |P_{\rho} \left(   R_n(h)  \le x   \right) - P_{\rho} \left(   \Tilde{R}_n(h)  \le x   \right)| \le D_n + P_{\rho} \left(   \left |R_n(h) - \Tilde{R}_n(h) \right | > n^{-1-\epsilon} \right)~.
    \end{align*}
Therefore,    
    $ \sup_{x \in \mathbf{R}} |J_n(x,h,P,\rho) - \tilde{J}_n(x,h,P,\rho)| \le D_n + C n^{-1-\epsilon} \left(E[|u_t|^{k}] + E[u_t^{2k}] + E[u_t^{4k}] \right),$
    which completes the proof of the Lemma. Note that the constant $C$ is defined in  Lemma \ref{lemma:approximation} and only depends on $a$, $h$, $k$, $C_\sigma$, and $c_u$.
\end{proof}

\begin{lemma}\label{lemma:approximation}
Suppose Assumption \ref{AV_assumptions} holds. For any fixed $h \in \mathbf{N}$ and $a \in (0,1)$. Then, for any $\rho \in [-1+a,1-a]$ and $\epsilon \in (0,1/2)$, there exist a constant $C = C(a,h,k,C_\sigma,c_u)>0$, where $k \ge 8(1+\epsilon)/(1-2\epsilon)$, and a real-valued function 
$ \mathcal{T}(\cdot~;\sigma^2,\psi_4^4,\rho) : \mathbf{R}^8 \to \mathbf{R},$ such that 
\begin{enumerate}
    \item $\mathcal{T}(\mathbf{0};\sigma^2,\psi_4^4,\rho) = 0$, 

    \item $\mathcal{T}(x;\sigma^2,\psi_4^4,\rho) $ is a polynomial of degree 3 in $x \in \mathbf{R}^8$   with coefficients depending continuously differentiable on $\sigma^2$, $\psi_4^4$, and  $\rho$, 

    \item $ P_{\rho} \left(   \left |R_n(h) - \Tilde{R}_n(h) \right | > n^{-1-\epsilon} \right) \le C n^{-1-\epsilon} \left(E[|u_t|^{k}] + E[u_t^{2k}] + E[u_t^{4k}] \right),$
\end{enumerate}
where  $\sigma^2 = E_P[u_1^2]$, $\psi_4^4 = E_P[u_1^4]$, 
$$\Tilde{R}_n(h) \equiv  (n-h)^{1/2} \mathcal{T} \left( \frac{1}{n-h} \sum_{t=1}^{n-h} X_t ;\sigma^2,\psi_4^4,\rho \right) $$ 
and the sequence $\{ X_t : 1 \le t \le n-h \}$ is defined in \eqref{eq:Xi-vec}. Furthermore, the asymptotic variance of $\Tilde{R}_n(h)$ equals one.
\end{lemma}

\begin{proof}
The proof has two main parts. We first use Lemmas \ref{lemma:numerator} and \ref{lemma:denominator} to approximate $R_n(h)$ using functions based on $\xi_t$, $u_t$, and $y_{t-1}$. We then replace $y_{t-1}$ by $z_{t-1}$. We specifically define the polynomial $\mathcal{T}$.  

\noindent \textbf{Part 1:} The derivations presented on page 1811 in \cite{MO-PM2020} implies  
$$ R_n(h) = \frac{ \sum_{t=1}^{n-h} \xi_{t}(h)\hat{u}_{t}(h)}{\left(\sum_{t=1}^{n-h}\hat{\xi}_{t}(h)^2 \hat{u}_{t}(h)^2\right)^{1/2}}~, $$
where $\xi_t(\rho,h) =  \sum_{\ell=1}^h \rho^{h-\ell} u_{t+\ell}$, $\hat{u}_t(h) = y_t - \hat{\rho}_n(h) y_{t-1}$, $\hat{\xi}_t(h) = y_{t+h} -\left( \hat{\beta}_n(h) y_t + \hat{\gamma}_n(h) y_{t-1} \right) $,  and the coefficients $(\hat{\beta}_n(h),\hat{\gamma}_n(h))$ is as in \eqref{eq:la_lp} and $\hat{\rho}_n(h)$ is defined in \eqref{eq:rho_hat_h}. Define
$$ f_n = \frac{ \sum_{t=1}^{n-h} \xi_{t}(\rho,h)\hat{u}_{t}(h)}{n-h} \quad \text{and} \quad g_n =   \frac{\sum_{t=1}^{n-h} \hat{\xi}_t(h)^2 \hat{u}_t(h)^2}{V(n-h)} - 1~,$$
where $V = E[\xi_t(\rho,h)^2 u_t^2] = \sigma^4 \sum_{\ell =1 }^h \rho^{2(h-\ell)}$. It follows that
$$ R_n(h) = (n-h)^{1/2}  V^{-1/2} f_n \left(1 + g_{n}  \right)^{-1/2}~.$$
Lemmas \ref{lemma:numerator}, \ref{lemma:denominator}, and \ref{lemma:algebra} imply
$$P( (n-h)^{1/2}  |V^{-1/2}f_n| >  \delta) \le C \delta^{-k} \left(E[|u_t|^{k}] + E[u_t^{2k}] + E[u_t^{4k}] \right)$$
and
$$P( (n-h)^{1/2}|g_{n}| >  \delta ) \le C \delta^{-k}\left(E[|u_t|^{k}] + E[u_t^{2k}] + E[u_t^{4k}] \right)~.$$

\noindent \textit{Step 1:} Define $ \tilde{R}_{g,n} = (n-h)^{1/2}  f_n V^{-1/2} \left(1 -\frac{1}{2} g_{n} + \frac{3}{8} g_{n}^2  \right)$. Due to Lemma \ref{lemma:algebra}, we have 
\begin{equation*} 
    P \left( (n-h)^{3/2} \left|R_n(h) -\tilde{R}_{g,n} \right| >  \delta^4 \right) \le C \delta^{-k} \left(E[|u_t|^{k}] + E[u_t^{2k}] + E[u_t^{4k}] \right) ~
\end{equation*}
since 
$P\left(n^{3/2}\left|\left(1 + g_{n}  \right)^{-1/2} -  \left(1 -\frac{1}{2} g_{n} + \frac{3}{8} g_{n}^2  \right) \right| >  \delta^3 \right) \le C \delta^{-k}\left(E[|u_t|^{k}] + E[u_t^{2k}] + E[u_t^{4k}] \right)$.

\noindent \textit{Step 2:} Define 
$$\Tilde{R}_{f,n} = (n-h)^{1/2} V^{-1/2} \left( f_{1,n} + f_{2,n} + f_{3,n}\right)  \left(1 -\frac{1}{2} g_{n} + \frac{3}{8} g_{n}^2  \right)~,$$
where $f_n = \sum_{j=1}^4 f_{j,n}$ as in Lemma \ref{lemma:numerator}. We conclude
\begin{align*}
    P \left( (n-h)^{3/2} \left|\tilde{R}_{g,n} -\tilde{R}_{f,n} \right| >  \delta^4 \right) &= P \left((n-h)^{4/2} \left|  V^{-1/2}f_{4,n}  \left(1 -\frac{1}{2} g_{n} + \frac{3}{8} g_{n}^2  \right) \right| >  \delta^4  \right)  \\
    &\le C \delta^{-k} \left(E[|u_t|^{k}] + E[u_t^{2k}] + E[u_t^{4k}] \right)~,
\end{align*} 
where the last inequality follows by Lemmas \ref{lemma:algebra} and \ref{lemma:numerator}.

\noindent \textit{Step 3:} Define 
$$\Tilde{R}_{fg,n} = (n-h)^{1/2} V^{-1/2} \left( f_{1,n} + f_{2,n} + f_{3,n} -\frac{1}{2} f_{1,n}g_n -\frac{1}{2} f_{2,n}g_n +\frac{3}{8} f_{1,n}g_n^2  \right)~.$$
Lemmas \ref{lemma:numerator}, \ref{lemma:denominator}, and \ref{lemma:algebra} imply
$$   P \left( (n-h)^{3/2} \left|\tilde{R}_{f,n} -\Tilde{R}_{fg,n} \right| >  \delta^4 \right) \le C \delta^{-k} \left(E[|u_t|^{k}] + E[u_t^{2k}] + E[u_t^{4k}] \right)~.$$

\noindent \textit{Step 4:} Define 
$$
    \Tilde{R}_{y,n}(h) = (n-h)^{1/2} V^{-1/2} \left( \sum_{j=1}^3 f_{j,n} -\frac{1}{2} f_{1,n}g_{1,n} -\frac{1}{2} f_{1,n}g_{2,n} -\frac{1}{2} f_{2,n}g_{1,n} +\frac{3}{8} f_{1,n}g_{1,n}^2\right)~,
$$
where  $ g_{n} = \sum_{j=1}^3 g_{j,n} $ as in Lemma \ref{lemma:denominator}. We use Lemmas \ref{lemma:numerator}, \ref{lemma:denominator}, and \ref{lemma:algebra} to conclude
$$   P \left( (n-h)^{3/2} \left| \Tilde{R}_{fg,n} - \Tilde{R}_{y,n}(h)  \right| >  \delta^4 \right) \le C \delta^{-k} \left(E[|u_t|^{k}] + E[u_t^{2k}] + E[u_t^{4k}] \right)~.$$

\noindent \textit{Step 5:} By Bonferroni's:
$$ P \left( (n-h)^{3/2} \left| R_n(h) - \Tilde{R}_{y,n}(h)  \right| >  \delta^4 \right) \le C \delta^{-k} \left(E[|u_t|^{k}] + E[u_t^{2k}] + E[u_t^{4k}] \right)~.$$

\noindent \textbf{Part 2:} We consider $ V^{1/2} \mathcal{T} (x;\sigma^2,\psi_4^4,\rho)$ is equal to the following polynomial
\begin{equation}
     f_1(x) + f_2(x) + f_3(x) - \frac{1}{2} \left( f_1(x)g_1(x) + f_1(x)g_2(x) + f_2(x)g_1(x) \right) + \frac{3}{8} f_1(x)g_1(x)^2~,
\end{equation}
where $ f_1(x) = x_1$, $f_2(x) = -\sigma^{-2}(1-\rho^2)x_5 x_6$, $f_3(x) = \sigma^{-4}(1-\rho^2)x_5x_6 (2\rho x_5 + x_2)$, $g_1(x) = V^{-1} x_3 $, and
\begin{align*} 
    g_2(x) = V^{-1} \left( \frac{\psi_4^4}{\sigma^4}x_1^2 - \frac{2}{\sigma^2} x_1x_4 + (1-\rho^2) \left( g(\rho,h)^2 x_5^2 - \frac{2}{\sigma^2} x_5 x_7 + x_6^2 - \frac{2}{\sigma^2} x_6 x_8 \right) \right)~.
\end{align*}
Note that $\Tilde{R}_{y,n}(h) =  V^{1/2}\mathcal{T} \left( (n-h)^{-1} \sum_{t=1}^{n-h} \tilde{X}_t ;\sigma^2,\psi_4^4,\rho \right)$, where
$$     \tilde{X}_t = (\xi_t u_t, u_t^2-\sigma^2,(\xi_t u_t)^2-V,\xi_t u_t^3, u_t y_{t-1}, \xi_t y_{t-1}, \xi_t^2 u_t y_{t-1}, \xi_t u_t^2 y_{t-1})~.$$
Since $z_t = y_t + \rho^t z_0$, it follows that
$$ P \left( (n-h)^{1/2} \left| \frac{\sum_{t=1}^{n-h}  f_t y_{t-1} }{n-h} - \frac{\sum_{t=1}^{n-h}  f_t z_{t-1} }{n-h}  \right| > \delta \right) \le C \delta^{-k} \left(E[u_{t}^{2k}]  + E[u_{t}^{4k}] \right)$$
for $f_t = u_t, \xi_t, \xi_t^2u_t, \xi_t u_t^2$. Then, Lemma \ref{lemma:algebra} and step 5 in part 1 implies
$$   P \left( (n-h)^{3/2} \left| R_n(h) - \Tilde{R}_{n}(h)  \right| >  \delta^4 \right) \le C \delta^{-k} \left(E[|u_t|^{k}] + E[u_t^{2k}] + E[u_t^{4k}] \right)~,$$
where $\Tilde{R}_{n}(h) =  V^{1/2} \mathcal{T} \left( (n-h)^{-1} \sum_{t=1}^{n-h} {X}_t ;\sigma^2,\psi_4^4,\rho \right)$ and the sequence $\{X_t : 1 \le t \le n-h \}$ is defined in \eqref{eq:appendix_A_xi}. As we mentioned before, the constant $C$ includes the constants $C$'s that appear in Lemmas \ref{lemma:numerator}, \ref{lemma:denominator}, \ref{lemma:moments_y} and \ref{lemma:sample_averages} that only depends on $a$, $h$, $k$, $C_{\sigma}$, and $c_u$. Finally, we take $\delta = n^{(1/2-\epsilon)/4}$
\end{proof}

\begin{lemma}\label{lemma:numerator}
Suppose Assumption \ref{AV_assumptions} holds. For any fixed $h, k \in \mathbf{N}$ and $a \in (0,1)$. Define
$$ f_n = \frac{ \sum_{t=1}^{n-h} \xi_{t}(\rho,h)\hat{u}_{t}(h)}{n-h}$$
where $\xi_{t}(\rho,h) = \sum_{\ell=1}^h \rho^{h-\ell} u_{t+\ell}$, $\hat{u}_{t}(h) = y_t - \hat{\rho}_n(h) y_{t-1}$, and $\hat{\rho}_n(h)$ is as in \eqref{eq:rho_hat_h}. Then, for any $\rho \in [-1+a,1-a]$, there exists a constant $C = C(h,k,a,C_{\sigma})$ such that we can write
$$f_n = f_{1,n} + f_{2,n} + f_{3,n} + f_{4,n}~,$$
where
$$ P\left( (n-h)^{j/2} |f_{j,n}| \ge \delta^j \right) \le C \delta^{-k} \left( E[|u_t|^{k}]  + E[u_t^{2k}] + E[u_t^{4k}] \right) $$
for any $\delta < (n-h)^{1/2}$  and $j \in \{1,2,3,4\}$. 
\end{lemma}

\begin{proof}
In what follows we use $\xi_t = \xi_{t}(\rho,h)$. Using the definition of $\hat{\rho}_n(h)$, we obtain
$$  f_{n}   = \frac{\sum_{t=1}^{n-h} \xi_t u_t}{n-h} -  \left(   \frac{ \sum_{t=1}^{n-h} u_ty_{t-1}}{n-h}\right) \left(  \frac{\sum_{t=1}^{n-h} \xi_t y_{t-1}}{n-h} \right) \left( \frac{\sum_{t=1}^{n-h} y_{t-1}^2}{n-h}  \right)^{-1}~.$$
Let us define the components of $f_{n}$ as follows:
\begin{align*}
    f_{1,n} &= \left( \frac{\sum_{t=1}^{n-h} \xi_t u_t}{n-h} \right) \notag\\
    f_{2,n} &= - \frac{(1-\rho^2)}{\sigma^2} \left( \frac{\sum_{t=1}^{n-h} u_t y_{t-1}}{n-h} \right)\left( \frac{\sum_{t=1}^{n-h} \xi_t y_{t-1}}{n-h} \right) \notag\\
    f_{3,n} &= \frac{(1-\rho^2)}{\sigma^4} \left( \frac{\sum_{t=1}^{n-h} u_t y_{t-1}}{n-h} \right)\left( \frac{\sum_{t=1}^{n-h} \xi_t y_{t-1}}{n-h} \right) \left(  2\rho  \frac{ \sum_{t=1}^{n-h} u_t y_{t-1}}{ n-h } + \frac{\sum_{t=1}^{n-h} u_t^2-\sigma^2}{n-h  }\right) \notag\\
    f_{4,n} &= f_{n} - \left(f_{1,n} +f_{2,n} +f_{3,n} \right)~.  
\end{align*}
Note that by construction $f_n = \sum_{j=1}^4 f_{j,n}$. Lemma \ref{lemma:sample_averages} guarantees that each sample average in $f_{1,n}$, $f_{2,n}$, and $f_{3,n}$ verify the conditions to use Lemma \ref{lemma:algebra}, which imply that $P\left( (n-h)^{j/2} |f_{j,n}| \ge \delta^j \right) \le C \delta^{-k}  E[u_t^{2k}] $ for $j=1,2,3$, where the constant $C$ includes $C_{\sigma}$ and $a$.

To prove $P\left( (n-h)^{4/2} |f_{j,n}| \ge \delta^4 \right) \le C \delta^{-k} \left(E[|u_t|^{k}] + E[u_t^{2k}] + E[u_t^{4k}] \right) $, we proceed in two steps.

\noindent \textbf{Step 1:} Define 
$$W_n = \frac{\sum_{t=1}^{n-h} (1-\rho^2)\sigma^{-2}y_{t-1}^2}{n-h} -1~.$$
We can use \eqref{eq:ar1_model} and algebra to derive the following identity
$$ W_n = -\frac{\sigma^{-2}y_{n-h}^2}{n-h} + 2\rho\sigma^{-2} \frac{\sum_{t=1}^{n-h} u_t y_{t-1}}{n-h} + \sigma^{-2}\frac{\sum_{t=1}^{n-h} u_t^2-\sigma^2}{n-h}~.$$
This implies that
$$f_{3,n} =  \frac{(1-\rho^2)}{\sigma^2} \left( \frac{\sum_{t=1}^{n-h} u_t y_{t-1}}{n-h} \right)\left( \frac{\sum_{t=1}^{n-h} \xi_t y_{t-1}}{n-h} \right) \left( W_n + \frac{\sigma^{-2}y_{n-h}^2}{n-h}\right)~.$$
Therefore, $f_{4,n} = f_{n}-f_{1,n}-f_{2,n}-f_{3,n}$ is equal to
$$ - \frac{(1-\rho^2)}{\sigma^2} \left( \frac{\sum_{t=1}^{n-h} u_t y_{t-1}}{n-h} \right)\left( \frac{\sum_{t=1}^{n-h} \xi_t y_{t-1}}{n-h} \right) \left( \left( 1 + W_n \right)^{-1} - (1-W_n) + \frac{\sigma^{-2}y_{n-h}^2}{n-h} \right)~.$$

\noindent \textbf{Step 2:} Due to Lemma \ref{lemma:algebra}, it is sufficient to show that
\begin{equation}\label{eq:aux1_fn4}
    P \left( (n-h)^{2/2}|(1+W_n)^{-1}-(1-W_{n})|>  \delta^2 \right)  \le C \delta^{-k} E[u_t^{2k}]   ~
\end{equation}
and
\begin{equation}\label{eq:aux2_fn4}
    P \left((n-h)^{2/2}\left|\frac{\sigma^{-2}y_{n-h}^2}{n-h} \right|> \delta^2 \right) \le C \delta^{-k} E[|u_t|^{k}]  ~.
\end{equation}
Note that \eqref{eq:aux1_fn4} follows by part 5 in Lemma \ref{lemma:algebra} since $P(n^{1/2}|W_n|>\delta) \le C \delta^{-k}E[u_t^{2k}]$ due to Lemma  \ref{lemma:sample_averages} and $\delta<n^{1/2}$. Finally, \eqref{eq:aux2_fn4} follows by Markov's inequality and Lemma \ref{lemma:moments_y}. As we mentioned before, the constant $C$ includes the constants $C$'s that appear in Lemmas \ref{lemma:moments_y} and \ref{lemma:sample_averages} that only depends on $a$, $h$, $k$, and $C_{\sigma}$.
\end{proof}

\begin{lemma}\label{lemma:denominator}
Suppose Assumption \ref{AV_assumptions} holds. For any fixed $h, k \in \mathbf{N}$ and $a \in (0,1)$. Define
$$ g_n =   \frac{\sum_{t=1}^{n-h} \hat{\xi}_t(h)^2 \hat{u}_t(h)^2}{V(n-h)} - 1$$
where $V = \sigma^2 \sum_{\ell=1}^h \rho^{2(h-\ell)}$, $\hat{\xi}_t(h)$ is as in \eqref{eq:la_lp} $\hat{u}_{t}(h) = y_t - \hat{\rho}_n(h) y_{t-1}$, and $\hat{\rho}_n(h)$ is as in \eqref{eq:rho_hat_h}. Then, for any $\rho \in [-1+a,1-a]$, there exists a constant $C = C(h,k,a,C_{\sigma},c_u)$ such that we can write
$$g_n = g_{1,n} + g_{2,n} + g_{3,n}~,$$
where
$$ P\left( (n-h)^{j/2} |g_{j,n}| \ge \delta^j \right) \le C \delta^{-k} \left( E[|u_t|^{k}]  + E[u_t^{2k}] + E[u_t^{4k}] \right) $$
for any $\delta < (n-h)^{1/2}$ and $j \in \{1,2,3\}$.
\end{lemma}

\begin{proof}
In what follows we use $\xi_t =  \xi_t(\rho,h) = \sum_{\ell=1}^{h} \rho^{h-\ell} u_{t + \ell}$. As we did for the case of $f_{n}$ in Lemma \ref{lemma:numerator}, we utilize the linear regression formulas to define the components of $g_{n}$ as functions of the sample average of functions of $\xi_t$, $u_t$, and $y_{t-1}$:
\begin{align*}
   V g_{1,n} &= \frac{\sum_{t=1}^{n-h} (\xi_t^2 u_t^2-V)}{n-h} \notag \\
   V g_{2,n} &=  
    \frac{\psi_4^4}{\sigma^4} \left( \frac{\sum_{t=1}^{n-h} \xi_t u_t}{n-h} \right)^2- \frac{2}{\sigma^2} \left( \frac{\sum_{t=1}^{n-h} \xi_t u_t}{n-h} \right) \left(\frac{\sum_{t=1}^{n-h}\xi_t  u_t^3}{n-h}\right)  \notag \\
   & + (1-\rho^2)g(\rho,h)^2 \left(\frac{\sum_{t=1}^{n-h} u_t y_{t-1}}{n-h}\right)^2 - \frac{2(1-\rho^2)}{\sigma^2}\left(\frac{\sum_{t=1}^{n-h} u_t y_{t-1}}{n-h}\right) \left(\frac{\sum_{t=1}^{n-h} \xi_t^2 u_t y_{t-1}}{n-h}\right) \notag \\
     &       +     (1-\rho^2)   \left( \frac{\sum_{t=1}^{n-h} \xi_t y_{t-1}}{n-h}  \right) ^2   -  \frac{2(1-\rho^2)}{\sigma^2}  \left( \frac{\sum_{t=1}^{n-h} \xi_t y_{t-1}}{n-h}  \right) \left(\frac{\sum_{t=1}^{n-h} \xi_t u_t^2   y_{t-1} }{n-h} \right) \notag \\
     V g_{3,n} &= V  g_{n}- V \left( g_{1,n}+g_{2,n} \right)~, 
\end{align*}
where $\psi_4^4 = E[u_t^4]$ and $g(\rho,h)^2 = \sum_{\ell=1}^h \rho^{2(h-\ell)}$. By construction $g_n = \sum_{j=1}^3 g_{j,n}$. Lemma \ref{lemma:sample_averages} guarantees that each sample average in $g_{1,n}$ and $g_{2,n}$ verify the conditions to use Lemma \ref{lemma:algebra}, which imply 
$$ P\left( (n-h)^{j/2} |g_{j,n}| \ge \delta^j \right) \le C \delta^{-k} \left( E[u_t^{2k}] + E[u_t^{4k}] \right) $$
for $\delta < (n-h)^{1/2}$ and $j \in \{1,2\}$, where the constant $C$ includes $C_{\sigma}$, $a$, and $c_u$ (since $\psi_4^4 \le 24 e^{c_u^2}$). In what follows we prove $P\left( (n-h)^{3/2} |g_{3,n}| \ge \delta^3 \right) \le C \delta^{-k} \left( E[u_t^{2k}] + E[u_t^{4k}] \right) $.

First, we write $Vg_{3,n} =  R_{g,1}+ R_{g,2}$, where  $R_{g,1}$ and $R_{g,2}$ are specified below.  We will prove that $ P( (n-h)^{3/2} |V^{-1} R_{g,j}| > \delta^3) \le C \delta^{-k} E[u_t^{4k}]$ for $j=1,2$. To compute $g_{3,n}$ we use equation \eqref{eq:la_lp} and the following equality
\begin{align*}
    \hat{\xi}_t(h) 
    %
    %
    = \xi_t -  (\hat{\beta}_n(h)-\beta(\rho,h)) u_t -  \hat{\eta}_n(\rho,h) y_{t-1}~,
\end{align*}
where $ \hat{\eta}_n(\rho,h) = \rho \hat{\beta}_n(h) + \hat{\gamma}_n(h)$. We also use that $\hat{u}_t(h) = y_t - \hat{\rho}_n(h) y_{t-1} = u_t - (\hat{\rho}_n(h) - \rho) y_{t-1}$. In what follows, we denote $\hat{\beta} = \hat{\beta}_n(h)$, $\hat{\rho} = \hat{\rho}_n(h)$ and $\hat{\eta} = \hat{\eta}_n(\rho,h)$, and $\sum = \sum_{t=1}^{n-h}$ to simplify the heavy notation. We obtain $R_{g,1}$ is equal to 
\begin{align*}
&2\left[  (\eta - \hat{\eta}) +\frac{(1-\rho^2)}{\sigma^2}  \left( \frac{\sum  \xi_t y_{t-1}}{n-h}  \right) \right] \left( \frac{\sum  \xi_t u_t^2   y_{t-1} }{n-h} \right) +  \frac{\sigma^4}{1-\rho^2} \left\{ (\eta-\hat{\eta})^2 - \frac{(1-\rho^2)^2}{\sigma^4}\left( \frac{\sum  \xi_t y_{t-1}}{n-h } \right)^2 \right\}\\
   & + 2\left[ (\rho - \hat{\rho}) +\frac{(1-\rho^2)}{\sigma^2}\left(\frac{\sum  u_ty_{t-1}}{n-h}\right)  \right] \left( \frac{\sum \xi_t^2u_ty_{t-1}}{n-h}\right) +  \psi_4^4\left\{ (\beta-\hat{\beta})^2 -\frac{1}{\sigma^4} \left(\frac{\sum \xi_t u_t}{n }\right)^2 \right\} \\ 
   &+  2\left[ (\beta-\hat{\beta}) + \frac{1}{\sigma^2} \left( \frac{\sum  \xi_t u_t}{n-h} \right) \right] \left(\frac{\sum \xi_t  u_t^3}{n-h} \right) +  \frac{\sigma^4g_2^2}{1-\rho^2} \left\{ (\rho-\hat{\rho})^2 - \frac{(1-\rho^2)^2}{\sigma^4} \left(\frac{\sum  u_ty_{t-1}}{n-h} \right)^2 \right\}  ~,
\end{align*}
and we obtain that $R_{g,2}$ is equal to
\begin{align*}
     &(\beta-\hat{\beta})^2 \left(\frac{\sum u_t^4}{n-h}  - \psi_4^4\right) + (\rho-\hat{\rho})^2 \left(\frac{\sum \xi_t^2 y_{t-1}^2}{n-h}- \frac{\sigma^4 g(\rho,h)^2}{1-\rho^2} \right) +    (\eta-\hat{\eta})^2 \left(\frac{\sum  u_t^2 y_{t-1}^2}{n-h} - \frac{\sigma^4}{1-\rho^2} \right) \\
    &  2 (\beta-\hat{\beta})^2(\rho - \hat{\rho}) \left(\frac{\sum u_t^3y_{t-1}}{n-h} \right) +  2(\eta-\hat{\eta})^2 (\rho - \hat{\rho}) \left(\frac{\sum u_ty_{t-1}^3}{n-h} \right) + 2 (\beta-\hat{\beta})  (\rho-\hat{\rho})^2 \left( \frac{\sum \xi_t u_ty_{t-1}^2}{n-h} \right) \\
    & + 4 (\beta-\hat{\beta})   (\rho - \hat{\rho}) \left(\frac{\sum \xi_t u_t^2 y_{t-1}}{n-h} \right) + 2  (\eta - \hat{\eta}) (\rho-\hat{\rho})^2 \left(\frac{\sum \xi_ty_{t-1}^3}{n-h}\right) + 4 (\eta - \hat{\eta}) (\rho - \hat{\rho}) \left( \frac{\sum \xi_t u_t y_{t-1}^2}{n-h} \right)\\
    & +  2 (\beta - \hat{\beta})(\eta - \hat{\eta}) \left( \frac{\sum u_t^3 y_{t-1}}{n-h} \right)  + 2 (\beta - \hat{\beta})(\eta - \hat{\eta}) (\rho-\hat{\rho})^2 \left(   \frac{\sum u_t y_{t-1}^3}{n-h}  \right)\\
    & +(\beta-\hat{\beta})^2(\rho-\hat{\rho})^2 \left( \frac{\sum u_t^2y_{t-1}^2}{n-h} - \frac{\sigma^4}{1-\rho^2} \right)  + 4 (\beta - \hat{\beta})(\eta - \hat{\eta})  (\rho - \hat{\rho})\left( \frac{\sum u_t^2y_{t-1}^2}{n-h} - \frac{\sigma^4}{1-\rho^2} \right) \\
    & + (\eta-\hat{\eta})^2 (\rho-\hat{\rho})^2 \left( \frac{\sum  y_{t-1}^4}{n-h} -  \frac{E[u_t^4]}{1-\rho^4} - 6 \frac{\rho^2 \sigma^4}{(1-\rho^2)(1-\rho^4)}   \right) + (\eta-\hat{\eta})^2 (\rho-\hat{\rho})^2 \frac{E[u_t^4]}{1-\rho^4}\\
    &  + \frac{\sigma^4}{1-\rho^2} (\beta-\hat{\beta})^2(\rho-\hat{\rho})^2 + 4 (\beta - \hat{\beta})(\eta - \hat{\eta})  (\rho - \hat{\rho})\frac{\sigma^4}{1-\rho^2} + (\eta-\hat{\eta})^2 (\rho-\hat{\rho})^2 \frac{6 \rho^2 \sigma^4}{(1-\rho^2)(1-\rho^4)} ~.
 \end{align*}
Note that $ P( (n-h)^{3/2} |V^{-1} R_{g,1}| > \delta^3) \le C \delta^{-k} \left(E[u_t^{2k}] + E[u_t^{4k}] \right)$ follows by Lemmas \ref{lemma:algebra},  \ref{lemma:sample_averages}, and \ref{lemma:replacement_estimators}, since each term between parenthesis in the definition of $R_{g,2}$ appears in Lemma \ref{lemma:sample_averages},  the terms in brackets appears in items 1-4 of Lemma \ref{lemma:replacement_estimators}, and the terms between curly brackets can be written as the product of terms like parenthesis and brackets terms. Similarly, $ P( (n-h)^{3/2} |V^{-1} R_{g,2}| > \delta^3) \le C \delta^{-k} \left(E[u_t^{2k}] + E[u_t^{4k}] \right)$ follows by Lemmas \ref{lemma:algebra},  \ref{lemma:sample_averages}, and \ref{lemma:replacement_estimators}, since each term between parenthesis in the definition of $R_{g,1}$ appears in Lemma \ref{lemma:sample_averages} or in items 5-8 of Lemma \ref{lemma:replacement_estimators}.  
\end{proof}

 \begin{lemma}\label{lemma:algebra}
    Let $\{W_{n,j}: 1 \le j \le r\}$ be a sequence of random variables. Suppose that there exist constants $c_j$ and $C$ such that 
    \begin{align*}
        P( n^{1/2} |W_{n,j}| >  c_j \delta ) \le C \delta^{-k}~,
    \end{align*}
    for $j=1,..,r$ and some $k \in \mathbf{N}$. Then, for any $r \ge 2$ and $\delta < n^{1/2}$, we have
    \begin{enumerate}
        \item $ P ( n^{1/2} |\sum_{j=1}^r W_{n,j}|  > (\sum_{j=1}^r c_j) \delta) \le r C \delta^{-k}~.$

        \item $ P ( n^{r/2} |\prod_{j=1}^r W_{n,j}|  > (\prod_{j=1}^r c_j) \delta^r ) \le r C \delta^{-k}~.$

        \item $ P (  n^{1/2} |W_{n,1} + \prod_{j=2}^r W_{n,j} |  > (c_1 + \prod_{j=2}^r c_j) \delta ) \le 2 C \delta^{-k}~.$

        \item If $c_1 n^{-1/2}\delta < 1$. Then, $P (  |W_{n,1}|  > 1-b ) \le C \delta^{-k}~$ for any $b \in (0,1-c_1 n^{-1/2}\delta)$.

        \item If $c_1 n^{-1/2}\delta < 1$. Then, for any $b \in (0,1-c_1 n^{-1/2}\delta)$, we have
        $$ P \left(  n^{3/2} | (1+W_{n,1})^{-1/2} - (1-\frac{1}{2}W_{n,1} + \frac{3}{8}W_{n,1}^2)|  > \frac{5}{16}b^{-7/2} c_1^3 \delta^3 \right) \le 4 C \delta^{-k}~$$
        and
        $$ P \left(  n^{2/2} | (1+W_{n,1})^{-1} - (1-W_{n,1}) |  > \frac{2}{1+(1-b)^3}  \delta^2 \right) \le 3 C \delta^{-k}~.$$       
    \end{enumerate}
\end{lemma}
\begin{proof}
Bonferroni's inequality and $\{  |W_{n,1}|  > 1-b \}  \subseteq \{ n^{1/2} |W_{n,1}| >  c_1 \delta \}$ for $b \le 1-c_1n^{-1/2}\delta$ imply the proof of items 1--4. To prove the first part of item 5, we use Bonferroni's inequality to conclude that the left-hand side in item 5 is lower than or equal to the sum of $P (  |W_{n,1}|  > 1-b )$ and 
$$ P \left(  n^{3/2} | (1+W_{n,1})^{-1/2} - (1-\frac{1}{2}W_{n,1} + \frac{3}{8}W_{n,1}^2)|  > \frac{5}{16}b^{-7/2} c_1^3 \delta^3, |W_{n,1}|  \le 1-b \right)~.$$
Item 4 implies that the former term is bounded by $ C \delta^{-k}$, while the latter term is lower than or equal to
$$  P \left(  n^{3/2} \frac{5}{16}b^{-7/2} |W_{n,1}|^3   > \frac{5}{16}b^{-7/2} c_1^3 \delta^3, |W_{n,1}|  \le 1-b \right)~,$$
where the left-hand side term inside the previous probability used the Taylor Polynomial error and $|W_{n,1}| \le 1-b$. By item 2, the above probability is lower than or equal to $3C\delta^{-k}$. Finally, adding the upper and lower bounds concludes the first part of item 5. The second part is analogous.
\end{proof}

\begin{lemma}\label{lemma:moments_y}
Suppose Assumption \ref{AV_assumptions} holds. For fixed $a \in (0,1)$ and $k \ge 1$. Then, for any $|\rho| \le 1-a$, there exists a constant $C = C(a,k)>0$ such that 
\begin{equation*}
    E[|y_n|^{k}] \le C\: E[|u_n|^{k}]~,~ \forall n \ge 1~.
\end{equation*}
and $ P(n^{-1/2} |y_n^s| > \delta ) \le C\: \delta^{-k} E[|u_n|^{sk}]~,~ \forall n \ge 1~,$
\end{lemma}

\begin{proof}
The proof goes by induction. For $k=1$, we have $|y_n| \le |\rho| |y_{n-1}| + |u_n|$, which implies that
$$ E[|y_n|] \le E[|u_n|] \left(1+|\rho|+...+|\rho|^{n-1} \right) \le E[|u_n|]~ a^{-1}~.$$
Therefore, the constant $C = a^{-1}$. We can also derive $ y_n^2 = \rho^2 y_{n-1}^2 + 2 \rho y_{n-1}u_n + u_n^2 $, which implies
$  E[y_n^2] \le \rho^2 E[y_{n-1}^2] + 2|\rho| E[|y_{n-1} u_n|] + E[u_n^2]$. 
Using that $E[|y_{n-1} u_n|] = E[|y_{n-1}|] E[|u_n|] \le a^{-1} E[u_n^2]$, we conclude
$ E[y_n^2] \le \rho^2 E[y_{n-1}^2] + (2 |\rho|a^{-1} + 1) E[u_n^2] $, which implies $ E[y_n^2] \le  (2 |\rho|a^{-1} + 1) E[u_n^2] \left( 1+ \rho^2 + ... + \rho^{2(n-1)} \right) \le C_1(a) E[u_n^2]$, where $C_1(a) = (2(1-a)/a+1)/a$. In this case, the constant $C = C_1(a)$.

Let us use $C_1(a)$ to construct $C_2(a)$ and so on. Suppose we have already computed $C_k(a)$. Now, let us compute $C_{k+1}(a)$. We have
$$ y_n^{2k+1} = \rho^{2k+1}y_{n-1}^{2k+1} + \sum_{\ell=1}^{2k} C_{\ell}^{2k+1} \rho^{2k+1-\ell} y_{n-1}^{2k+1-\ell} u_n^{\ell} + u_n^{2k+1} ~.$$
By triangular inequality and similar arguments as before, we obtain
$$ E[y_n^{2k+1}] \le |\rho|^{2k+1}E[|y_{n-1}|^{2k+1}] + \sum_{\ell=1}^{2k} C_{\ell}^{2k+1} |\rho|^{2k+1-\ell} E[|y_{n-1}|^{2k+1-\ell}]E[|u_n^{\ell}|] + E[|u_n|^{2k+1}]~, $$
and by the inductive hypothesis, we know $E[|y_{n-1}|^{2k+1-\ell}]E[|u_n^{\ell}|] \le C_{k-\ell/2} E[|u_n|^{2k+1}]$, where we used that $E[|X|^a]E[|X|^b] \le E[|X|^{a+b}]$. Thus, the inductive hypothesis implies that
$$ E[y_n^{2k+1}] \le |\rho|^{2k+1}E[|y_{n-1}|^{2k+1}] + E[|u_n|^{2k+1}] \left(\sum_{\ell=1}^{2k} C_{\ell}^{2k+1} |\rho|^{2k+1-\ell} C_{k-\ell/2} + 1 \right)~, $$
in a similar way as in the initial case, we conclude. Note that the final constant $C$ only involves $a$ and $k$. The other case is analogous.
\end{proof}

\begin{lemma}\label{lemma:sample_averages}
    Suppose Assumption \ref{AV_assumptions} holds. For a given $h, k \in \mathbf{N}$ and $a \in (0,1)$. Then, for any $|\rho|\le 1-a$ and $h \in \mathbf{N}$, there exist a constant $C = C(a,k,r,s,h,C_{\sigma})>0$ such that
    \begin{enumerate}
        \item $ P(~(n-h)^{1/2} | (n-h)^{-1} \sum_{t=1}^{n-h} u_t^r y_{t-1}^s - m_{r,s} | > \delta ) \le C \delta^{-k} E\left[|u_t|^{(r+s)k} \right] $

        \item $ P(~(n-h)^{1/2} | (n-h)^{-1} \sum_{t=1}^{n-h} \xi_t u_t^r y_{t-1}^s | > \delta ) \le C \delta^{-k} E\left[|u_t|^{(1+r+s)k} \right] $ 

        \item $ P(~(n-h)^{1/2} | (n-h)^{-1} \sum_{t=1}^{n-h} \xi_t^2 u_t^2 - V | > 3\delta ) \le C \delta^{-k} E[u_t^{4k}]$    

       \item $ P(~(n-h)^{1/2} | (n-h)^{-1} \sum_{t=1}^{n-h} \xi_t^2 u_t y_{t-1}| > 3\delta ) \le C \delta^{-k} E[u_t^{4k}] $

  
       \item $ P(~(n-h)^{1/2} | (n-h)^{-1} \sum_{t=1}^{n-h} \xi_t^2 y_{t-1}^2 - \sigma^4 g(\rho,h)^2 (1-\rho^2)^{-1} | > 5\delta ) \le C \delta^{-k}E[u_t^{4k}] $ 
        
    \end{enumerate}
         for any $\delta>0$ and $n > h$, where $\xi_t = \xi_t(\rho,h) = \sum_{\ell=1}^h \rho^{h-\ell} u_{t+\ell}$, $g(\rho,h) = \left(\sum_{\ell=1}^h \rho^{2(h-\ell)}  \right)^{1/2} $, and
         $m_{r,s} = E\left[ u_t^r \left(\sum_{j\ge 1} \rho^{j-1} u_{t-j} \right)^s\right]$.
  \end{lemma} 

  \begin{proof}
    Define the filtration $\mathcal{F}_j = \sigma(u_t: t\le j)$. In what follows, we use Markov's inequality, Lemmas \ref{lemma:algebra}, \ref{lemma:martingales}, and \ref{lemma:moments_y}. The constant $C$ will replace other constants and will only depend on $a$, $k$, $r$, $s$, $h$, $C_\sigma$, and the constants that appear in Lemmas \ref{lemma:martingales} and \ref{lemma:moments_y}. 

    \noindent \textbf{Item 1:} We prove this item by induction on $s$. First, consider $s=0$ and $r \ge 0$. Note that $\{ (u_t^r - E[ u_t^r], \mathcal{F}_t): 1 \le t \le n-h\}$ define a martingale difference sequence. Therefore, Markov's inequality and Lemma \ref{lemma:martingales} imply  $ P(~(n-h)^{1/2} | (n-h)^{-1} \sum_{t=1}^{n-h} u_t^r - m_{r,0} | > \delta ) \le C \delta^{-k} E\left[|u_t|^{rk} \right] $, since $m_{r,0} = E[ u_t^r]$. Let us suppose that item 1 holds for any $(r,s)$ such that $r \ge 0 $ and $s \le s_0$ (this is a strong inductive hypothesis). Next, let us prove item 1 for $(r,s_0+1)$. We write
    $$(n-h)^{-1} \sum_{t=1}^{n-h} u_t^r y_{t-1}^{s_0+1} - m_{r,s_0+1} = I_1 + I_2~,$$
    where $I_1 = (n-h)^{-1} \sum_{t=1}^{n-h} \left(u_t^r-m_{r,0} \right) y_{t-1}^{s_0+1}$ and $        I_2 = (n-h)^{-1} \sum_{t=1}^{n-h}  m_{r,0} \left( y_{t-1}^{s_0+1} -m_{0,s_0+1}\right)$.
    Note that $\{ ( \left(u_t^r-m_{r,0} \right) y_{t-1}^{s_0+1}, \mathcal{F}_t): 1 \le t \le n-h \}$ define a martingale difference sequence; therefore, we conclude that $P((n-h)^{1/2}|I_1| > \delta) \le C \delta^{-k} E[|u_t|^{k(r+s_0+1)}] $ using Markov's inequality and Lemmas  \ref{lemma:martingales}  and \ref{lemma:moments_y}. Now, let us write
    $$ y_t^{s_0+1} = (\rho y_{t-1} + u_t)^{s_0+1} = \rho^{s_0+1} y_{t-1}^{s_0+1} + \sum_{j=1}^{s_0+1} \binom{s_0+1}{j}\rho^{s_0+1-j}  u_t^{j}  y_{t-1}^{s_0+1-j}  ~,$$
    which implies the following identity
    $$ (1-\rho^{s_0+1}) (n-h)^{-1} \sum_{t=1}^{n-h} y_{t-1}^{s_0+1}  = \frac{-y_t^{s_0+1}}{n-h} + \sum_{j=1}^{s_0+1} \binom{s_0+1}{j}\rho^{s_0+1-j} \left(  (n-h)^{-1} \sum_{t=1}^{n-h}  u_t^{j}  y_{t-1}^{s_0+1-j} \right) ~. $$
    In a similar way, using $z_t = \sum_{j\ge 1} \rho^{j-1} u_{t-j}  $ instead of $y_t$, we can derive the following identity
    $$ (1-\rho^{s_0+1}) m_{0,s_0+1} = \sum_{j=1}^{s_0+1} \binom{s_0+1}{j}\rho^{s_0+1-j} m_{j,s_0+1-j}~.$$
    Using that $|m_{r,0}|^k = |E[u_t^r]|^k \le E[|u_t|^{rk}]$, the previous two identities, the inductive hypothesis to $(n-h)^{-1} \sum_{t=1}^{n-h}  u_t^{j}  y_{t-1}^{s_0+1-j} - m_{j,s_0+1-j}$ for $j = 1,\ldots,s_0+1$, and Lemmas \ref{lemma:moments_y} and \ref{lemma:algebra}, we conclude that $P((n-h)^{1/2}|I_2| > \delta) \le C \delta^{-k} E[|u_t|^{k(r+s_0+1)}] $, which completes the proof due to Lemma \ref{lemma:algebra}.

    \noindent \textbf{Item 2:} Consider the following derivation for a sequence of random variable $f_{t} \in \mathcal{F}_{t}$:
    \begin{align*}
        \sum_{t=1}^{n-h} \xi_{t} f_t &= \sum_{t=1}^{n-h} \sum_{\ell=1}^{h} \rho^{h-\ell} u_{t+\ell} ~ f_{t} 
         = \sum_{t=1}^{n-h} \sum_{j=t+1}^{t+h} \rho^{t+h-j} u_{j}~ f_{t}
         = \sum_{j=1}^{n}  u_{j} b_{n,j}~,
    \end{align*}
    where $b_{n,j} =  \sum_{t=j-h}^{j-1} \rho^{t+h-j}~ f_{t}~ I\{ 1 \le t \le n-h \}$. Note that $\{u_j b_{n,j} \in \mathcal{F}_j: 1 \le j \le n-h\}$ defines a martingale difference sequence and 
    \begin{align*}
       E[|u_j b_{n,j}|^k] &\le h^{k-1}\sum_{t=j-h}^{j-1}  |\rho|^{(t+h-j)k}~ E[|u_j|^k] E[|f_{t}|^k] I\{ 1 \le t \le n-h \} \\ 
       %
       %
       &\le C_{k,h,a} E[|u_t|^k] \max_{1 \le t \le n-h} E[|f_{t}|^k]~,
    \end{align*}
    where $ C_{k,h,a}  = h^{k-1}  \sum_{\ell=1}^{h}  (1-a)^{(h-\ell)k}$. Now we take $f_t = u_t^ry_{t-1}^s$ and use $E[|u_t^ry_{t-1}^s|^k] \le C E[|u_t|^{rk}]E[|u_t|^{sk}]$ due to Lemma \ref{lemma:moments_y}. Finally, Jensen's inequality implies $E[|u_t|^k]E[|f_{t}|^k] \le E[|u_t|^{(1+s+r)k}]$. We conclude due to Lemma \ref{lemma:martingales}.

\noindent \textbf{Item 3--4:} We can write $ (n-h)^{-1} \sum_{t=1}^{n-h} \xi_t^2 u_t^2 - V$ and $(n-h)^{-1} \sum_{t=1}^{n-h} \xi_t^2 u_t y_{t-1}$ as the sum three martingale difference sequences. For item 3, we use the same decomposition used to prove \eqref{eq:aux1_lemmaA1} in the proof of Lemma \ref{MO-PM:lemmaA1}. For item 4, a similar decomposition is possible. For both items, we can conclude as in the proof of item 2 by using Bonferroni's inequality and Lemma \ref{lemma:martingales}; therefore, the details are omitted.

\noindent \textbf{Item 5:} Let us write $\sum_{t=1}^{n-h} \xi_t^2 y_{t-1}^2 - \sigma^4 g(\rho,h)^2 (1-\rho^2)^{-1}$ as the sum of three terms:
\begin{align*}
    \sum_{j=1}^n (u_{j}^2-\sigma^2) b_{n,j} + \sum_{j=1}^{n} u_{j}  d_{n,j} +  g(\rho,h)^2\sigma^4(1-\rho^2)^{-1} \sum_{j=1}^{n-h} (\sigma^{-2}(1-\rho^2) y_{t-1}^2 - 1)~,
\end{align*} 
where $b_{n,j}  = \sum_{t=j-h}^{j-1} \rho^{2(h+t-j)} y_{t-1}^2 I\{ 1 \le t \le n-h \}$ and
\begin{align*}
    d_{n,j} =  \sum_{t=j-h}^{j-1} \sum_{\ell_2=1}^{j-t-1}  u_{t+\ell_2}\rho^{2h-j+t-\ell_2} y_{t-1}^2 I\{ 1 \le t \le n-h\}~.
\end{align*}

Note that $(u_{j}^2-\sigma^2) b_{n,j}$ and $u_{j}  d_{n,j}$ define two martingale difference sequences with respect to $\mathcal{F}_{j-1}$. By Lemma \ref{lemma:martingales} and similar derivations as in the proof of item 2, we obtain
$$ I_1  = P(~| (n-h)^{-1/2} \sum_{j=1}^n (u_{j}^2-\sigma^2) b_{n,j} | > \delta) \le \delta^{-k} C E[u_t^{4k}]$$
and
$$ I_2  = P(~| (n-h)^{-1/2} \sum_{j=1}^{n} u_{j}  d_{n,j} | > \delta) \le \delta^{-k}  C  E[u_t^{4k}] ~,$$
for some constant $C$ that depends on $h$, $k$, and $a$. 

By item 1 ($r=0$ and $s=2$) and $m_{0,2} = \sigma^2(1-\rho^2)^{-1}$, it follows that
$$ I_3  = P(~| (n-h)^{-1/2} \frac{ g(\rho,h)^2\sigma^4}{1-\rho^2} \sum_{j=1}^{n-h} (\sigma^{-2}(1-\rho^2) y_{t-1}^2 - 1)| > \delta) \le \delta^{-k} C  E[u_t^{4k}]~,$$
for some constant $C$ that depends on $h$, $k$, $a$, and $C_\sigma$.

Finally, Bonferroni's inequality and the previous inequalities imply
$$ P(~(n-h)^{1/2} | (n-h)^{-1} \sum_{t=1}^{n-h} \xi_t^2 y_{t-1}^2 - \sigma^4 g(\rho,h)^2 (1-\rho^2)^{-1} | > 5\delta ) \le \delta^{-k} C  E[u_t^{4k}]~,$$
where the constant $C$ absorbs all the previous constants.
\end{proof}

\begin{lemma}\label{lemma:replacement_estimators}
Suppose Assumption \ref{AV_assumptions} holds. For fixed $h,k \in \mathbf{N}$ and $a \in (0,1)$. Then, for any $|\rho|\le 1-a$, there exist a constant $C = C(a,k,h,C_{\sigma})>0$ such that
    \begin{enumerate}
        \item $P\left((n-h)^{2/2} \left| \hat{\rho}_n(h) - \rho - (1-\rho^2)\sigma^{-2} \frac{\sum_{t=1}^{n-h} u_t y_{t-1}}{n-h} \right| > \delta^2 \right) \le C \delta^{-k} E[u_t^{2k}] $

        \item $P\left( (n-h)^{2/2} \left| \hat{\beta}_n(h) - \beta(\rho,h) - \sigma^{-2}\frac{\sum_{t=1}^{n-h} \xi_t(\rho,h) u_t}{n-h} \right| > \delta^2 \right) \le C \delta^{-k} E[u_t^{2k}] $

        \item $P\left( (n-h)^{2/2} \left| \hat{\gamma}_n(h) -  \frac{(1-\rho^2)\sum_{t=1}^{n-h} \xi_t(\rho,h) y_{t-1}}{\sigma^{2}(n-h)} +  \frac{\rho \sum_{t=1}^{n-h} \xi_t(\rho,h) u_t}{\sigma^{2}(n-h)} \right| > \delta^2 \right) \le C \delta^{-k} E[u_t^{2k}]  $  

        \item $P\left( (n-h)^{2/2} \left| \hat{\eta}_n(\rho,h) - \eta(\rho,h) - \frac{(1-\rho^2) \sum_{t=1}^{n-h} \xi_t y_{t-1} }{ \sigma^2(n-h)}  \right| > \delta^2 \right) \le C   \delta^{-k} E[u_t^{2k}]$

        \item $P(n^{1/2} | \hat{\rho}_n - \rho| > \delta ) \le C \delta^{-k} E[u_t^{2k}] $

        \item $P(n^{1/2} | \hat{\beta}_n(h) - \beta(\rho,h)| > \delta ) \le C \delta^{-k} E[u_t^{2k}] $

        \item $P(n^{1/2} | \hat{\gamma}_n(h) | >  \delta ) \le C \delta^{-k} E[u_t^{2k}]  $ 

        \item $P(n^{1/2} | \hat{\eta}_n - \eta| >  \delta ) \le C \delta^{-k} E[u_t^{2k}]  $ 
    \end{enumerate}
     for any $\delta<n^{1/2}$, where $\hat{\rho}_n(h)$ is as in \eqref{eq:rho_hat_h}, $(\hat{\beta}_n(h),\hat{\gamma}_n(h))$ is as in \eqref{eq:la_lp}, and $\xi_t(\rho,h) = \sum_{\ell=1}^h \rho^{h-\ell} u_{t+\ell}$. $\hat{\eta}_n(\rho,h) = \rho \hat{\beta}_n(h) + \hat{\gamma}_n(h)$, $\eta(\rho,h) = \rho \beta(\rho,h)$.
\end{lemma}

\begin{proof}
To prove item 1, we first use the definition of $\hat{\rho}_n(h)$,
$$ \hat{\rho}_n(h) - \rho  = \frac{ (n-h)^{-1}\sum_{t=1}^{n-h} u_t y_{t-1}}{(n-h)^{-1}\sum_{t=1}^{n-h} y_{t-1}^2} = \frac{ (1-\rho^2)\sum u_t y_{t-1}}{ \sigma^2 (n-h)} (1+W_n)^{-1}~, $$
where 
$ W_n = (1-\rho^2)\sigma^{-2} (n-h)^{-1}\sum_{t=1}^{n-h} y_{t-1}^2 -1.$
Using this notation, we have
$$ \hat{\rho}_n(h) - \rho -  \frac{ (1-\rho^2)\sum u_t y_{t-1}}{ \sigma^2 (n-h)}  = \frac{ (1-\rho^2)\sum u_t y_{t-1}}{ \sigma^2 (n-h)} \left(  (1+W_n)^{-1} - 1 \right)~.$$
Since $P(n^{1/2} | (n-h)^{-1} \sum_{t=1}^{n-h} u_ty_{t-1} | > \delta ) \le C \delta^{-k} E[u_t^{2k}] $ holds by  Lemma \ref{lemma:sample_averages}, it is sufficient to show that $P(n^{1/2}|(1+W_n)^{-1} - 1| > \delta ) \le C \delta^{-k} E[u_t^{2k}]$ due to Lemma \ref{lemma:algebra}. To prove the last inequality we use $P(n^{1/2}|W_n| > \delta) \le C \delta^{-k} E[u_t^{2k}]$ (which holds by  Lemma \ref{lemma:sample_averages}) and part 5 in Lemma \ref{lemma:algebra}.

The proof of items 2–3 follows from the same arguments as before; therefore, the details are omitted. Finally, the proof of item 4 follows by the results of items 2 and 3, the definition of $\hat{\eta}_n(\rho,h)$ and $\eta(\rho,h)$, and Bonferroni's inequality. Items 5-8 are implied by items 1-4, Bonferroni's inequality, and Lemma \ref{lemma:sample_averages}. 
\end{proof}

\begin{lemma}\label{lemma:martingales}
Let $\{Z_t: 1 \le t \le n\}$ be a martingale difference sequence. Then, for any $k \ge 2$, we have 
\begin{equation*}
    E\left[\left| n^{-1/2} \sum_{t=1}^{n} Z_t \right|^k \right] \le d_k \beta_{n,k}~,
\end{equation*}
where $\beta_{n,k} = n^{-1} \sum_{t=1}^n E[|Z_t|^k]$ and $d_k = (8(k-1) \max\{1,2^{k-3} \})^k$.
\end{lemma}
\begin{proof}
See \cite{Dharmadhikari:Fabian:Jogdeo1968}, where this lemma is the main theorem. 
\end{proof}  


\subsection{Proof of the Lemma B.5}\label{sec:proof_lemma_b5}

\begin{proof}
    For item 1, for any fixed $\epsilon>0$, there exist $N_0 = N_0(\epsilon)$ such that the next inclusion
    $$\{\left|\hat{\rho}_n \right| > 1-a/2\} \subseteq \{ n^{1/2}| \hat{\rho}_n - {\rho} | > n^{1/2-\epsilon} \} \cup \{|{\rho}| > 1-a \}~,$$
    holds for any $n \ge N_0$. Since $|{\rho}| \le 1-a$, we conclude
    $$ P \left( \left|\hat{\rho}_n \right| > 1-a/2  \right) \le  P \left( n^{1/2}| \hat{\rho}_n - {\rho} | > n^{1/2-\epsilon}  \right) \le C_1 n^{-1-\epsilon}   E[u_t^{2k_1}] ~,$$
    for $k_1 \ge 2(1+\epsilon)/(1-2\epsilon)$, where the last inequality follows from Lemma \ref{lemma:replacement_estimators} in Appendix \ref{sec:appendix_proof_lemma_b4} by taking $\delta = n^{1/2-\epsilon}$. This proves item 1 since $C_1   E[u_t^{2k_1}] \le C = C(c_u,k_1,C_1)$.

    For item 2, we use the definition of $\Tilde{u}_t$ in \eqref{eq:centered_residuals}, $\hat{u}_t = y_t - \hat{\rho}_n y_{t-1}$, where $\hat{\rho}_n$ is as in \eqref{eq:rho_hat}, and the model \eqref{eq:ar1_model} to obtain
    $$n^{-1} \sum_{t=1}^n \Tilde{u}_t^r =  n^{-1} \sum_{t=1}^n \left( \hat{u}_t - \Bar{\hat{u}} \right)^r = n^{-1} \sum_{t=1}^n \left( u_t + (\rho-\hat{\rho}_n ) y_{t-1} - \Bar{\hat{u}} \right)^r~,$$
    where $\Bar{\hat{u}} = n^{-1} \sum_{t=1}^n u_t + (\rho-\hat{\rho}_n) n^{-1} \sum_{t=1}^n y_{t-1}$. Using the multinomial formula and the previous expression, we have that $n^{-1} \sum_{t=1}^n \Tilde{u}_t^r$ is equal to
    $$ n^{-1} \sum_{t=1}^n \sum_{r_1+r_2+r_3 = r} \binom{r}{r_1, r_2, r_3} ~ u_t^{r_1} \left( (\rho-\hat{\rho}_n)y_{t-1} \right)^{r_2} (-\Bar{\hat{u}})^{r_3} = I_1 + I_2 + E[u_t^r]~,$$
    where
    \begin{align*}
         I_1 &=  \sum_{r_1+r_2+r_3 = r} \binom{r}{r_1, r_2, r_3} ~ (\rho-\hat{\rho}_n)^{r_2} (-\Bar{\hat{u}})^{r_3}  \left \{  n^{-1} \sum_{t=1}^n u_t^{r_1} y_{t-1}^{r_2} - m_{r_1,r_2} \right\} \\
        I_2 &= \sum_{r_1+r_2+r_3 = r} \binom{r}{r_1, r_2, r_3} ~ (\rho-\hat{\rho}_n)^{r_2} (-\Bar{\hat{u}})^{r_3}   m_{r_1,r_2}   - E[u_t^r]I\{r_1 = r\} ~\\
        m_{r_1,r_2} &= E \left[u_t^{r_1} \left(\sum_{j \ge 1} \rho^{t-1-j} u_j\right)^{r_2}  \right]~.
    \end{align*}
    Note that Lemmas \ref{lemma:algebra}, \ref{lemma:sample_averages}, and \ref{lemma:replacement_estimators} in Appendix \ref{sec:appendix_proof_lemma_b4} imply that $P(|\rho-\hat{\rho}_n|>n^{-\epsilon}) \le C n^{-1-\epsilon}$, $P(|\Bar{\hat{u}}|>n^{-\epsilon}) \le C n^{-1-\epsilon},$ and 
    $ P( | n^{-1} \sum_{t=1}^n u_t^{r_1} y_{t-1}^{r_2} - m_{r_1,r_2}| > n^{-\epsilon} ) \le C n^{-1-\epsilon}$ 
    for some constant $C$. Therefore,  Lemma \ref{lemma:algebra} implies $P(|I_j|>n^{-\epsilon}) \le C n^{-1-\epsilon} $ for $j=1,2$. This implies that, for a fixed $r$, we have $ P(|n^{-1} \sum_{t=1}^n \Tilde{u}_t^r - E[u_t^r] |>n^{-\epsilon}) \le C n^{-1-\epsilon}$.

    For item 3, we note that item 2 implies
    $ P(|n^{-1} \sum_{t=1}^n \Tilde{u}_t^2 - E[u_t^2] |> E[u_t^2]/2 ) \le C n^{-1-\epsilon}~.$
    Therefore, we conclude item 3 by taking $\tilde{C}_\sigma = E[u_t^2]/2$.
    For item 4, we note that item 2 implies
    $ P(|n^{-1} \sum_{t=1}^n \Tilde{u}_t^{4k} - E[u_t^{4k}] |> 1 ) \le C n^{-1-\epsilon}~.$
    Then, we conclude item 4 by taking $M = E[u_t^{4k}]+1$.  
\end{proof}

\subsection{Proof of Theorem \ref{thm:edgeworth_approx_Jtilde}}\label{sec:proof_thm_b2}

\begin{proof}
The proof of this theorem has two steps.

\noindent \textbf{Step 1:} The sample average $ (n-h)^{-1/2} \sum_{t=1}^{n-h} X_t$ has a valid Edgeworth expansion up to an error $o\left( n^{-3/2} \right)$ due to the results in \cite{GotzeHipp1983}. Assumption \ref{appendix:AV_assumptions2} and the definition of $X_t$ in \eqref{eq:Xi-vec} guarantees that we can use Theorem 1.2 in  \cite{gotze1994asymptotic}, and this in turn implies that we can use the results in \cite{GotzeHipp1983} (Theorem 2.8 and Remark 2.12). %
We obtain an approximation error of $o\left( n^{-3/2} \right)$ since $E[|X_t|^{6}] < +\infty$, which holds due to Assumption \ref{appendix:AV_assumptions2}.(iii).

\noindent \textbf{Step 2:} The proof of Theorem 2 in \cite{bhattacharya1978validity} and the Edgeworth expansion for the sample average $ (n-h)^{-1/2} \sum_{t=1}^{n-h} X_t$ guarantee the existence of Edgeworth expansion for the distribution $\Tilde{J}_n$ defined in \eqref{eq:tilde_J}. Furthermore, the function $q_j(x,h,P,\rho)$ for $j=1,2,3$ is a polynomial in $x$ with coefficients that are polynomials of the moments of $X_t$ (up to order $j+2$) since the sequence $X_t$ is strictly stationary ($|\rho|<1$). In particular, the coefficients of the polynomial $q_j(x,h,P,\rho)$ for $j=1,2$ are polynomials of moments of $P$ (up to order 12) and $\rho$ since the moments of $X_t$ can be computed using the moments of $u_t$ and $\rho$. Moreover, $q_j(x,h,P,\rho) = (-1)^jq_j(-x,h,P,\rho)$ since the sequence $X_t$ is strictly stationary. 
\end{proof}

\subsection{Proof of Theorem \ref{thm:edgeworth_approx_Jtilde_bootstrap}}\label{sec:proof_thm_b3}

\begin{proof}
The proof has two steps. 

\noindent \textbf{Step 1:} Define the events $E_{n,1} = \{ |\hat{\rho}_n| \le 1- a/2\}$, $E_{n,2} = \{ n^{-1} \sum_{t=1}^n \Tilde{u}_t^2 \ge \tilde{C}_\sigma \}$, and $E_{n,3} = \{ n^{-1} \sum_{t=1}^n \Tilde{u}_t^{4k} < M \}$, where $\tilde{C}_\sigma$ and $M$ are as in Lemma \ref{lemma:high-prob-event}. Define $E_{n} = E_{n,1} \cap E_{n,2} \cap E_{n,3}$. By Lemma \ref{lemma:high-prob-event} and Assumption \ref{AV_assumptions} it follows that $P(E_n^c) \le C_2 n^{-1-\epsilon}$ for some constant $C_2$  that depends on the moments of $u_t$. Since $k >8$, it follows that, conditional on $E_n$, the empirical distribution $\hat{P}_n$ verifies part (i) and (iii) of Assumption \ref{appendix:AV_assumptions2}. It is important to mention that \cite{gotze1994asymptotic} use part (ii) of Assumption \ref{appendix:AV_assumptions2} to guarantee the dependent-data version of the Cramer condition that appears in \cite{GotzeHipp1983}; see Lemma 2.3 in \cite{gotze1994asymptotic}.

\noindent \textbf{Step 2:} Condition (iii) in Lemma 2.3 in \cite{gotze1994asymptotic} holds for the bootstrap sequence $X_{b,t}^*$ since it holds for the original sequence $X_t$, otherwise the function $F$ in \eqref{eq:def_F_appendix} verifies equation (8) in \cite{gotze1994asymptotic}. Therefore, the dependent-data version of the Cramer condition holds for the bootstrap sequence $X_{b,t}^*$. The results in  \cite{gotze1994asymptotic} implied that Edgeworth expansion exists for the sample average. Then, conditional on the event $E_n$ we can repeat the arguments presented in the proof of Theorem \ref{thm:edgeworth_approx_Jtilde}.
\end{proof}



\counterwithin{figure}{section}
\counterwithin{table}{section}
\renewcommand\thesection{E}
\renewcommand\thesubsection{E.\arabic{subsection}}
\renewcommand{\thetable}{E.\arabic{table}} 
\renewcommand{\thelemma}{E.\arabic{lemma}}
\renewcommand{\theequation}{E.\arabic{equation}}
\setcounter{equation}{0}

\section{Additional Tables}\label{sec:appendix_table}

This appendix presents the additional results of the simulations.

\begin{table}[h!]
\begin{center}
\setlength{\tabcolsep}{6.0pt} 
\renewcommand{\arraystretch}{0.95} 
\begin{tabular}{cc ccccccccc}
\hline 
\multicolumn{1}{c}{$\rho$} &  $h$ & $\text{RB}$ & $\text{RB}_{per-t}$ & $\text{RB}_{hc3}$ & $\text{WB}$ & $\text{WB}_{per-t}$ & $\text{GB}_{LR}$ & $\text{AA}$ & $\text{AA}_{hc2}$ & $\text{AA}_{hc3}$  \\
\hline
 \multicolumn{11}{c}{Design 1: Gaussian i.i.d. shocks}  \\
\hline
 0.95 &     1 &  0.35 &  0.35 &  0.35 &  0.35 &  0.35 &  0.13 &  0.33 &  0.34 &  0.35\\
  &     6 &  0.83 &  0.81 &  0.83 &  0.86 &  0.84 &  0.54 &  0.71 &  0.73 &  0.74\\
  &    12 &  1.07 &  1.03 &  1.07 &  1.12 &  1.09 &  0.81 &  0.89 &  0.91 &  0.93\\
  &    18 &  1.15 &  1.11 &  1.15 &  1.21 &  1.17 &  0.97 &  0.98 &  1.00 &  1.03\\
\hline
 1.00 &     1 &  0.35 &  0.35 &  0.35 &  0.35 &  0.35 &  0.07 &  0.33 &  0.34 &  0.35\\
  &     6 &  0.97 &  0.93 &  0.97 &  1.00 &  0.96 &  0.42 &  0.80 &  0.82 &  0.84\\
  &    12 &  1.51 &  1.41 &  1.51 &  1.57 &  1.48 &  0.77 &  1.12 &  1.15 &  1.17\\
  &    18 &  2.01 &  1.83 &  2.01 &  2.09 &  1.92 &  1.07 &  1.36 &  1.39 &  1.42\\
\hline
 \multicolumn{11}{c}{Design 2: Gaussian GARCH shocks}  \\
\hline
 0.95 &     1 &  0.44 &  0.43 &  0.44 &  0.46 &  0.45 &  0.13 &  0.41 &  0.43 &  0.44\\
  &     6 &  0.93 &  0.91 &  0.94 &  1.00 &  0.98 &  0.54 &  0.80 &  0.82 &  0.84\\
  &    12 &  1.10 &  1.06 &  1.11 &  1.19 &  1.15 &  0.79 &  0.91 &  0.94 &  0.97\\
  &    18 &  1.13 &  1.09 &  1.13 &  1.22 &  1.18 &  0.94 &  0.95 &  0.98 &  1.01\\
\hline
 1.00 &     1 &  0.44 &  0.43 &  0.44 &  0.45 &  0.45 &  0.07 &  0.41 &  0.42 &  0.44\\
  &     6 &  1.10 &  1.06 &  1.11 &  1.17 &  1.13 &  0.42 &  0.91 &  0.93 &  0.96\\
  &    12 &  1.60 &  1.50 &  1.61 &  1.73 &  1.63 &  0.77 &  1.18 &  1.22 &  1.25\\
  &    18 &  2.04 &  1.86 &  2.05 &  2.21 &  2.04 &  1.06 &  1.37 &  1.41 &  1.45\\
\hline
 \multicolumn{11}{c}{Design 3: t-student i.i.d. shocks}  \\
\hline
 0.95 &     1 &  0.33 &  0.33 &  0.34 &  0.33 &  0.33 &  0.13 &  0.31 &  0.32 &  0.33\\
  &     6 &  0.81 &  0.79 &  0.82 &  0.84 &  0.82 &  0.54 &  0.68 &  0.71 &  0.73\\
  &    12 &  1.05 &  1.02 &  1.06 &  1.10 &  1.07 &  0.80 &  0.86 &  0.89 &  0.93\\
  &    18 &  1.14 &  1.10 &  1.15 &  1.19 &  1.16 &  0.97 &  0.94 &  0.98 &  1.02\\
\hline
 1.00 &     1 &  0.33 &  0.33 &  0.34 &  0.34 &  0.33 &  0.07 &  0.31 &  0.32 &  0.33\\
  &     6 &  0.94 &  0.91 &  0.95 &  0.97 &  0.94 &  0.42 &  0.77 &  0.79 &  0.82\\
  &    12 &  1.49 &  1.39 &  1.50 &  1.54 &  1.45 &  0.77 &  1.07 &  1.11 &  1.16\\
  &    18 &  1.96 &  1.79 &  1.98 &  2.03 &  1.87 &  1.08 &  1.30 &  1.35 &  1.41\\
\hline
 \multicolumn{11}{c}{Design 4: mix-gaussian GARCH shocks}  \\
\hline
 0.95 &     1 &  0.46 &  0.45 &  0.46 &  0.46 &  0.45 &  0.13 &  0.42 &  0.43 &  0.44\\
  &     6 &  0.89 &  0.87 &  0.90 &  0.96 &  0.94 &  0.55 &  0.77 &  0.79 &  0.82\\
  &    12 &  1.01 &  0.98 &  1.02 &  1.11 &  1.07 &  0.78 &  0.86 &  0.88 &  0.91\\
  &    18 &  1.02 &  1.00 &  1.03 &  1.12 &  1.08 &  0.88 &  0.89 &  0.91 &  0.94\\
\hline
 1.00 &     1 &  0.46 &  0.45 &  0.46 &  0.46 &  0.46 &  0.08 &  0.42 &  0.43 &  0.44\\
  &     6 &  1.06 &  1.01 &  1.07 &  1.12 &  1.09 &  0.45 &  0.87 &  0.90 &  0.92\\
  &    12 &  1.50 &  1.40 &  1.51 &  1.62 &  1.53 &  0.79 &  1.11 &  1.14 &  1.18\\
  &    18 &  1.88 &  1.71 &  1.89 &  2.06 &  1.90 &  1.08 &  1.28 &  1.31 &  1.35\\
\hline
\end{tabular}
 \caption{ \small Median length of confidence intervals for $\beta(\rho,h)$ with a nominal level of 90\% and $n = 95$. 5,000 simulations and 1,000 bootstrap iterations.}\vspace{-0.5cm}\label{table1a}
\end{center} 
\end{table}
  
  \newpage
\begin{table}[h!]
\begin{center}
\setlength{\tabcolsep}{6.0pt} 
\renewcommand{\arraystretch}{0.95} 
\begin{tabular}{cc ccccccccc}
\hline 
\multicolumn{1}{c}{$\rho$} &  $h$ & $\text{RB}$ & $\text{RB}_{per-t}$ & $\text{RB}_{hc3}$ & $\text{WB}$ & $\text{WB}_{per-t}$ & $\text{GB}_{LR}$ & $\text{AA}$ & $\text{AA}_{hc2}$ & $\text{AA}_{hc3}$  \\
\hline
 \multicolumn{11}{c}{Design 1: Gaussian i.i.d. shocks}  \\
\hline
 0.95 &     1 & 80.32 & 77.86 & 80.36 & 79.80 & 77.42 & 33.00 & 80.14 & 79.98 & 79.98\\
  &     6 & 92.02 & 88.44 & 92.00 & 91.92 & 89.02 & 85.68 & 91.88 & 91.92 & 91.84\\
  &    12 & 91.82 & 90.14 & 91.74 & 91.94 & 90.56 & 88.80 & 91.24 & 91.28 & 91.24\\
  &    18 & 91.90 & 90.36 & 91.98 & 91.72 & 90.22 & 89.30 & 91.38 & 91.40 & 91.42\\
\hline
 1.00 &     1 & 79.04 & 76.14 & 79.12 & 79.50 & 76.40 & 15.28 & 79.32 & 79.34 & 79.24\\
  &     6 & 92.06 & 87.66 & 92.22 & 91.82 & 88.64 & 75.14 & 91.88 & 92.02 & 91.94\\
  &    12 & 92.56 & 89.86 & 92.70 & 92.78 & 90.12 & 84.96 & 93.06 & 93.00 & 93.06\\
  &    18 & 92.06 & 90.72 & 92.08 & 92.04 & 90.50 & 87.54 & 92.08 & 92.14 & 92.18\\
\hline
 \multicolumn{11}{c}{Design 2: Gaussian GARCH shocks}  \\
\hline
 0.95 &     1 & 84.48 & 82.16 & 84.62 & 84.64 & 82.88 & 41.62 & 84.26 & 84.46 & 84.62\\
  &     6 & 92.30 & 89.74 & 92.36 & 92.44 & 89.18 & 87.80 & 92.22 & 92.20 & 92.20\\
  &    12 & 91.94 & 90.44 & 91.90 & 92.14 & 90.28 & 89.90 & 91.72 & 91.70 & 91.64\\
  &    18 & 91.76 & 90.70 & 91.76 & 91.66 & 90.52 & 90.30 & 91.50 & 91.56 & 91.58\\
\hline
 1.00 &     1 & 84.30 & 81.48 & 84.46 & 84.36 & 82.22 & 16.94 & 84.34 & 84.42 & 84.38\\
  &     6 & 92.52 & 89.44 & 92.64 & 92.70 & 89.50 & 75.92 & 92.62 & 92.66 & 92.72\\
  &    12 & 92.78 & 90.28 & 92.86 & 92.40 & 90.06 & 85.76 & 93.02 & 92.86 & 92.94\\
  &    18 & 92.14 & 90.78 & 92.12 & 92.20 & 90.32 & 88.18 & 92.28 & 92.32 & 92.28\\
\hline
 \multicolumn{11}{c}{Design 3: t-student i.i.d. shocks}  \\
\hline
 0.95 &     1 & 78.36 & 74.78 & 78.54 & 77.46 & 74.68 & 32.14 & 78.34 & 78.02 & 78.32\\
  &     6 & 91.80 & 88.08 & 91.80 & 92.08 & 88.34 & 86.58 & 91.72 & 91.60 & 91.54\\
  &    12 & 91.84 & 90.40 & 92.00 & 91.74 & 90.08 & 88.90 & 92.08 & 92.06 & 92.06\\
  &    18 & 91.74 & 89.86 & 91.50 & 91.70 & 89.82 & 89.52 & 91.32 & 91.32 & 91.44\\
\hline
 1.00 &     1 & 77.48 & 73.44 & 77.98 & 76.18 & 73.88 & 14.62 & 77.30 & 77.50 & 77.68\\
  &     6 & 92.00 & 88.28 & 91.78 & 92.06 & 88.08 & 73.00 & 92.02 & 91.78 & 91.72\\
  &    12 & 92.90 & 89.62 & 92.80 & 92.80 & 89.62 & 83.60 & 92.82 & 92.82 & 92.66\\
  &    18 & 92.16 & 90.26 & 92.20 & 92.36 & 90.60 & 86.60 & 92.30 & 92.42 & 92.40\\
\hline
 \multicolumn{11}{c}{Design 4: mix-gaussian GARCH shocks}  \\
\hline
 0.95 &     1 & 92.92 & 86.04 & 93.40 & 93.76 & 91.70 & 47.00 & 92.64 & 92.84 & 93.18\\
  &     6 & 93.68 & 91.62 & 93.72 & 93.64 & 92.00 & 90.54 & 93.72 & 93.78 & 93.80\\
  &    12 & 92.48 & 91.64 & 92.38 & 92.54 & 91.70 & 90.92 & 92.60 & 92.60 & 92.52\\
  &    18 & 92.00 & 91.26 & 91.82 & 91.78 & 90.76 & 90.86 & 91.70 & 91.62 & 91.74\\
\hline
 1.00 &     1 & 93.14 & 84.72 & 93.72 & 93.54 & 91.42 & 20.76 & 92.78 & 93.18 & 93.44\\
  &     6 & 94.20 & 91.86 & 94.20 & 93.96 & 92.14 & 80.00 & 94.52 & 94.58 & 94.50\\
  &    12 & 92.30 & 91.48 & 92.28 & 92.40 & 91.58 & 87.54 & 92.88 & 92.88 & 92.80\\
  &    18 & 92.26 & 92.12 & 92.26 & 92.16 & 91.64 & 89.22 & 92.26 & 92.30 & 92.20\\
\hline
\end{tabular}
 \caption{ \small Coverage probability (in \%) of (size-adjusted) confidence intervals for $\beta(\rho,h) \times 0.9$ with a nominal level of 90\% and $n = 95$. 5,000 simulations and 1,000 bootstrap iterations.}\vspace{-0.5cm}\label{table4a}
\end{center} 
\end{table}

\begin{table}[h!]
\begin{center}
\setlength{\tabcolsep}{6.0pt} 
\renewcommand{\arraystretch}{0.9} 
\begin{tabular}{cc ccccccccc}
\hline 
\multicolumn{1}{c}{$\rho$} &  $h$ & $\text{RB}$ & $\text{RB}_{per-t}$ & $\text{RB}_{hc3}$ & $\text{WB}$ & $\text{WB}_{per-t}$ &  $\text{GB}_{LR}$ & $\text{AA}$ & $\text{AA}_{hc2}$ & $\text{AA}_{hc3}$  \\
\hline
 \multicolumn{10}{c}{Design 1: Gaussian i.i.d. shocks}  \\
\hline
 0.95 &    18 & 89.44 & 87.52 & 89.46 & 89.66 & 88.74 & 90.02 & 87.20 & 87.40 & 87.68\\
  &    40 & 90.30 & 87.92 & 90.28 & 91.06 & 88.98 & 90.02 & 88.74 & 89.08 & 89.48\\
\hline
 1.00 &    18 & 88.62 & 88.78 & 88.56 & 89.14 & 89.58 & 90.66 & 82.16 & 82.70 & 83.02\\
  &    40 & 86.56 & 84.80 & 86.52 & 86.64 & 85.78 & 90.66 &  78.56 & 78.88 & 79.16\\
\hline
 \multicolumn{10}{c}{Design 2: Gaussian GARCH shocks}  \\
 \hline
 0.95 &    18 & 86.64 & 86.44 & 86.70 & 88.58 & 88.20 & 81.98 & 84.28 & 84.62 & 85.22\\
  &    40 & 89.18 & 86.84 & 89.24 & 90.36 & 87.90 & 81.98 & 87.46 & 87.88 & 88.32\\
\hline
 1.00 &    18 & 87.10 & 87.50 & 87.22 & 89.26 & 89.72 & 87.56 & 80.82 & 81.28 & 81.94\\
  &    40 & 84.10 & 82.56 & 84.10 & 85.46 & 85.16 & 87.56 & 76.54 & 76.88 & 77.30\\
\hline
 \multicolumn{10}{c}{Design 3: t-student i.i.d. shocks}  \\
  \hline
 0.95 &    18 & 89.06 & 88.04 & 89.06 & 89.82 & 88.78 & 89.96 & 86.04 & 86.68 & 87.38\\
  &    40 & 89.84 & 87.26 & 89.94 & 90.66 & 88.70 & 89.96 & 88.04 & 88.72 & 89.46\\
\hline
 1.00 &    18 & 89.36 & 88.96 & 89.54 & 89.98 & 89.42 & 90.92 & 82.36 & 83.00 & 83.64\\
  &    40 & 85.96 & 84.90 & 85.82 & 86.52 & 86.24 & 90.92 & 77.90 & 78.54 & 79.28\\
\hline
 \multicolumn{10}{c}{Design 4: mix-gaussian GARCH shocks} \\
 \hline
 0.95 &    18 & 83.00 & 86.32 & 83.10 & 84.50 & 87.50 & 83.20 & 81.22 & 81.54 & 82.00\\
  &    40 & 87.20 & 86.18 & 87.24 & 88.64 & 87.56 & 83.20 & 86.00 & 86.50 & 87.02\\
\hline
 1.00 &    18 & 84.06 & 88.68 & 84.22 & 86.04 & 90.42 & 86.84 & 76.98 & 77.34 & 77.80\\
  &    40 & 81.24 & 84.06 & 81.20 & 82.76 & 85.90 & 86.84 & 73.44 & 73.98 & 74.36\\
\hline
\end{tabular}
 \caption{\small Coverage probability (in \%) of confidence intervals for $\beta(\rho,h)$ with a nominal level of 90\% and $n = 240$. 5,000 simulations and 1,000 bootstrap iterations.}\ \label{table3a}
\end{center} 
\end{table}

 \newpage
\subsection{Monte-Carlo Simulations for a VAR model}\label{appendix:simulations-var}

We consider a Bivariate VAR(4) model as in \cite{MO-PM2020_SM}:
\[
y_{1,t} = \rho y_{1,t-1} + u_{1,t}, \quad \left(1 - \frac{1}{2} L \right)^{4} y_{2,t} = \frac{1}{2} y_{1,t-1} + u_{2,t}, 
\begin{pmatrix} u_{1,t} \\ u_{2,t} \end{pmatrix} \overset{i.i.d.}{\sim} N \left( 0, \begin{pmatrix} 1 & 0.3 \\ 0.3 & 1 \end{pmatrix} \right)~.
\]
We construct confidence intervals for the reduced-form impulse response of $y_{2,t}$ with respect to the shock $u_{1,t}$, that is, we set $\nu = (1,0)$ and $i=2$ according to the notation of Section \ref{sec:lp-residual-bootstrap_VAR}. The confidence intervals that we use are listed below:

\begin{enumerate}
\item \textbf{RB:} confidence interval as in \eqref{eq:bootstrap_CI_VAR} based on the LP-residual bootstrap in Section \ref{sec:lp-residual-bootstrap_VAR}.

\item \textbf{$\text{RB}_{per-t}$:}  equal-tailed percentile-t confidence interval based on the LP-residual bootstrap described in Section \ref{sec:lp-residual-bootstrap_VAR} and similar to the one described in Remark \ref{rem:percentile-t_bootstrap}.

\item \textbf{WB:}  confidence interval as in \eqref{eq:bootstrap_CI_VAR} but using $c_n^{wb,*}(h,1-\alpha)$ instead of $c_n^*(h,1-\alpha)$, where $c_n^{wb,*}(h,1-\alpha)$ is based on the LP-wild bootstrap, that is, the LP-residual bootstrap described in Section \ref{sec:lp-residual-bootstrap_VAR} with a different step 2 as in Remark \ref{rem:wild-bootstrap}.

\item \textbf{$\text{WB}_{per-t}$:}  equal-tailed percentile-t confidence interval based on the LP-wild bootstrap.

\item \textbf{AA:} standard confidence interval as in \cite{MO-PM2020}.

\end{enumerate}

Table \ref{table_var4} presents the results of the simulation in terms of coverage probability and median length. It reports qualitatively similar results as the ones presented in Section \ref{sec:simulation_study}. The confidence interval \textbf{RB} performs better than \textbf{AA} and is similar to \textbf{WB} for all horizons.

\begin{table}[h!]
\begin{center}
\setlength{\tabcolsep}{6.0pt} 
\renewcommand{\arraystretch}{1.02} 
\centering
\begin{tabular}{cccccccccccc}
\toprule
&  \multicolumn{5}{c}{Coverage} & \multicolumn{5}{c}{Median length} \\
\cmidrule(lr){2-6} \cmidrule(lr){7-11}
$h$  & RB & RB$_{per-t}$ & WB & WB$_{per-t}$ & AA & RB & RB$_{per-t}$ & WB & WB$_{per-t}$ & AA \\
\midrule
\multicolumn{11}{c}{$\rho = 0.00$} \\
1 & 89.700 & 89.660 & 90.160 & 90.040 & 89.280 & 0.232 & 0.232 & 0.235 & 0.235 & 0.228 \\
6 & 90.440 & 89.180 & 91.200 & 90.060 & 88.640 & 1.450 & 1.440 & 1.487 & 1.476 & 1.376 \\
12 & 89.320 & 88.780 & 90.320 & 89.860 & 87.480 & 1.572 & 1.565 & 1.614 & 1.605 & 1.490 \\
36 & 89.940 & 89.500 & 90.760 & 90.560 & 88.160 & 1.657 & 1.653 & 1.705 & 1.701 & 1.577 \\
60 & 89.300 & 89.220 & 90.220 & 90.200 & 87.200 & 1.771 & 1.767 & 1.826 & 1.824 & 1.671 \\
\midrule
\multicolumn{11}{c}{$\rho = 0.50$} \\
1 & 89.820 & 89.800 & 90.160 & 90.060 & 89.380 & 0.232 & 0.232 & 0.235 & 0.235 & 0.228 \\
6 & 90.340 & 88.920 & 91.040 & 89.860 & 88.040 & 1.705 & 1.680 & 1.744 & 1.720 & 1.590 \\
12 & 89.080 & 87.820 & 90.040 & 88.640 & 86.660 & 1.968 & 1.942 & 2.017 & 1.997 & 1.830 \\
36 & 89.880 & 89.460 & 90.840 & 90.320 & 88.160 & 2.036 & 2.023 & 2.095 & 2.082 & 1.934 \\
60 & 89.060 & 89.060 & 90.120 & 90.220 & 86.500 & 2.178 & 2.169 & 2.251 & 2.238 & 2.048 \\
\midrule
\multicolumn{11}{c}{$\rho = 0.95$} \\
1 & 89.740 & 89.740 & 90.120 & 90.080 & 89.320 & 0.232 & 0.232 & 0.236 & 0.236 & 0.229 \\
6 & 89.640 & 88.880 & 90.360 & 89.500 & 85.960 & 2.346 & 2.243 & 2.393 & 2.289 & 2.084 \\
12 & 87.960 & 87.860 & 88.500 & 88.660 & 81.500 & 4.824 & 4.374 & 4.935 & 4.482 & 3.786 \\
36 & 82.740 & 80.660 & 83.740 & 82.020 & 77.380 & 6.207 & 5.640 & 6.421 & 5.836 & 5.040 \\
60 & 88.300 & 87.300 & 89.380 & 88.500 & 82.460 & 6.003 & 5.644 & 6.197 & 5.826 & 5.185 \\
\midrule
\multicolumn{11}{c}{$\rho = 1.00$} \\
1 & 89.900 & 89.780 & 90.200 & 90.100 & 89.420 & 0.233 & 0.233 & 0.236 & 0.236 & 0.229 \\
6 & 87.420 & 88.520 & 87.940 & 89.180 & 82.060 & 2.483 & 2.333 & 2.534 & 2.381 & 2.152 \\
12 & 82.920 & 87.020 & 83.960 & 87.860 & 71.540 & 5.950 & 5.142 & 6.067 & 5.266 & 4.306 \\
36 & 70.060 & 73.460 & 70.800 & 74.580 & 46.320 & 14.159 & 10.863 & 14.556 & 11.188 & 7.528 \\
60 & 53.060 & 57.300 & 54.340 & 58.560 & 28.860 & 13.932 & 10.433 & 14.365 & 10.767 & 7.904 \\
\bottomrule
\end{tabular} 
 \caption{\small Coverage probability (in \%) and median length of confidence intervals for $\beta(\rho,h)$ with a nominal level of 90\% and $n = 240$. 5,000 simulations and 2,000 bootstrap iterations. }\vspace{-1.0cm}\label{table_var4}
\end{center} 
\end{table}

\bibliographystyle{ecta.bst}

\bibliography{references.bib}